\documentclass[rmp,aps,reprint,longbibliography]{revtex4-1}
\pdfoutput=1

\usepackage{amsmath, amssymb, amsthm}
\usepackage[mathscr]{euscript}     
\usepackage{qiprmp} 
\usepackage[continuous]{thmenvironments} 
\usepackage{layoutcommands} 

\usepackage[shortcuts]{extdash} 
\usepackage{enumitem} 
\setlist{nolistsep}
\usepackage[caption=false,listofformat=subsimple,labelformat=simple]{subfig} 
\usepackage{tikz} 
\usetikzlibrary{shapes,fit,decorations.pathmorphing,arrows.meta} 

\usepackage[colorlinks=true]{hyperref} 
\usepackage[open,numbered]{bookmark} 
\usepackage{breakurl} 
\usepackage{hyperlinksrmp} 

\begin{document}

\title{Security in quantum cryptography}

\author{Christopher Portmann}
\email{Electronic address: chportma@ethz.ch}
\affiliation{Department of Computer Science, ETH Zurich, 8092 Zurich,
  Switzerland}

\author{Renato Renner}
\email{Electronic address: renner@ethz.ch}

\affiliation{Institute for Theoretical Physics, ETH Zurich, 8093
  Zurich, Switzerland}

\date{\today}

\begin{abstract}
  Quantum cryptography exploits principles of quantum physics for the secure processing of information.   A prominent example is secure communication, i.e., the task of transmitting confidential messages from one location to another. The cryptographic requirement here is that the transmitted messages remain inaccessible to anyone other than the designated recipients, even if the communication channel is untrusted. In classical cryptography, this can usually only be guaranteed under computational hardness assumptions, e.g., that factoring large integers is infeasible. In contrast, the security of quantum cryptography relies entirely on the laws of quantum mechanics.  Here we review this physical notion of security, focusing on quantum key distribution and secure communication.
\end{abstract}

\maketitle

\tableofcontents

\section{Security from physical principles}
\label{sec:intro}

Communication theory is concerned with the task of making information
available to different parties. The sender of a message $x$ wants
$x$ to become accessible to a designated set of recipients. In
\emph{cryptography}, one adds to this a somewhat opposite requirement
\--- that of restricting the availability of information. The sender
of~$x$ also wishes to have a guarantee that $x$ remains inaccessible
to \emph{adversaries}, i.e., parties other than the intended recipients.
The term \emph{security} refers to this additional guarantee.

Testing whether a communication protocol works correctly is easy. It
suffices to compare the message~$x$ sent with the received
one. Testing security, however, is more subtle. To ensure that an
adversary cannot read~$x$, one needs to exclude all physically
possible eavesdropping strategies.  Since there are infinitely many such strategies it is not possible, at least not by direct experiments, to prove  that a cryptographic scheme is secure \---
although a successful hacking experiment would of course show the
opposite.

But the situation is not as hopeless as this sounds. Security can be established indirectly, provided that one is ready to make certain assumptions about the capabilities of the adversaries. Clearly, the weaker these assumptions are, the more confident we can be that they apply to any realistic adversary, and hence that a cryptographic scheme based on them is actually secure. 

The security of most cryptographic schemes used today relies on computational hardness assumptions. They correspond to constraints on the adversaries' computational resources. For example, it is assumed that adversaries do not have the capacity to factor large integers~\cite{RSA78}. This is a relatively strong assumption, justified merely by the belief that the currently known algorithms for factoring cannot be substantially improved in the foreseeable future \--- and that  quantum computers, powerful enough to run Shor's (efficient) factoring algorithm~\cite{Shor97}, cannot be built. Cryptographic schemes whose security is based on assumptions of this type are commonly termed \emph{computationally secure}.

In contrast to this, the main assumption that enters quantum cryptography is that adversaries are subject to the laws of quantum mechanics.\footnote{To prove security, one usually also requires that adversaries cannot manipulate the local devices (such as senders and receivers) of the legitimate parties. But, remarkably, this (seemingly necessary) requirement can be weakened \--- this is the topic of device-independent cryptography, which we discuss in~\secref{sec:open.di}.}  This assumption completely substitutes computational hardness assumptions, i.e., security holds even if the adversaries can use unbounded computational resources to process their information.\footnote{Though one may naturally also consider computationally secure quantum cryptography, which we do in, e.g., \secref{sec:computational} and \secref{sec:open.computational}.} To distinguish this from computational security, the resulting security is sometimes termed \emph{information-theoretic}, reflecting the fact that it can be defined in terms of purely information-theoretic concepts~\cite{Shannon49}.

\subsection{Completeness of quantum theory} \label{sec:completeness}

The assumption that adversaries are subject to the laws of quantum mechanics appears to be  rather straightforward to justify. Indeed, quantum mechanics is one of our best tested physical theories. As of yet, no  experiment has been able to detect deviations from its predictions. Of particular relevance for cryptography are non-classical features of quantum mechanics, such as entanglement between remote subsystems, which have been tested by Bell experiments~\cite{FreedmanClauser,Aspect81,Aspect82,Tittel,Weihs,Rowe,Giustina13,Christensen,Hensen,Giustina15,Shalm,Rosenfeld}. However, the assumption  that enters quantum cryptography not only concerns the \emph{correctness} of quantum mechanics (as one may naively think), but also its \emph{completeness}. This is an important point, and we therefore devote this entire subsection to it. 

Quantum mechanics is a \emph{non-deterministic} theory in the following sense. Even if we know, for instance, the polarisation direction $\psi$ of a photon to arbitrary accuracy, the theory will not in general allow us to predict with certainty the outcome $z$ of a polarisation measurement of, say, the vertical versus the horizontal direction. The statement that we can obtain from quantum mechanics may even be completely uninformative. For example, if the polarisation $\psi$ before the measurement was diagonal, the theory merely tells us that a measurement of the vertical versus the horizontal direction will yield both possible outcomes~$z$ with equal probability. 

It is  conceivable that non-determinism is just a limitation of current quantum theory, rather than a fundamental property of nature. This would mean that there could exist another theory that gives better predictions. In the example above, it could be that the photon, in addition to its polarisation state~$\psi$, has certain not yet discovered properties $\lambda$ on which the measurement outcome~$z$ depends.  A theory that takes into account $\lambda$ could then yield more informative predictions for~$z$ than quantum mechanics. If this were the case then quantum mechanics could not be considered a complete theory. 

Quantum cryptography is built on the use of physical systems, such as photons, as information carriers. The incompleteness of quantum mechanics would hence imply that the theory does not give a full account of all  information contained in these systems. This would have severe consequences for security claims. For example, a  cryptographic scheme for transmitting a confidential message~$x$ may be claimed to be secure on the grounds that the quantum state $\psi$ of the information carriers  gathered by an adversary  is independent of~$x$. Nonetheless, it could still be that the adversary's information carriers have an extra property, $\lambda$, which is not described by quantum theory and hence not included in $\psi$. The independence of $\psi$ from~$x$ is then not sufficient to guarantee that the adversary cannot learn the secret message.

A possible way around this problem is to simply \emph{assume} that no adversary can access properties of physical systems, like $\lambda$ in the above example, which are not captured by their quantum state~$\psi$. But such an assumption seems to be similarly difficult to justify as, for instance, the non-existence of an efficient factoring algorithm. The fact that we have not yet been able to discover $\lambda$ does not mean that it does not exist (or that it cannot be discovered). 

Fortunately, the problem can be resolved in a more fundamental manner. The solution is based on a long sequence of work dating back to~\textcite{Born26,EPR35},  where the question regarding the completeness of quantum mechanics was raised. The central insight resulting from this work was that the set of possible theories that could improve the predictions of quantum mechanics is highly constrained. For example, no such theory can yield deterministic predictions, based on additional parameters~$\lambda$, unless it is non-local\footnote{This concept will be briefly discussed in the context of device-independent cryptography in \secref{sec:alternative.di}.}~\cite{Bell64} and contextual \cite{Bell66,KocSpe67}. More recently, it has been shown that no theory can improve the predictions of quantum mechanics unless it violates the requirement that measurement settings can be chosen freely, i.e., independently of other parameters of the theory~\cite{CR11}.\footnote{More precisely, according to~\textcite{BellFree}, variables are ``free'' if they ``have implications only in their future light cones.'' In other words, they are uncorrelated to anything outside their causal future. This notion has sometimes also been called ``free will''~\cite{Conway2006}.} The completeness of quantum mechanics is hence implied by the assumption that physics does not prevent us from making free choices \--- an assumption that appears to be unavoidable in cryptography anyway~\cite{ER14}. 

\subsection{Correctness of quantum-theoretic description}

In the previous section we have seen that the security of quantum cryptography crucially relies on the completeness of quantum mechanics, but that the latter can be derived from the requirement that one can make free choices.  It is of course still necessary to assume that quantum mechanics is correct, in the sense that it accurately describes the  hardware used for implementing a cryptographic protocol. But since  quantum mechanics consists of a set of different rules,  we should  be more specific about what this correctness assumption really means. 

Quantum cryptographic protocols are usually described within the framework of quantum information theory~\cite{nielsen2010quantum}, which provides the necessary formalism to talk about information carriers and operations on them. Any information carrier is modelled as a quantum system $S$ with an associated Hilbert space $\mathcal{H}_S$, and the information encoded in~$S$ corresponds to its state. In the case of ``classical'' information,  the different values~$x$ of a variable with range~$\mathcal{X}$ are represented by different elements from a fixed orthonormal basis  $\{\ket{x}\}_{x \in \mathcal{X}}$ of $\mathcal{H}_S$. If the marginal state of  a system~$S$ has the form  $\rho_{S} = \sum_{x} p_{x} \proj{x}$  this means that $S$ carries the value~$x$ with probability~$p_x$. 
Any processing of information (including, for instance, a measurement)  corresponds to a change of the state of the involved information carriers, and is represented mathematically by a  trace-preserving, completely positive map.\footnote{We refer to standard textbooks in quantum information theory, such as~\textcite{nielsen2010quantum}, for a description of these concepts. An argument that justifies their use in the context of cryptography can be found in \textcite{RenesRenner2020}.}

The modelling of real-world implementations in terms of these rather abstract information-theoretic notions is a highly non-trivial task. To illustrate this, take for example an optical scheme for quantum key distribution, where information is communicated by an encoding in the polarisation of individual photons. This suggests a description where each photon sent over the optical channel is regarded as an individual quantum system. However, photons are just excitations of the electromagnetic field and thus a priori not objects with their own identity. (That is, they are indistinguishable.) A solution to this problem could be that one ``labels'' the photons by the time at which they are sent out, i.e., photons sent at  different times are regarded as different quantum systems, $S$. But there could be more than one photon emitted at a particular time, and these different photons could or could not have the same frequency. One may now choose to take this into account by modelling the photon number and their frequency as internal degrees of freedoms of the system $S$. Or, one could choose the frequency to be an additional system label, so that photons with different frequencies are regarded as different systems. 

This example shows that the translation of an actual physical setup into the language of quantum information theory is prone to mistakes and certainly not unique. Nonetheless, it is critical for security \--- if done incorrectly, the security statements, which are derived within quantum information theory, are vacuous. Particular care must be taken to ensure that no information carriers that are present in an implementation are omitted.  A realistic photon source may, for instance, sometimes emit two instead of only one photon whose polarisation encodes the same value, and this second photon may be accessible to an eavesdropper  (see Section~\secref{sec:attacks:hacking}). This possibility must therefore be included in the quantum information-theoretic description of that photon source.  If it was not, it would represent a \emph{side channel} to the adversary that is not accounted for by the security proof. Side channels may also occur in other components, such as photon detectors, for instance. We refer to the review of~\cite{SBCDLP09} for a general discussion of these practical aspects of quantum cryptography. 


\subsection{Overview of this review}

In this review we will focus on the information-theoretic layer of security proofs, i.e., we will presume that we have a correct quantum information-theoretic description of the cryptographic hardware. The existence of such a description is indeed a standard assumption made for security proofs and usually termed \emph{device-dependence}. It contrasts \emph{device-independent} cryptography, where this assumption is considerably relaxed (see~\secref{sec:alternative.di} for a brief discussion). 

We will start in the next section by introducing general concepts from cryptography. From then on we will largely focus on Quantum Key Distribution (QKD), which currently is the most widespread application of  quantum cryptography. It is also an excellent concrete example to discuss security definitions, the underlying assumptions, as well as proof techniques. Towards the end of this review we will explain how these notions apply to cryptographic tasks other than key distribution.




\section{Cryptographic security definitions}
\label{sec:ac}

\subsection{Real-world ideal-world paradigm}
\label{sec:ac.realideal}

Cryptographic schemes are not usually perfectly secure. Rather, they provide a certain level of security that is quantified by one or several parameters. Take for instance an encryption scheme. It could only be called perfectly secure if we had a guarantee that an adversary can learn absolutely nothing about the encrypted message \--- something that turns out to be impossible to achieve in practice. Still, we have encryption schemes, e.g., those built in quantum cryptography, that are ``almost perfectly secure''. So we need a quantitative definition that makes precise what this means.

Devising sensible quantitative definitions can be challenging.  Consider, for instance, a protocol that encrypts quantum information contained in a $d$-dimensional register~$A$ by applying a unitary $U_k$ that depends on a uniformly chosen key $k \in \cK$. It has been proposed, e.g., in \textcite{AS04,HLSW04,DN06}, that the security of such a scheme may be defined by requiring that, for any state $\rho_A$,
\begin{equation} \label{eq:encryption.naive}
  \frac{1}{2}\trnorm{\frac{1}{|\cK|} \sum_{k \in \cK} U_k \rho_A \hconj{U}_k -
    \tau_A} \leq \eps,
\end{equation}
where $\eps \geq 0$ is the security paramter, $\tau_A = \frac{1}{d} I$ is the fully mixed state, and $\trnorm{\cdot}$ denotes the trace norm or
Schatten $1$-norm. The definition has been justified by the argument that an adversary who does not know the key~$k$ cannot distinguish the encryption of the state from $\tau_A$
(except with advantage\footnote{See the text around~\eqref{eq:adv1} for a definition of the notion of a \emph{distinguishing advantage}.} $\eps$). However, it has later been realized that this does not hide
the information in the $A$ system from an adversary who may hold a purification
$R$ of the information~$A$ \cite{ABW09}. To take this into account, one should instead require that,
for any $\rho_{AR}$,
\begin{equation*} 
  \frac{1}{2}\trnorm{\frac{1}{|\cK|} \sum_{k \in \cK} \left(U_k \otimes I_R
    \right) \rho_{AR} \left(\hconj{U}_k \otimes I_R\right) - \tau_A
    \otimes \rho_R} \leq \eps,
\end{equation*}
where $\rho_R$ is the reduced density operator of $\rho_{AR}$. Note
that, crucially, this criterion is not implied
by~\eqnref{eq:encryption.naive} above \cite{Wat18}.

As another example, early works on Quantum Key Distribution (QKD),
e.g., \textcite{May96,BBBMR00,SP00}, measured the secrecy of a secret
key in terms of the \emph{accessible information}\footnote{This
  captures the information a player may obtain by measuring her
  quantum state, and is formally defined in \eqnref{eq:localqkd} in
  \secref{sec:qkd.other.ai}, see also \textcite{nielsen2010quantum}.}
between the key and all information that may be accessible to an
adversary. In security proofs it was then shown that this value is
small, apparently implying that the key is almost perfectly
secret. Later one has realised, however, that the accessible
information is not a good measure for secrecy: even if this measure is
exponentially small in the key size, an adversary may for example be
able to infer the second part of the key upon seeing the first
part~\cite{KRBM07}. This makes the key unusable for many applications,
such as encryption, as described in detail in
\secref{sec:qkd.other.ai}.

Problems analogous to the ones outlined above are well known in classical cryptography. They were addressed independently by
\textcite{PW00,PW01} and \textcite{Can01}, building on a series of
earlier works~\cite{GMW86,Bea92,MR92,Can00}, with a security paradigm
that we will refer to as the ``real\-/world ideal\-/world'' paradigm.
Its gist lies in quantifying how well some \emph{real} protocol for a
cryptographic task can be distinguished from some \emph{ideal} system
that fulfils the task perfectly.

As a simple (non-cryptographic) example, we consider channel coding, i.e., the task of constructing a noiseless channel form a noisy one. Suppose that Alice
and Bob only have access to a noisy channel, drawn in
\figref{fig:security.channel.noisy}. In order to send a 
message, Alice will encode it in a larger message space that has
redundancies. Upon reception, Bob will decode it, using the redundancies to correct errors~\cite{nielsen2010quantum}. Putting together the encoder, the noisy channel, and the 
decoder, as illustrated in
\figref{fig:security.channel.construction}, gives a new channel. Ideally, this constructed channel, which we call the \emph{real world}, should behave like a perfect, noiseless channel,
\figref{fig:security.channel.noiseless}, which we therefore call the \emph{ideal world}. To quantify how well we achieved this goal, we measure how close the real world is to the ideal world.

\begin{figure}[tb]
  \centering
  \subfloat[Noiseless channel][\label{fig:security.channel.noiseless}A
  noiseless channel that perfectly delivers the message from Alice to
  Bob.]{
\begin{tikzpicture}[
sArrow/.style={->,>=stealth,thick},
thinResource/.style={draw,thick,minimum width=1.618*2cm,minimum height=1cm}]

\small

\node[thinResource] (channel) at (0,0) {};
\node (alice) at (-3,0) {Alice};
\node (bob) at (3,0) {Bob};

\draw[sArrow] (alice) to node[pos=.1,auto] {$\rho$} node[pos=.9,auto] {$\rho$}  (bob);

\end{tikzpicture}
}

\vspace{6pt}

\subfloat[Noisy channel][\label{fig:security.channel.noisy}A noisy
channel that alters the message sent from Alice to Bob.]{
\begin{tikzpicture}[
sLine/.style={-,thick},
nLine/.style={-,thick,decorate,decoration={snake,amplitude=.4mm,segment length=2mm,post length=1mm}},
sArrow/.style={->,>=stealth,thick},
thinResource/.style={draw,thick,minimum width=1.618*2cm,minimum height=1cm}]

\small

\node[thinResource] (channel) at (0,0) {};
\node (alice) at (-3,0) {Alice};
\node (bob) at (3,0) {Bob};

\draw[sLine] (alice) to node[pos=.5,auto] {$\rho$} (channel.west);
\draw[nLine] (channel.west) to (channel.east);
\draw[sArrow] (channel.east) to node[pos=.5,auto] {$\tilde{\rho}$} (bob);

\end{tikzpicture}
}

\vspace{6pt}

\subfloat[Noiseless channel
construction][\label{fig:security.channel.construction}Alice encodes
  her message into a larger space, and Bob decodes upon reception.]{
\begin{tikzpicture}[
sLine/.style={-,thick},
nLine/.style={-,thick,decorate,decoration={snake,amplitude=.4mm,segment length=2mm,post length=1mm}},
sArrow/.style={->,>=stealth,thick},
thinResource/.style={draw,thick,minimum width=1.618*2cm,minimum height=1cm},
converter/.style={draw,thick,rounded corners,minimum width=.8cm,minimum height=.8cm}]

\small

\node[thinResource] (channel) at (0,0) {};
\node[converter] (A) at (-2.5,0) {Enc};
\node[converter] (B) at (2.5,0) {Dec};
\node (alice) at (-3.4,0) {};
\node (bob) at (3.4,0) {};

\draw[sArrow] (alice.center) to node[pos=.5,auto] {$\rho$} (A);
\draw[sLine] (A) to node[pos=.5,auto] {$\sigma$} (channel.west);
\draw[nLine] (channel.west) to (channel.east);
\draw[sArrow] (channel.east) to node[pos=.5,auto] {$\tilde{\sigma}$} (B);
\draw[sArrow] (B) to node[pos=.5,auto] {$\tilde{\rho}$} (bob.center);

\end{tikzpicture}
}
\caption[Channels for error
correcting]{\label{fig:security.channels}In these diagrams each box
  represents a reactive system which produces an output upon receiving
  an input. Boxes with rounded corners are local operations performed
  by a party [e.g., encoding and decoding in
  \subref{fig:security.channel.construction}]. The rectangular box is
  a (possibly noisy) channel form Alice to Bob, which upon receiving
  Alice's input produces an output at Bob's end of the channel. The
  arrows represent quantum states being transmitted from one system to
  another.}
\end{figure}
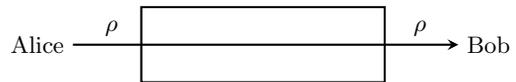
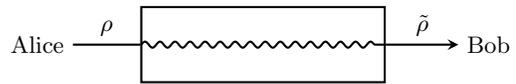
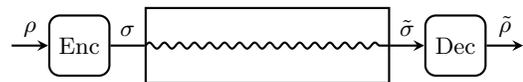

For this, we consider a hypothetical game in which a
\emph{distinguisher} has \emph{black box access} to an unknown system
as shown in \figref{fig:distinguisher}. The unknown system is,
depending on a random bit~$B$, either the real world ($B=0$) or the
ideal world ($B=1$). The term \emph{black box access} means that the
distinguisher is not provided with a description of the system, and in
particular has no direct access to the bit~$B$, but otherwise can
interact arbitrarily with it. In the case of our noiseless channel
construction problem, the distinguisher can generate any joint state
$\rho_{AR}$ it desires, input the $A$ part into the channel, and then
measure the joint state of the channel output and its purification
$R$.  The distinguisher is then asked to guess whether it interacts
with the real world ($B=0$) or with the ideal one ($B=1$). Let $D$ be
a random variable denoting the distinguisher's guess. The
\emph{distinguishing advantage} of the distinguisher is then defined as the
difference between the probabilities that it guesses correctly and
erroneously, namely
\begin{equation} \label{eq:adv1} \left|\Pr[D = 0|B=0] - \Pr[D = 0|B=1]
  \right|.\end{equation} The distinguishing advantage for a class of
distinguishers (e.g., computationally bounded or unbounded
distinguishers), is then defined as the supremum of \eqnref{eq:adv1}
over all distinguishers in this set. For example, in the case of channel coding, the distinguishing advantage for unbounded
distinguishers corresponds to the diamond\-/norm between the
channels~\cite{Wat18}. A protocol is considered secure if the
distinguishing advantage is small \--- or, more accurately, the (level
of) security of a protocol is parametrized by this advantage and the
corresponding class of distinguishers.

\begin{figure}[tb]
\begin{tikzpicture}[
sArrow/.style={->,>=stealth,thick},
largeResource/.style={draw,thick,minimum width=1.618*2cm,minimum height=2cm},
lrnode/.style={minimum width=1.36*2cm,minimum height=.2cm},
llrnode/.style={minimum width=.2cm,minimum height=1.5cm},
tlrnode/.style={minimum width=.2cm,minimum height=.5cm}]

\small

\def\v{.6}
\def\u{2.27} 
\def\w{-1.9} 

\node[lrnode] (r1) at (0,\v) {};
\node[lrnode] (r2) at (0,0) {};
\node[lrnode] (r3) at (0,-\v) {};
\node[llrnode] (rr1) at (-\v,0) {};
\node[llrnode] (rr2) at (0,0) {};
\node[llrnode] (rr3) at (\v,0) {};
\node[largeResource] (R) at (0,0) {\footnotesize Unknown system};

\draw[thick] (-1.618-1.15,1) -- ++(.75,0) -- ++(0,-2.4)  --
++(1.618*2+.8,0)  -- ++(0,2.4) -- ++(.75,0) -- ++(0,-3.4) --
++(-1.618*2-2.3,0) -- cycle;

\node at (0,\w) {\footnotesize Distinguisher};
\node[tlrnode] (dd1) at (-\v,\w) {}; 
\node[tlrnode] (dd2) at (0,\w) {}; 
\node[tlrnode] (dd3) at (\v,\w) {};
\node[inner sep=0] (d1) at (-\u,\v) {};
\node[inner sep=0] (d2) at (-\u,0) {};
\node[inner sep=0] (d3) at (-\u,-\v) {};
\node[inner sep=0] (d4) at (\u,\v) {};
\node[inner sep=0] (d5) at (\u,0) {};
\node[inner sep=0] (d6) at (\u,-\v) {};
\node[inner sep=0] (d0) at (0,\w-.5-.7) {};

\draw[sArrow] (rr1) to (dd1);
\draw[sArrow] (dd2) to (rr2);
\draw[sArrow] (rr3) to (dd3);
\draw[sArrow] (d1) to (r1);
\draw[sArrow] (r2) to (d2);
\draw[sArrow] (d3) to (r3);
\draw[sArrow] (r1) to (d4);
\draw[sArrow] (d5) to (r2);
\draw[sArrow] (r3) to (d6);
\draw[sArrow] (dd2) to node[auto,pos=.6] {\footnotesize $D \in \{0,1\}$} (d0);






\end{tikzpicture}

\caption[Distinguishing systems]{\label{fig:distinguisher}In the real-world ideal-world paradigm, security is defined in terms of indistinguishability. A distinguisher  has black-box access to a system that, depending on an unknown bit $B$, is either the real cryptographic protocol ($B=0$) or an ideal functionality ($B=1$). After interacting with
  the system, the distinguisher outputs a guess $D$ for $B$. The real protocol is considered as secure as the ideal system if the success probability $\Pr[D=B]$ of the best possible distinguisher is close to that of a random guess, i.e., to~$\frac{1}{2}$.}
\end{figure}
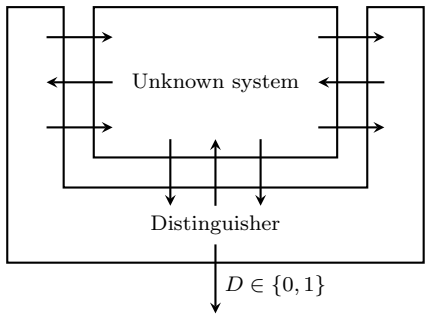

In its essence, the real-world ideal-world paradigm avoids defining
\emph{security}; instead, it provides a simple description, the ideal
world, of what should happen in the real world. In the example of
channel coding, the real world might involve a complex noise model as
well as encoding and decoding operations, whereas the ideal world is
just an identity map.  When evaluating whether such a security
statement is appropriate, one asks whether the ideal world captures
what we need, or whether one should design a different ideal world.

A crucial property of the real-world ideal-world paradigm is that the resulting notion of security is \emph{composable}.  This means that the security of a protocol is guaranteed even if it is composed with other protocols to form a larger cryptographic system. In fact, to ensure composability,  the notion of distinguishability has to be chosen appropriately. Specifically, the distinguisher must have access jointly to
all information available normally to the honest parties as well as to the adversary. The role of the distinguisher is hence to capture ``the rest of the world'', everything that exists around the system of interest. In particular, the distinguisher may choose the inputs to the protocol (that might come from a previously run protocol), receive its outputs (that could be used in a subsequent protocol), and  simultaneously take the role of the adversary, possibly eavesdropping on the communication channels and tampering with messages. 


\subsection{The Abstract Cryptography framework}
\label{sec:ac.ac}

In modern cryptography, security claims (and their proofs) are usually
phrased within a theoretical framework. The framework does not only
provide a common language, but also ensures composability, in the
sense described above. That is, security claims that hold for
individual components can be turned into a security claim for a more
complex cryptographic scheme built from them. The first frameworks to
achieve this for classical cryptography are the \emph{reactive
  simulatability framework} of~\textcite{PW00,PW01} and the
\emph{universal composability framework} of~\textcite{Can01}, which
both use the real-world ideal-world paradigm.

These frameworks have been further developed \cite{BPW04,BPW07,CDPW07,Can20} and several variations have been
proposed~\cite{MMS03,MRST06,CCKLLPS06a,CCKLLPS06b,Kus06,HS13}. The differences between them concern mostly how they describe information-processing systems, i.e., how the individual devices carry out computations and how they schedule messages when communicating (e.g., synchronously or asynchronously).  While this modelling was mostly based on classical notions of computation and communication, the frameworks have also been adapted to quantum cryptography by
\textcite{BM04} and \textcite{Unr04,Unr10}. 

In \textcite{MR11} [see also, \textcite{Mau12,MR16}] a framework was
proposed, \emph{Abstract Cryptography} (AC), that is largely
independent of the underlying modelling of the information-processing
devices, and therefore applies equally to classical and quantum
settings. We use it for the presentation here, for it enables a
self-contained description without the need to specify (unnecessary)
technical details.\footnote{The security of QKD could be equivalently
  modelled based on the work of \textcite{Unr10} with minor
  adaptations to capture finite statements instead of only
  asymptotics.}
In the following, we briefly describe the two basic paradigms on which
the framework is based, \emph{abstraction} and
\emph{constructibility}.

\paragraph{Abstraction.}
The traditional approach to defining security \--- used in all the
frameworks cited except for AC \--- can be seen as \emph{bottom\-/up}. One
first defines (at a low level) a computational model (e.g., a Turing
machine), and then proceeds by modelling how the machines communicate (e.g., by
writing to and reading from shared tapes). Next, one introduces higher-level notions such as indistinguishability. Finally, these notions are used to define security. 

In contrast, AC uses a \emph{top\-/down} approach. In order to state definitions and develop a theory, one starts from the other end, the highest possible level of abstraction. There, cryptographic systems are simply regarded as elements of a set, which can be combined to form new systems. One then proceeds down to lower levels of abstraction, introducing in each of them the minimum necessary specializations. Only on these lower levels it is modelled how exactly the cryptographic systems process information and how they communicate when they are combined (e.g., synchronously or asynchronously). The notion of indistinguishability is first defined on the highest abstraction level as an (arbitrary) metric on the set of cryptographic systems. On lower abstraction levels it can then be instantiated in different ways, e.g., to capture the distinguishing power of a computationally bounded or unbounded environment. 

Abstraction has not only the advantage that it generalises the treatment, but it usually also simplifies it, as unnecessary specificities are avoided.  It may be compared, for instance, to the use of group theory in mathematics, which is an abstraction of more special concepts such as matrix multiplication.  In a bottom\-/up approach, one would
start introducing a rule for taking the product between matrices and then, based on that rule, study the properties of the multiplication operation. In contrast to this, the
top\-/down approach taken here corresponds to first defining the (abstract) multiplication group and prove theorems already on this level. 



\paragraph{Constructibility.}
Cryptography can be regarded as a resource theory, where certain
desired resources are \emph{constructed} from a set of given
resources.\footnoteremember{fn:resourcetheory}{This view is widespread
  in cryptography and made formal by composable frameworks \--- e.g.,
  \textcite{PW00,PW01,Can01,MR11}. Resource theories have also been
  used in many different ways to capture (certain aspects of) quantum
  mechanics, see the recent reviews \textcite{SAP17,CG19}.} The
constructions are defined by protocols. For example, a QKD protocol
uses a quantum communication channel together with an authentic
channel\footnote{An authentic channel guarantees that the message
  received comes from the legitimate sender, and has not been tampered
  with or generated by an adversary.} as resources to construct the
resource of a secret key. This latter resource may then be used by
other protocols, e.g., an encryption protocol, to construct a secure
communication channel, which is again a resource. Similarly, the
authentic channel used by the QKD protocol can itself be constructed
from an insecure channel resource and short (uniform) secret key
\cite{WC81}. And given a weak secret key (i.e., not necessarily
uniform and not perfectly correlated randomness shared by the
communication partners) and 2-way insecure channels, one may construct
an almost perfect secret key (i.e., uniform and perfectly correlated
randomness) using so-called non\-/malleable
extractors~\cite{RW03,DW09,ACLV19}. Composing the authentication
protocol with the QKD protocol results in a scheme which constructs a
long secret key from a short secret key and insecure channels \--- and
composing this again with non\-/malleable extractors constructs this
long key from only a weak key and insecure channels. Part of the
resulting long secret key can be used in further rounds of
authentication and QKD to produce even more secret key. This is
illustrated in \figref{fig:construction}, and discussed in detail in
\secref{sec:smt}.

\begin{figure*}[tbp]

\begin{tikzpicture}[sArrow/.style={-{To[]},thick},
textbox/.style={draw,text width=2.6cm,text centered,minimum height=.6cm}]
\small

\def\d{1.8}
\def\e{2.5}

\node[textbox] (z2) at (1*\d,-1*\e) {Weak secret key};
\node[textbox] (z4) at (3*\d,-1*\e) {Insecure classical channels};
\node[textbox] (a1) at (0*\d,0*\e) {Insecure classical channel};
\node[textbox] (a3) at (2*\d,0*\e) {Short secret key};
\node (m3) at (2*\d,-.5*\e) {};
\node[draw,text width=1.8cm,text centered] (b2) at (1*\d,1*\e) {Authentic channel};
\node (n2) at (1*\d,.5*\e) {};
\node[textbox] (b4) at (3*\d,1*\e) {Insecure quantum channel};
\node[textbox] (c3) at (2*\d,2*\e) {Long secret key};
\node (o3) at (2*\d,1.5*\e) {};
\node[textbox] (c1) at (0*\d,2*\e) {Insecure classical channel};
\node[textbox] (c5) at (4*\d,2*\e) {Insecure classical channel};
\node[draw,text width=1.8cm,text centered] (d2) at (1*\d,3*\e) {Authentic channel};
\node[draw,text width=1.8cm,text centered] (d4) at (3*\d,3*\e) {Authentic channel};
\node (p2) at (1*\d,2.5*\e) {};
\node (p4) at (3*\d,2.5*\e) {};
\node[textbox] (e25) at (1.5*\d,4*\e) {Secure channel};
\node (q3) at (2*\d,3*\e) {};
\node (r25) at (1.5*\d,3.5*\e) {};
\node[textbox] (d6) at (5*\d,3*\e) {Insecure quantum channel};
\node[textbox] (e5) at (4*\d,4*\e) {Long secret key};
\node (r5) at (4*\d,3.5*\e) {};

\draw[thick] (z2) to (m3.center);
\draw[thick] (z4) to (m3.center);
\draw[sArrow] (m3.center) to node[auto,swap] {Non-malleable Extractors} (a3);
\draw[thick] (a1) to (n2.center);
\draw[thick] (a3) to (n2.center);
\draw[sArrow] (n2.center) to node[auto] {Authentication} (b2);
\draw[thick] (b2) to (o3.center);
\draw[thick] (b4) to (o3.center);
\draw[sArrow] (o3.center) to node[auto] {QKD} (c3);
\draw[thick] (c1) to (p2.center);
\draw[thick] (c3) to (p2.center);
\draw[sArrow] (p2.center) to node[auto] {Authentication} (d2);
\draw[thick] (c3) to (p4.center);
\draw[thick] (c5) to (p4.center);
\draw[sArrow] (p4.center) to node[auto,swap] {Authentication} (d4);
\draw[thick] (c3) to (q3.center);
\draw[thick] (d2) to (r25.center);
\draw[thick] (q3.center) to (r25.center);
\draw[sArrow] (r25.center) to node[auto] {One-Time Pad} (e25);
\draw[thick] (d4) to (r5.center);
\draw[thick] (d6) to (r5.center);
\draw[sArrow] (r5.center) to node[auto,swap] {QKD} (e5);

\end{tikzpicture}

\caption[Recursive construction of
resources]{\label{fig:construction}A constructive view of
  cryptography. A cryptographic protocol uses (weak) resources to
  construct other (stronger) resources. These resources are depicted
  in the boxes, and the arrows are protocols. Each box is a
  one-time-use resource, so the same resource appears in multiple
  boxes if different protocols require it. The long secret key
  resource in the center of the figure is split in three shorter keys,
  each of which is used by a separate protocol. The example of secure
  message transmission illustrated here is discussed in detail in
  \secref{sec:smt}.}
\end{figure*}
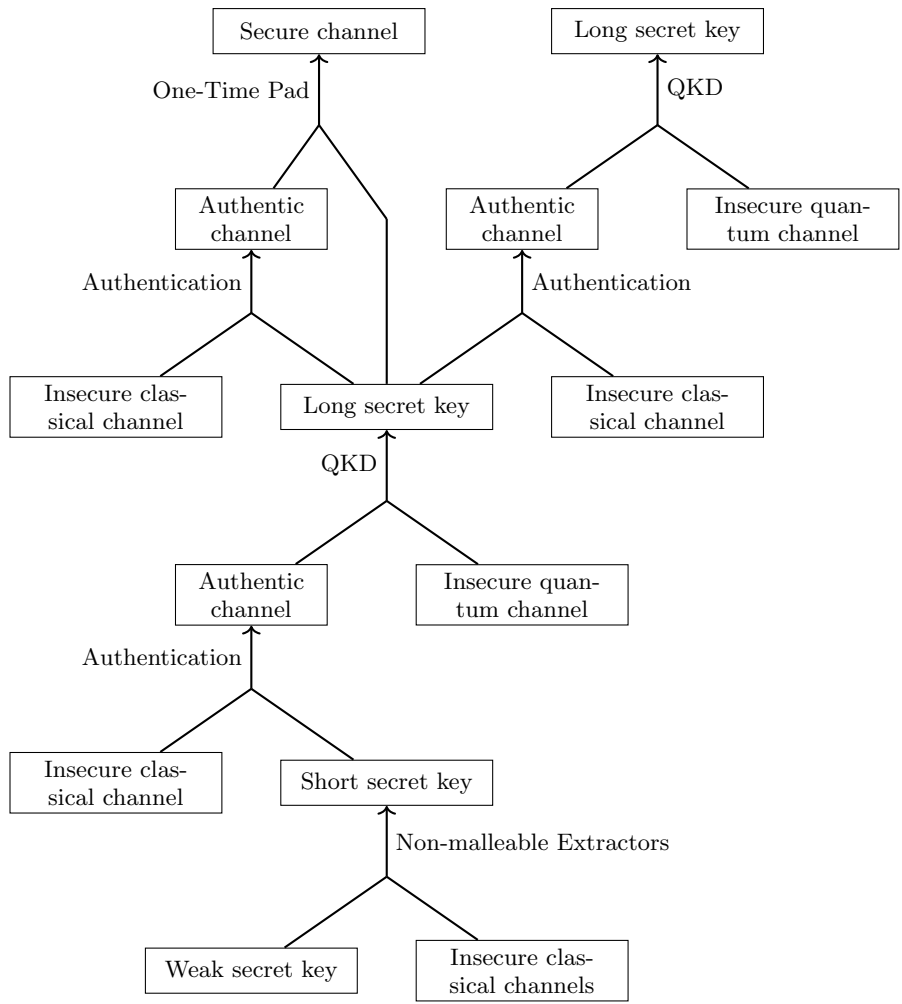

The resources used and constructed in cryptography are interactive
systems shared between players. A system that distributes secret key
or the different types of channels mentioned in the paragraph above
are examples of such resources. These are formalized on an abstract
level in \secref{sec:ac.systems}, and possible instantiations are
discussed in \secref{sec:ac.instantiating}. Static resources such as
coherent states \cite{BCP14,MZYM19} can be seen as a special case of
these.

\subsection{Example: the one-time pad}
\label{sec:ac.otp}

In this section, we describe how the notions introduced above are
employed to specify the security of a cryptographic protocol. For this
we consider a concrete example, One-Time Pad (OTP)
encryption~\cite{Vernam26}. The OTP assumes that the players have access
to an authentic channel, i.e., one which provides the receiver with
the guarantee that the messages received come from the correct sender,
but there is no guarantee about the secrecy of the messages sent on
such a channel, i.e., they may leak to Eve. The OTP also requires the
players to have access to a secret key. These two resources are drawn
as boxes with square corners in \figref{fig:otp.real}. According to
the protocol, the sender, Alice, encrypts a message~$x$ as
$y \coloneqq x \xor k$, where $k$ is the secret key, and where $\xor$
denotes the bit-wise exclusive OR operation. The ciphertext $y$ is
then sent over an authentic channel to the receiver, Bob, who decrypts
it by carrying out the operation $x = y \xor k$. At the same time, $y$
may also leak to an adversary, Eve.

\begin{figure}[tb]
  \centering
  \subfloat[Real OTP system][\label{fig:otp.real}The real OTP
    system consists of the OTP protocol $(\pi^{\otp}_A,\pi^{\otp}_B)$ together
   a secret key and authentic channel resources.]{
\begin{tikzpicture}[
sArrow/.style={->,>=stealth,thick},
thinResource/.style={draw,thick,minimum width=2cm,minimum height=1cm},
protocol/.style={draw,rounded corners,thick,minimum width=1.3cm,minimum height=2.5cm},
pnode/.style={minimum width=.9cm,minimum height=.5cm}]

\small

\def\t{3.7} 
\def\u{2.05} 
\def\v{.75}

\node[pnode] (a1) at (-\u,\v) {};
\node[pnode] (a2) at (-\u,0) {};
\node[pnode] (a3) at (-\u,-\v) {};
\node[protocol,text width=1.1cm] (a) at (-\u,0) {\footnotesize $y
  =$\\$\quad x \xor k$};
\node[yshift=-2,above right] at (a.north west) {\footnotesize
  $\pi^{\otp}_A$};
\node (alice) at (-\t,0) {Alice};

\node[pnode] (b1) at (\u,\v) {};
\node[pnode] (b2) at (\u,0) {};
\node[pnode] (b3) at (\u,-\v) {};
\node[protocol,text width=1.1cm] (b) at (\u,0) {\footnotesize $x
  =$\\$\quad y \xor k$};
\node[yshift=-2,above right] (blabel) at (b.north west) {\footnotesize $\pi^{\otp}_B$};
\node (bob) at (\t,0) {Bob};

\node[thinResource] (keyBox) at (0,\v) {};
\node[draw] (key) at (0,\v) {key};
\node[yshift=-2,above right] at (keyBox.north west) {\footnotesize Secret key};
\node[thinResource] (channel) at (0,-\v) {};
\node[yshift=-1.5,above right] at (channel.north west) {\footnotesize
  Authentic ch.};
\node (eve) at (0,-2.15) {Eve};
\node (junc) at (eve |- a3) {};

\draw[sArrow] (key) to node[auto,swap,pos=.3] {$k$} (a1);
\draw[sArrow] (key) to node[auto,pos=.3] {$k$} (b1);

\draw[sArrow] (alice) to node[auto,pos=.1] {$x$} (a2);
\draw[sArrow] (b2) to node[auto,pos=.8] {$x$} (bob);

\draw[sArrow] (a3) to node[pos=.13,auto] {$y$} node[pos=.87,auto] {$y$} (b3);
\draw[sArrow] (junc.center) to node[pos=.75,auto] {$y$} (eve);

\node[draw, gray, thick, fit=(a)(b)(blabel), inner sep=6] {};

\end{tikzpicture}}

\vspace{6pt}

\subfloat[Ideal One-Time OTP system][\label{fig:otp.ideal}The ideal
OTP system consists of the ideal secure channel and a simulator
$\sigma^{\otp}_E$.]{
\begin{tikzpicture}[
sArrow/.style={->,>=stealth,thick},
thinResource/.style={draw,thick,minimum width=1.618*2cm,minimum height=1cm},
protocolLong/.style={draw,rounded corners,thick,minimum height=1cm,minimum width=2.8cm}]

\small

\def\t{2.368} 
\def\u{-.75}
\def\v{.75}
\def\w{-1.75} 

\node[thinResource] (channel) at (0,\v) {};
\node[yshift=-1.5,above right] (chlabel)at (channel.north west) {\footnotesize
  Secure channel};
\node (alice) at (-2.6,\v) {Alice};
\node (bob) at (2.6,\v) {Bob};

\node[protocolLong] (sim) at (0,\u) {};
\node[xshift=-1.5,yshift=-2,above right] at (sim.north west) {$\sigma^{\otp}_E$};
\node[draw] (rand) at (0,\u) {\footnotesize Random string};

\draw[sArrow] (alice) to node[pos=.03,auto] {$x$}
node[pos=.96,auto] {$x$} (bob);
\draw[sArrow,dotted] (0,\v) to node[pos=.6,auto] {$|x|$} (rand);

\node (eve) at (0,\w-.4) {Eve};
\draw[sArrow] (rand) to node[pos=.75,auto] {$y$} (eve);

\node[draw, gray, thick, fit=(channel)(chlabel)(sim), inner sep=6] {};

\end{tikzpicture}}

\caption[Real and Ideal One-Time Pad systems]{\label{fig:otp}The real
  and ideal one-time pad systems. Boxes with rounded corners are local
  systems executed at Alice's, Bob's or Eve's interfaces. The
  rectangular boxes are shared resources modeling channels or shared
  keys. Arrows represent the
  transmission of messages between systems or to the environment (distinguisher).

  The real world is depicted in
  \subref{fig:otp.real}. The
   protocol consists of a part $\pi^{\otp}_A$ executed by Alice (who
   has access to the interfaces on the left hand side) and a part
    $\pi^{\otp}_B$ executed by Bob (on the right hand side). It
    takes a message $x$ at Alice's outer interface as
    well as a key~$k$, and output a ciphertext $y$ towards the
    authentic channel.  Bob's part of the protocol takes $y$ and~$k$ as
    input, and outputs the decrypted message. The channel may leak~$y$
    at Eve's interface (at the bottom).

    The ideal world is depicted in \subref{fig:otp.ideal}. The secure
    channel transmits the message perfectly from Alice's to Bob's
    interface, leaking only the message length at Eve's interface. The
    simulator $\sigma^{\otp}_E$ generates a random string $y$ of
    length $|x|$, making the real and ideal systems perfectly  indistinguishable.}
\end{figure}
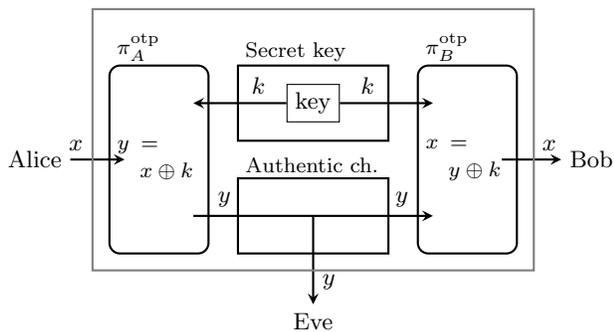
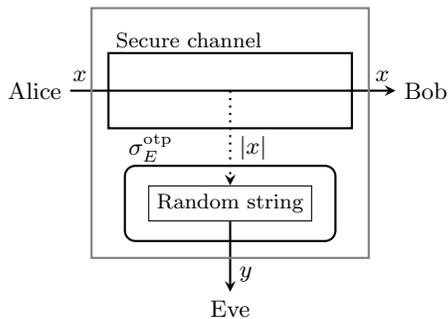

For this example, the goal of the OTP is to add confidentiality to an
authentic channel,\footnote{Alternatively, one may use the OTP with a
  completely insecure channel, and thus obtain a malleable
  confidential channel~\cite{MRT12}.} i.e., the ideal system is a secure
channel, drawn as a box with square corners in
\figref{fig:otp.ideal}. This is a channel which only leaks the message size
but no other information to Eve. It is straightforward to verify that
in the real system, provided that the key~$k$ is uniformly distributed
over bitstrings of the same length as the message~$x$, the ciphertext
$y$ is statistically independent of the message~$x$. The
ciphertext~$y$ hence does not provide Eve with any information
about~$x$, except potentially for its length~$|x|$.  It thus
constructs a secure channel from Alice to Bob. 








To make the real and ideal systems comparable, we consider an entire class of
systems, which are obtained by appending an arbitrary system, called a
\emph{simulator}, to the Eve interface of the ideal channel
resource. The systems from this class are sometimes called
\emph{relaxations} (of the ideal system). The idea is that none of
these relaxations can be more useful to Eve than the original ideal
channel, because she may always herself carry out the task of the
simulator. Security now means that the real system is
indistinguishable from at least one relaxation of the ideal system. In
our example of the OTP, such a relaxation may be obtained by a
simulator that simply generates a random string of length~$|x|$ and
outputs it at the Eve interface, as depicted in
\figref{fig:otp.ideal}.

To establish security of OTP encryption, it is therefore sufficient to
show that the real system depicted by \figref{fig:otp.real} is
indistinguishable from the relaxation of the ideal secure channel shown in~\figref{fig:otp.ideal}. That is, the two systems must behave
identically when they interact with a distinguisher.  This is indeed
the case. For both of them, if the distinguisher inputs $x$ at Alice's
interface, the same string $x$ is output at Bob's interface and a
uniformly random string of length $|x|$ is output at Eve's
interface. The two systems are thus perfectly indistinguishable \---
if the distinguisher were to take a guess for which of the two it is interacting with, it would be right with
probability exactly $1/2$. In this sense, the OTP construction is
perfectly secure.

If two systems are indistinguishable, they can be used interchangeably in any setting. For example, let some protocol $\pi'$ be proven secure if Alice and Bob are connected by a secure channel. Since the OTP constructs such a channel, it can be used in lieu of the secure channel, and composed with $\pi'$. Or equivalently, the contrapositive: if composing the OTP and $\pi'$ were to leak some vital information, which would not happen with a secure channel, a distinguisher that is either given the real or ideal system could run $\pi'$ internally and check whether this leak occurs to find out with which of the two it is interacting.





\subsection{Abstract theory of cryptographic systems}
\label{sec:ac.systems}

The previous sections introduced the concepts of resources, protocols
and simulators in an informal manner. Now, following the spirit of the
AC framework described in \secref{sec:ac.ac}), we provide an axiomatic specification of these
concepts. This will allow us to give a definition of cryptographic security, which is precise, but at the same time largely independent of
implementation details. In particular, it does not depend on the underlying computational model or the scheduling of messages exchanged between the systems. 

While this abstract approach to defining security is rather universal, we note that, when describing concrete systems and
their compositions such as those 
depicted in \figref{fig:otp}, their behaviour must of course be
specified in detail. This may be done using various
frameworks for modeling interactive (quantum) systems, e.g., the
Quantum Combs of \textcite{CDP09} or the Causal Boxes of
\textcite{PMMRT17}. This is discussed further in
\secref{sec:ac.instantiating}.

Nevertheless, the definitions that now follow refer to an abstract notion of a \emph{system}. Following the idea of abstraction motivated above, and  continuing the analogy to group theory used in \secref{sec:ac.ac}, it is sufficient to think of systems as objects on which certain operations are defined, such as their composition.  We will consider two types of systems, which we call resources and converters, and which have slightly different properties.

\paragraph{Resources.} A \emph{resource} is a system with
interfaces specified by a set $\cI$ (e.g., $\cI = \{A,B,E\}$).  
Each
interface $i \in \cI$ models how a player~$i$ can access the system
(e.g., provide inputs and read outputs). Examples of resources are a communication channel or any of the objects that appears in Fig.~\ref{fig:construction} as a box. We will sometimes use the term \emph{$\cI$-resource} to specify the interface set. Resources are equipped with
a parallel composition operator, denoted by~$\|$, that maps two $\cI$-resources to
another $\cI$-resource.

\paragraph{Converters.} A \emph{converter} is a system with two
interfaces, an \emph{inside} interface and an \emph{outside}
interface. A converter can be appended to a resource, converting it
into a new resource. For this the inside interface connects to an
interface of a resource, and the outside interface becomes the new
interface of the new resource \--- see the OTP example in
\figref{fig:otp}, where the gray boxes are new resources resulting
from composing resources and converters. We write either
$\alpha_i \aR$ or $\aR\alpha_i$ to denote the new resource with the
converter $\alpha$ connected at the interface $i$ of
$\aR$.\footnote{There is no mathematical difference between
  $\alpha_i \aR$ and $\aR\alpha_i$. It sometimes simplifies the
  notation to have the converters for some players written on the
  right of the resource and the ones for others on the left, rather
  than having all of them at the same side, hence the two notations.}
Simulators and protocols are examples of converters (see below).

Converters can be composed among themselves. There are two ways of
doing this, referred to as serial and parallel composition. These are
defined as
\begin{align*} 
(\alpha\beta)_i \aR & \coloneqq
  \alpha_i (\beta_i \aR) \\ \intertext{and}  (\alpha \| \beta)_i
  (\aR\|\aS) & \coloneqq (\alpha_i \aR) \| (\beta_i \aS), \end{align*}
  respectively. 
  
  \paragraph{Protocols.} A \emph{(cryptographic) protocol} is a family
  $\alpha = \{\alpha_i\}_i$, of converters (one for every honest
  player). A protocol can be applied to a resource $\aR$, giving a new
  resource denoted by $\alpha\aR$ or $\aR\alpha$. This resource is
  obtained by connecting each member of the family to the interface
  specified by its index.


  \paragraph{Metric.} As explained in \secref{sec:ac.realideal}, the
  distance between resources can be quantified using the notion of
  distinguishers. More generally, one may in principle
  consider any arbitrary pseudo\-/metric, $d(\cdot,\cdot)$, so that
  the following conditions hold:\footnote{If also
    $d(\aR,\aS) = 0 \implies \aR=\aS$ holds then $d$ is a metric.}
\begin{align*} \text{(identity)} & &
  d\left(\aR,\aR\right) & = 0, \\ 
  \text{(symmetry)} & & d\left(\aR,\aS\right) & =
  d\left(\aS,\aR\right), \\ 
  \text{(triangle inequality)} & & d\left(\aR,\aS\right)
  & \leq d\left(\aR,\aT\right) +
  d\left(\aT,\aS\right). 
\end{align*} Furthermore, the pseudo\-/metric must be non-increasing
under composition with resources and converters.\footnote{This only
  holds for information\-/theoretic security, which is the topic of
  most of this review.} This means that for any converter $\alpha$ and
resources $\aR,\aS,\aT$,
\begin{equation*} 
d(\alpha\aR,\alpha\aS)
  \leq d(\aR,\aS) \quad \text{and} \quad d(\aR\|\aT,\aS\|\aT) \leq
  d(\aR,\aS). \end{equation*}
In this work we often simply write $\aR \close{\eps} \aS$ instead of
$d(\aR,\aS) \leq \eps$.


\subsection{Security definition}
\label{sec:ac.security}

We are now ready to define the security of a cryptographic
protocol. We do so in \defref{def:security} in the three-player
setting, for honest Alice and Bob, and dishonest Eve \--- and
illustrate this definition in Fig.~\ref{fig:security}. Thus, in the
following, all resources have three interfaces, denoted $A$, $B$ and
$E$, and we consider honest behaviors (given by a protocol
$(\pi_A,\pi_B)$) at the $A$- and $B$\=/interfaces, but arbitrary
behavior at the $E$\=/interface. We refer to \textcite{MR11} for the
general case, when arbitrary players can be dishonest.

\begin{figure*}[tb]

\begin{tikzpicture}[scale=.8,
sArrow/.style={->,>=stealth,thick},
largeResource/.style={draw,thick,minimum width=1.618*2cm,minimum height=2cm},
lrnode/.style={minimum width=1.36*2cm,minimum height=.2cm},
llrnode/.style={minimum width=.2cm,minimum height=1.5cm},
tlrnode/.style={minimum width=.2cm,minimum height=.5cm},
thinResource/.style={draw,thick,minimum width=1.618*2cm,minimum height=1cm},
protocol/.style={draw,rounded corners,thick,minimum width=1.545cm,minimum height=2.5cm},
pnode/.style={minimum width=1cm,minimum height=.5cm},
protocolLong/.style={draw,rounded corners,thick,minimum height=1cm,minimum width=2.8cm}]

\small

\def\t{4.422} 
\def\u{2.9} 
\def\v{.6}
\def\w{1.8}

\def\a{5.2}
\def\b{8.3}
\def\tb{2.368} 
\def\vb{.4}
\def\ub{.75}
\def\wb{2.05}

\node[scale=.8,pnode] (a1) at (-\u,\v) {};
\node[scale=.8,pnode] (a2) at (-\u,0) {};
\node[scale=.8,pnode] (a3) at (-\u,-\v) {};
\node[scale=.8,protocol] (a) at (-\u,0) {};
\node[yshift=-2,above right] at (a.north west) {\footnotesize
  $\pi_A$};
\node (alice1) at (-\t,\v) {};
\node (alice2) at (-\t,0) {};
\node (alice3) at (-\t,-\v) {};

\node[scale=.8,pnode] (b1) at (\u,\v) {};
\node[scale=.8,pnode] (b2) at (\u,0) {};
\node[scale=.8,pnode] (b3) at (\u,-\v) {};
\node[scale=.8,protocol] (b) at (\u,0) {};
\node[yshift=-2,above right] at (b.north west) {\footnotesize $\pi_B$};
\node (bob1) at (\t,\v) {};
\node (bob2) at (\t,0) {};
\node (bob3) at (\t,-\v) {};

\node[scale=.8,lrnode] (r1) at (0,\v) {};
\node[scale=.8,lrnode] (r2) at (0,0) {};
\node[scale=.8,lrnode] (r3) at (0,-\v) {};
\node[scale=.8,llrnode] (rr1) at (-\v,0) {};
\node[scale=.8,llrnode] (rr2) at (0,0) {};
\node[scale=.8,llrnode] (rr3) at (\v,0) {};
\node[scale=.8,largeResource] (R) at (0,0) {};
\node[yshift=-2,above right] at (R.north west) {\footnotesize $\aR$};

\node (eve1) at (-\v,-\w) {};
\node (eve2) at (0,-\w) {};
\node (eve3) at (\v,-\w) {};

\draw[sArrow] (alice1) to (a1);
\draw[sArrow] (a2) to (alice2);
\draw[sArrow] (alice3) to (a3);

\draw[sArrow] (a1) to (r1);
\draw[sArrow] (r2) to (a2);
\draw[sArrow] (a3) to (r3);

\draw[sArrow] (r1) to (b1);
\draw[sArrow] (b2) to (r2);
\draw[sArrow] (r3) to (b3);

\draw[sArrow] (b1) to (bob1);
\draw[sArrow] (bob2) to (b2);
\draw[sArrow] (b3) to (bob3);

\draw[sArrow] (rr1) to (eve1);
\draw[sArrow] (eve2) to (rr2);
\draw[sArrow] (rr3) to (eve3);

\node at (\a,0) {\Large $\close{\eps}$};

\node[scale=.8,lrnode] (s1) at (\b,\vb+\ub) {};
\node[scale=.8,lrnode] (s2) at (\b,\ub) {};
\node[scale=.8,lrnode] (s3) at (\b,-\vb+\ub) {};
\node[scale=.8,tlrnode] (ss1) at (\b-\v,\ub) {};
\node[scale=.8,tlrnode] (ss2) at (\b,\ub) {};
\node[scale=.8,tlrnode] (ss3) at (\b+\v,\ub) {};
\node[scale=.8,thinResource] (S) at (\b,\ub) {};
\node[yshift=-2,above right] at (S.north west) {\footnotesize $\aS$};

\node[scale=.8,tlrnode] (t1) at (\b-\v,-\ub) {};
\node[scale=.8,tlrnode] (t2) at (\b,-\ub) {};
\node[scale=.8,tlrnode] (t3) at (\b+\v,-\ub) {};
\node[scale=.8,protocolLong] (sim) at (\b,-\ub) {};
\node[xshift=-2,below right] at (sim.north east) {\footnotesize $\sigma_E$};

\node (cate1) at (\b-\tb,\vb+\ub) {};
\node (cate2) at (\b-\tb,\ub) {};
\node (cate3) at (\b-\tb,-\vb+\ub) {};

\node (dave1) at (\b+\tb,\vb+\ub) {};
\node (dave2) at (\b+\tb,\ub) {};
\node (dave3) at (\b+\tb,-\vb+\ub) {};

\node (finn1) at (\b-\v,-\wb) {};
\node (finn2) at (\b,-\wb) {};
\node (finn3) at (\b+\v,-\wb) {};

\draw[sArrow] (cate1) to (s1);
\draw[sArrow] (s2) to (cate2);
\draw[sArrow] (cate3) to (s3);

\draw[sArrow] (s1) to (dave1);
\draw[sArrow] (dave2) to (s2);
\draw[sArrow] (s3) to (dave3);

\draw[sArrow] (ss1) to (t1);
\draw[sArrow] (t2) to (ss2);
\draw[sArrow] (ss3) to (t3);

\draw[sArrow] (t1) to (finn1);
\draw[sArrow] (finn2) to (t2);
\draw[sArrow] (t3) to (finn3);

\end{tikzpicture}


\caption[Cryptographic security]{\label{fig:security}A depiction of
  \defref{def:security}: a protocol $(\pi_A,\pi_B)$ constructs
  $\aS$ from $\aR$ within $\eps$ if the condition
  illustrated in this figure holds. The sequences of arrows at the
  interfaces between the objects represent (arbitrary) rounds of
  communication.}
\end{figure*}
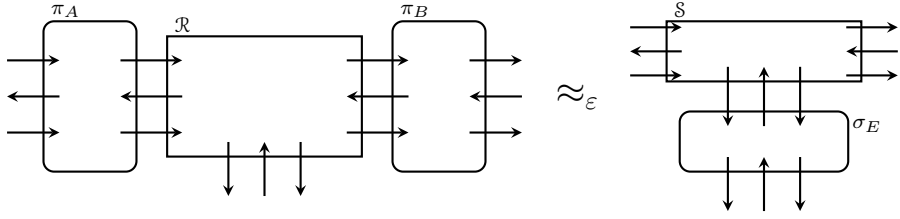

\begin{deff}[Cryptographic security~\cite{MR11}] \label{def:security}
  Let $\pi_{AB} = (\pi_A,\pi_B)$ be a protocol and $\aR$ and $\aS$ two
   resources.  We say that \emph{$\pi_{AB}$ constructs $\aS$
    from $\aR$ within $\eps$}, denoted by
  \[ 
    \aR \xrightarrow{\pi,\eps} \aS \ ,
  \]
  if there exists a converter $\sigma_E$ (called \emph{simulator})
  such that
   \begin{equation} \label{eq:def.sec} d(\pi_{AB}\aR,\aS\sigma_E) \leq
     \eps.\end{equation} If it is clear from the context what
   resources $\aR$ and $\aS$ are meant, we simply say that $\pi_{AB}$
   is $\eps$\=/secure.
\end{deff}

Although this security definition does not refer to any computational
notions, one usually only considers protocols whose converters are
computationally efficient.\footnote{In principle, any reasonable
  notion of efficiency could be considered here. However, if one takes
  the common asymptotic notion of computational complexity classes,
  one would need to describe systems in terms of a computational model
  that enables such asymptotic considerations.}  Furthermore, if one
requires security to hold under composition with protocols that have
only computational security, it is necessary to restrict the choice of
the simulator $\sigma_E$ to converters that are computationally
efficient.  All the converters and resources considered in this work
are efficient in the standard sense, so we will not mention this any
further.

For a given protocol, we usually want to make several security
statements, e.g., one about what is achieved in the presence of an
adversary (sometimes referred to as either the \emph{soundness} or
\emph{security} of a protocol), another about what is achieved when no
adversary is present (usually called either \emph{completeness} or
\emph{correctness}\footnote{In the QKD literature, correctness has
  another meaning \--- it captures the property that Alice and Bob end
  up with identical keys when Eve is active. The term
  \emph{robustness} is traditionally used in the QKD literature to
  denote the performance of a QKD protocol under honest (noisy)
  conditions. We refer to \secref{sec:security.rob} for a discussion
  of the relation between completeness and robustness.}). These two
cases are captured by considering different resources $\aR$ and $\aS$,
but the same protocol $\pi_{AB}$. We will illustrate this in
\secref{sec:qkd} for the case of QKD.


If two protocols $\pi$ and $\pi'$ are
$\eps$- and $\eps'$\=/secure, the composition of the two is
$(\eps+\eps')$\=/secure. More precisely, let protocols $\pi$ and
$\pi'$ construct $\aS$ from $\aR$ and $\aT$
from $\aS$ within $\eps$ and $\eps'$, respectively,
i.e.,
\[\aR \xrightarrow{\pi,\eps}\aS \qquad \text{and}
  \qquad \aS \xrightarrow{\pi',\eps'}\aT.\] It is
then a consequence of the triangle inequality of the distinguishing
metric that $\pi'\pi$ constructs $\aT$ from $\aR$
within $\eps+\eps'$,
\[\aR \xrightarrow{\pi'\pi,\eps+\eps'}\aT.\] A similar
statement holds for parallel composition. Let $\pi$ and $\pi'$
construct $\aS$ and $\aS'$ from $\aR$ and
$\aR'$ within $\eps$ and $\eps'$, respectively, i.e.,
\[\aR \xrightarrow{\pi,\eps}\aS \qquad \text{and}
\qquad \aR' \xrightarrow{\pi',\eps'} \aS'.\]
If these resources and protocols are composed in parallel, we find that $\pi \| \pi'$ constructs $\aS \| \aS'$ from $\aR \| \aR'$ within $\eps+\eps'$, 
\[\aR \| \aR' \xrightarrow{\pi \| \pi',\eps+\eps'}
\aS \| \aS'.\]
Proofs of these statements can be found in \textcite{MR11,Mau12}.

\subsection{Interpretation of the security parameter}
\label{sec:ac.interpretation}

Any pseudo\-/metric which satisfies the basic axioms can be used in
\defref{def:security}. However, the usual pseudo\-/metric is the
\emph{distinguishing advantage}, which was introduced in
\eqnref{eq:adv1} in \secref{sec:ac.realideal}. For two resources $\aR$
and $\aS$ and a distinguisher $\fD$, \eqnref{eq:adv1} may be rewritten
as
  \begin{equation}
  \label{eq:adv2} 
    d^{\fD}\left(\aR,\aS\right) \coloneqq \left| \Pr[\fD(\aR) = 0] - \Pr[\fD(\aS) = 0] \right|,
  \end{equation}
  where $\fD(\aR)$ and $\fD(\aS)$ are the random variables
  corresponding to the output of the distinguisher when interacting
  with $\aR$ and $\aS$, respectively. Alternatively, one may
    define the distinguishing advantage for $\fD$ as
\begin{equation}
  \label{eq:adv3}
  d^{\fD}\left(\aR,\aS\right) \coloneqq \left| 2\ddistinguish{\aR,\aS} - 1
  \right|,
\end{equation}
where $\ddistinguish{\aR,\aS}$ is the probability for $\fD$ to
correctly guess with which of $\aR$ or $\aS$ it is interacting when
either one is chosen with probability $1/2$, i.e.,
\[\ddistinguish{\aR,\aS} \coloneqq \frac{1}{2} \Pr[\fD(\aR) = 0] + \frac{1}{2}
  \Pr[\fD(\aS) = 1].\] It is easy to see that \eqnsref{eq:adv2} and
\eqref{eq:adv3} are equivalent.

One then takes the supremum of this expression over all distinguishers
$\fD$ of a given class $\bD$, i.e.,
  \begin{equation}
  \label{eq:adv4} 
  d^{\bD}\left(\aR,\aS\right)  \coloneqq \sup_{\fD \in \bD} d^{\fD}\left(\aR,\aS\right).
  \end{equation}
  The class $\bD$ may be restricted to a particular set of systems
  (e.g., those that are computationally efficient). The strongest
  security notion corresponds to not imposing any restriction on the
  set of distinguishers (beyond what is allowed by physical laws),
  which is the one considered in most of this work, and which we
  denote
\[
  d(\aR,\aS) \leq \eps \qquad \text{or} \qquad \aR \close{\eps} \aS.
\]
  
The distinguishing advantage is of particular importance because it
has an operational interpretation. If the distinguisher notices a
difference between the two, then something in the real setting did not
behave ideally. This can be loosely interpreted as a failure
occurring. If a distinguisher can guess correctly with probability $1$
with which system it is interacting (i.e.,
$\distinguish{\aR,\aS} = 1$), a failure must occur systematically. If,
conversely, it can only guess correctly with probability $1/2$ (which
corresponds to a random guess), this means that the real system always
behaves like the ideal one, hence no failure occurs at all. The
practically relevant cases are those in between. As shown in
\appendixref{app:op}, a guessing probability
$ \distinguish{\aR,\aS} = p$ corresponds to a failure with probability
$\eps = 2p-1$, which is exactly the distinguishing advantage. The
latter can thus be interpreted as the probability that a failure
occurs in the real protocol. This operational interpretation is
crucial for applications, where one must be able to specify what
maximum value $\eps$ one is ready to tolerate.

A bound on the security $\eps$ of a protocol does however not tell us how ``bad'' this failure is. For example, a key distribution protocol which produces perfectly uniform keys for Alice and Bob, but with probability $\eps$ the keys of Alice and Bob are different, is $\eps$\=/secure. Likewise, a protocol which gives $1$ bit of the key to Eve with probability $\eps$, but is perfect otherwise, and another protocol which gives the entire key to Eve with probability $\eps$, but is perfect otherwise, are both $\eps$\=/secure as well. One could argue that leaking the entire key is worse than leaking one bit, which is worse than not leaking anything but generating mismatching keys, and this should be reflected in the level of security of the protocol. However, leaking one bit can be as bad as leaking the entire key if only one bit of the message is vital, and this happens to be the bit obtained by Eve. Having mismatching keys and therefore misinterpreting a message could have more dire consequences than leaking the message to Eve. How bad a failure is depends on the use of the protocol, and since the purpose of cryptographic security is to make a security statement that is valid for all contexts, bounding the probability that a failure (grave or not) occurs, is the best it can do.

The above is particularly relevant if one considers larger cryptographic tasks that may, for instance, use key distribution numerous times as a subprotocol. Since, as described, a security bound gives no idea of the gravity of a failure, the failure of the key distribution protocol could have an impact on the entire cryptographic system. For example, if the key is used to authenticate later communication, the security of the latter may  be affected by a failure in key distribution. This makes it necessary to choose the probability $\eps$ of a failure in any protocol small enough so that the accumulation of all possible failure probabilities used for the larger cryptographic task are still small.  One way of doing this is to increase the security parameter of a protocol on a regular basis, e.g., once a year the parameters are tweaked so that the new probability of a failure is divided by two. If the accumulated failure during the first year is given by $\eps$, then the total failure over an arbitrarily long lifetime of the system is bounded by $2\eps > \eps + \eps/2 + \eps/4 + \dotsc$

\subsection{Instantiating systems}
\label{sec:ac.instantiating}

As mentioned previously, specifying a concrete behavior of a system
requires a model of systems that satisfies the axioms presented in
\secref{sec:ac.systems}, i.e., provides composition and a
pseudo\-/metric with the required properties. In most of this review
we consider interactive quantum systems with sequential scheduling,
i.e., a system receives a (quantum) message, then sends a (quantum)
message, then receives a (quantum) message, etc. Such systems were
analyzed independently by
\textcite{GW07,CDP09,Har11,Gut12,Har12,Har15} [see also
\textcite{Har05,Har07}], to which we refer in the following using the
term from \textcite{CDP09}, namely \emph{quantum combs}. Quantum combs
are a generalization of \emph{random systems}~\cite{Mau02,MPR07} to
quantum information theory.

What these works essentially show is that an interactive system which
receives the $i^{\text{th}}$ input in register $A_i$ and produces the
$i^{\text{th}}$ output in register $B_i$ and which processes $n$
inputs can be fully described by a completely\-/positive, trace
preserving (CPTP) map
\[ \cE : \cL\left(\bigotimes_{i = 1}^n \hilbert_{A_i}\right) \to
  \cL\left(\bigotimes_{i = 1}^n \hilbert_{B_i}\right).\] Conversely,
any such CPTP map corresponds to an interactive system if it respects
causality, i.e., if for any $j \leq n$ and any
$\rho,\sigma \in \cL\left(\bigotimes_{i = 1}^n \hilbert_{A_i}\right)$
with $\trace[A_{>j}]{\rho} = \trace[A_{>j}]{\sigma}$ we have
\[ \ktrace[B_{>j}]{\cE(\rho)} = \ktrace[B_{>j}]{\cE(\sigma)},\] where
$X_{>j} \coloneqq \bigotimes_{i = j+1}^nX_i$.

Systems such as the resources and converters for the one-time pad in
\figref{fig:otp} \--- or the quantum key distribution systems that
will come in \secref{sec:qkd} \--- all correspond to specific quantum
combs. (Nonetheless we usually give informal descriptions of such systems,
rather than using the comb formalism, especially when the details of their behaviour is not relevant for our claims.) The only results discussed in this
work that cannot be modelled as quantum combs are the relativistic
systems reviewed in \secref{sec:relativistic}, which require a more
complex model of systems that can capture space\-/time and also
satisfies the required axioms, e.g., the causal boxes from
\textcite{PMMRT17}.


\section{Defining security of QKD}
\label{sec:qkd}

The first Quantum Key Distribution (QKD) protocols were proposed
independently by \textcite{BB84} \--- inspired by early work on
quantum money by \textcite{Wie83} \--- and by \textcite{Eke91}. The
original papers discussed security in the presence of an eavesdropper
that could perform only limited operations on the quantum channel. The
models of security evolved over time \--- a review of these is given
in \secref{sec:qkd.other} \--- and the security criterion used today
was introduced in 2005~\cite{RK05,BHLMO05,Ren05}, the so\-/called
\emph{trace distance criterion}. It was argued that $\rho_{KE}$, the
joint state of the final key $K$ and quantum information gathered by
an eavesdropper $E$, must be close to an ideal key, $\tau_K$, that is
perfectly uniform and independent from the adversary's information
$\rho_E$:
\begin{equation} \label{eq:d} (1-\pabort) D \left(     \rho_{KE},\tau_K \otimes \rho_E \right) \leq \eps \   ,
\end{equation} 
where $\pabort$ is the probability that the protocol
aborts,\footnote{In \textcite{Ren05}, \eqnref{eq:d} was introduced
  with a subnormalized state $\rho_{KE}$, with
  $\trace{\rho_{KE}}=1-\pabort$, instead of explicitly writing the
  factor $(1-\pabort)$. The two formulations are however
  mathematically equivalent.} $D(\cdot,\cdot)$ is the trace
distance\footnote{This metric corresponds to the distinguishing
  advantage between two quantum states, and is formally defined in
  \appendixref{app:op}.} and $\eps \in [0,1]$ is a (small) real
number. This security criterion was discussed within the cryptography
frameworks introduced in \secref{sec:ac} by \textcite{BHLMO05,MR09}
\--- see also \appendixref{app:op}.

We note that \eqnref{eq:d} only captures how much an adversary knows
about the key (called \emph{secrecy} in the QKD literature). A QKD
scheme must additionally guarantee that Alice and Bob hold the same
key with high probability (called \emph{correctness}) and that under
reasonable noisy conditions a QKD scheme produces a key with high
probability (called \emph{robustness}). In this section, we describe
how these security notions fit into the general framework described in
\secref{sec:ac}. For this we first explain in
\secref{sec:qkd.overview} how to use the AC framework to model the
task achieved by a QKD protocol \--- namely constructing a secret key
resource from an insecure quantum channel and an authentic classical
channel \--- and write out the corresponding security
definitions. Then in \secref{sec:security} we show how to derive
secrecy, correctness and robustness from these security
definitions. And finally, in \secref{sec:qkd.other} we review other
security definitions that have appeared in the literature, and explain
how they relate to the trace distance criterion, namely \eqnref{eq:d}.

\subsection{The real and ideal QKD systems}
\label{sec:qkd.overview}

In order to apply the general AC security definition to QKD, we first need to specify the ideal key resource, which we do in \secref{sec:qkd.ideal}. Likewise, we specify in \secref{sec:qkd.protocol} the real QKD system consisting of the protocol, an authentic classical channel and an insecure quantum channel. Plugging these systems in \defref{def:security}, we obtain in \secref{sec:qkd.security} the security criteria for QKD.

\subsubsection{Ideal key}
\label{sec:qkd.ideal}

The goal of a key distribution protocol is to generate a secret key shared between two players, Alice and Bob. One can represent such a resource by a box, one end of which is in Alice's lab, and another in Bob's. It provides each of them with a secret key of a given length, but does not give Eve any information about the key. This is illustrated in \figref{fig:qkd.resource.simple}, and is the key resource we used in the OTP construction [\figref{fig:otp.real}].

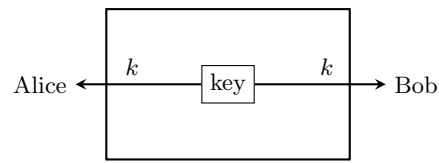
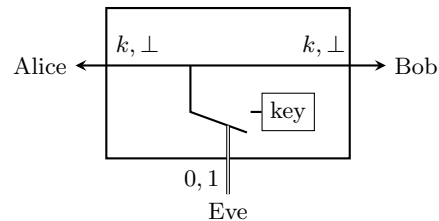
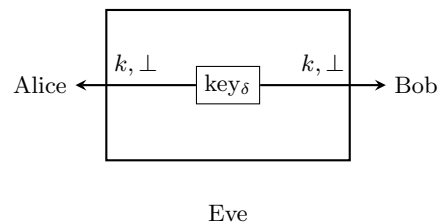
\begin{figure}[tb]
  \centering

  \subfloat[Simple ideal key][\label{fig:qkd.resource.simple}A
  resource that always gives a key $k$ to Alice and Bob, and nothing
  to Eve.]{
\begin{tikzpicture}[
sArrow/.style={->,>=stealth,thick},
largeResource/.style={draw,thick,minimum width=1.618*2cm,minimum height=2cm}]

\small

\def\u{0} 

\node[largeResource] (keyBox) at (0,0) {};
\node (alice) at (-2.5,\u) {Alice};
\node (bob) at (2.5,\u) {Bob};
\node (eve) at (0,-1.7) {Eve};
\node[draw] (key) at (0,0) {key};

\draw[sArrow] (key) to node[pos=.55,auto,swap] {$k$} (alice);
\draw[sArrow] (key) to node[pos=.55,auto] {$k$} (bob);

\end{tikzpicture}
} 

\vspace{6pt}

\subfloat[Ideal key with switch][\label{fig:qkd.resource.switch}A
resource that allows Eve to decide if Alice and Bob get a key $k$ or
an error $\bot$.]{
\begin{tikzpicture}[
sArrow/.style={->,>=stealth,thick},
largeResource/.style={draw,thick,minimum width=1.618*2cm,minimum height=2cm}]

\small

\def\u{.236} 

\node[largeResource] (keyBox) at (0,0) {};
\node (alice) at (-2.5,\u) {Alice};
\node (bob) at (2.5,\u) {Bob};
\node (eve) at (0,-1.7) {Eve};
\node[draw] (key) at (.8,\u/2-.5) {key};
\node (junc) at (-.5,0 |- key.center) {};

\draw[sArrow,<->] (alice) to node[pos=.2,auto] {$k,\bot$} node[pos=.8,auto] {$k,\bot$} (bob);
\draw[thick] (junc.center |- 0,\u) to (junc.center) to node[pos=.666] (handle) {} +(160:-.8);
\draw[thick] (.3,0 |- junc.center) to (key);
\draw[double] (eve) to node[pos=.2,auto] {$0,1$} (handle.center);

\end{tikzpicture}
}

\vspace{6pt}

\subfloat[Probabilistic Ideal
key][\label{fig:qkd.resource.probabilistic}A resource that generates a
perfect key with probability $1-\delta$ and outputs an error $\bot$
with probability $\delta$.]{








\begin{tikzpicture}[
sArrow/.style={->,>=stealth,thick},
largeResource/.style={draw,thick,minimum width=1.618*2cm,minimum height=2cm}]

\small

\def\u{0} 

\node[largeResource] (keyBox) at (0,0) {};
\node (alice) at (-2.5,\u) {Alice};
\node (bob) at (2.5,\u) {Bob};
\node (eve) at (0,-1.7) {Eve};
\node[draw] (key) at (0,0) {key$_\delta$};

\draw[sArrow] (key) to node[pos=.5,auto,swap] {$k,\bot$} (alice);
\draw[sArrow] (key) to node[pos=.5,auto] {$k,\bot$} (bob);

\end{tikzpicture}
}

\caption[Secret key resources]{\label{fig:qkd.resource}Some depictions
  of shared secret key resources.}
\end{figure}

However, if we wish to realize such a functionality with QKD, there is
a caveat: an eavesdropper can always prevent any real QKD protocol
from generating a key by cutting or jumbling the communication lines
between Alice and Bob, and this must be taken into account by the
definition of the ideal resource. This box thus also has an interface
accessible to Eve, which provides her with a switch that, when
pressed, prevents the box from generating this key. We depict this in
\figref{fig:qkd.resource.switch}.

If an OTP protocol uses the key generated by the resource of
\figref{fig:qkd.resource.switch}, we need to consider two cases. If
Eve prevents a key from being generated, the construction is trivially
secure \--- in this case, Alice and Bob do not have a key and
therefore cannot send any message. And in the case where a key
is generated, we have the situation depicted by
\figref{fig:qkd.resource.simple}, which is the situation we already
analyzed in \secref{sec:ac.otp}.

As explained above, an adversary can prevent a key from getting
distributed by disrupting the communication channels. But even if no
adversary is present, one might still wish to take into account that,
due to noise or other disturbance, it can happen that no key is
generated. One may in this case be able to bound the probability of
successfully distributing a key, and so the ideal resource constructed
is stronger than that of \figref{fig:qkd.resource.switch} (where there
is no bound on the probability of getting a key), but weaker than that
of \figref{fig:qkd.resource.simple} (where a key is generated with
probability $1$). This middle point is depicted in
\figref{fig:qkd.resource.probabilistic} (where a key is generated
with probability $1-\delta$) and is treated in
\secref{sec:security.rob}.


\subsubsection{Real QKD system}
\label{sec:qkd.protocol}

\paragraph{Protocol.}
There exist various types of QKD protocols, which differ by their use of resources and hence practical feasibility \cite{SBCDLP09}. For example, in \emph{entanglement\-/based} protocols, first proposed in \textcite{Eke91}, Alice and Bob use a source of entanglement together with a classical authentic (but otherwise insecure) communication channel to generate their keys. Here, we focus on prepare-and-measure schemes, where instead of having access to entanglement, it is assumed that Alice can send quantum states to Bob. These protocols, for which  \textcite{BB84} is the most prominent example, are technologically less challenging than entanglement-based ones, for they do not require the generation of entanglement. Alice merely has to prepare states and send them to Bob, and Bob has to measure them. 

QKD protocols can roughly be divided into three phases: quantum state
distribution, error estimation and classical post\-/processing. In the
first, Alice sends some quantum states to Bob, who measures them upon
reception, obtaining a classical string, called the \emph{raw key}.
In the error estimation phase, they sample some bits at random
positions in the raw key and estimate the noise on the quantum channel
by comparing these values to what Bob should have obtained. If the
noise level is above a certain threshold, they abort the protocol and
output an error message. If the noise is low enough, they move on to
the third phase, in which they perform error correction and privacy
amplification on their respective strings. Error correction allows Bob
to correct the bits where his raw key is different from
Alice's. Privacy amplification turns the raw key, about which an
adversary may still have partial information, into the final secret
key, i.e., uniform strings $k_A$ and $k_B$ for Alice and Bob,
respectively (which should ideally be equal).

\paragraph{Resources.} 
The security of a QKD protocol depends of course also on the resources
we start with. As mentioned previously, we are interested in making
 statements about two cases. In the presence of an active adversary, we
wish to guarantee that any key generated is secure (soundness). But
this is not sufficient, since a protocol that always aborts and never
distributes a key satisfies such a criterion, but is completely
pointless. We thus also want to guarantee that if no adversary is
present \--- only natural (low) noise \--- a key will be generated
with high probability (completeness).

These two cases are modeled by considering different resources in the
real world. In the case of an active adversary, the resources
available for a prepare\-/and\-/measure scheme are a one-way insecure
quantum channel from Alice to Bob (i.e., Eve may change and insert
messages on the channel) and a classical two-way authentic channel
(i.e., it allows authenticated communication from Alice to Bob and Bob
to Alice, but Eve may also listen in). These are illustrated in
\figref{fig:qkd.real.adv}. Recall that this construction is then
supposed to realize the ideal system depicted in
\figref{fig:qkd.resource.switch}.

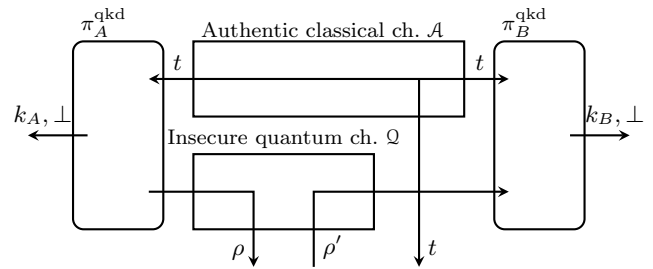
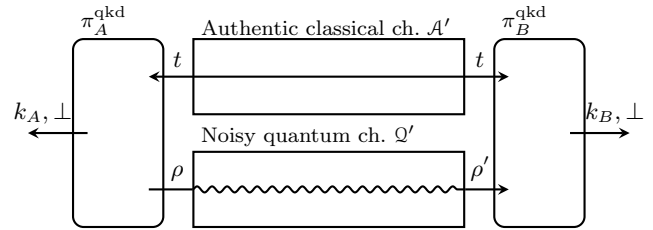
\begin{figure}[tb]
  \centering \subfloat[With adversary][\label{fig:qkd.real.adv}Eve's
  interfaces of the channel resources give her full access to the
  quantum communication and allow her to read the messages on the
  authentic channel.]{
\begin{tikzpicture}[
sArrow/.style={->,>=stealth,thick},
thinResource/.style={draw,thick,minimum width=2.4cm,minimum height=1cm},
longResource/.style={draw,thick,minimum width=3.6cm,minimum height=1cm},
protocol/.style={draw,rounded corners,thick,minimum width=1.2cm,minimum height=2.5cm},
pnode/.style={minimum width=.8cm,minimum height=.5cm}]

\small

\def\t{4} 
\def\u{2.8} 
\def\v{.75}
\def\w{.6} 

\node[pnode] (a1) at (-\u,\v) {};
\node[pnode] (a2) at (-\u,0) {};
\node[pnode] (a3) at (-\u,-\v) {};
\node[protocol] (a) at (-\u,0) {};
\node[yshift=-2,above right] at (a.north west) {\footnotesize
  $\pi^{\qkd}_A$};
\node (alice) at (-\t,0) {};

\node[pnode] (b1) at (\u,\v) {};
\node[pnode] (b2) at (\u,0) {};
\node[pnode] (b3) at (\u,-\v) {};
\node[protocol] (b) at (\u,0) {};
\node[yshift=-2,above right] at (b.north west) {\footnotesize $\pi^{\qkd}_B$};
\node (bob) at (\t,0) {};

\node[longResource] (cch) at (0,\v) {};
\node[yshift=-2,above right] at (cch.north west) {\footnotesize
  Authentic classical ch.~$\aA$};
\node[thinResource] (qch) at (-\w,-\v) {};
\node[yshift=-1.5,above] at (qch.north) {\footnotesize
  Insecure quantum ch.~$\aQ$};
\node (eveq1) at (-\w-.4,-1.75) {};
\node (junc1) at (eveq1 |- a3) {};
\node (eveq2) at (-\w+.4,-1.75) {};
\node (junc2) at (eveq2 |- a3) {};
\node (evec) at (\w+\w,-1.75) {};
\node (junc3) at (evec |- b1) {};

\draw[sArrow,<->] (a1) to node[auto,pos=.08] {$t$} node[auto,pos=.92] {$t$}  (b1);
\draw[sArrow] (junc3.center) to node[auto,pos=.9] {$t$} (evec.center);

\draw[sArrow] (a2) to node[auto,pos=.75,swap] {$k_{A},\bot$} (alice.center);
\draw[sArrow] (b2) to node[auto,pos=.75] {$k_{B},\bot$} (bob.center);

\draw[sArrow] (a3) to (junc1.center) to node[pos=.8,auto,swap] {$\rho$} (eveq1.center);
\draw[sArrow] (eveq2.center) to node[pos=.264,auto,swap] {$\rho'$} (junc2.center) to (b3);

\end{tikzpicture}
}

\vspace{6pt}

\subfloat[Without adversary][\label{fig:qkd.real.noise}In a model with
  natural noise, the resources $\aQ$ and $\aA$ are replaced with
  (non\-/malicious) variants $\aQ'$ and $\aA'$ that have
  a blank interface for Eve and a fixed noise model for the channel
  $\aQ'$.]{
\begin{tikzpicture}[
sArrow/.style={->,>=stealth,thick},
sLine/.style={-,thick},
nLine/.style={-,thick,decorate,decoration={snake,amplitude=.4mm,segment length=2mm,post length=1mm}},
thinResource/.style={draw,thick,minimum width=2.4cm,minimum height=1cm},
longResource/.style={draw,thick,minimum width=3.6cm,minimum height=1cm},
filter/.style={draw,thick,minimum width=1.618cm,minimum height=1cm},
lineFilter/.style={draw,ultra thick,minimum width=.8cm,inner sep=0},
protocol/.style={draw,rounded corners,thick,minimum width=1.2cm,minimum height=2.5cm},
pnode/.style={minimum width=.8cm,minimum height=.5cm}]

\small

\def\t{4} 
\def\u{2.8} 
\def\v{.75}
\def\w{.6} 

\node[pnode] (a1) at (-\u,\v) {};
\node[pnode] (a2) at (-\u,0) {};
\node[pnode] (a3) at (-\u,-\v) {};
\node[protocol] (a) at (-\u,0) {};
\node[yshift=-2,above right] at (a.north west) {\footnotesize
  $\pi^{\qkd}_A$};
\node (alice) at (-\t,0) {};

\node[pnode] (b1) at (\u,\v) {};
\node[pnode] (b2) at (\u,0) {};
\node[pnode] (b3) at (\u,-\v) {};
\node[protocol] (b) at (\u,0) {};
\node[yshift=-2,above right] at (b.north west) {\footnotesize $\pi^{\qkd}_B$};
\node (bob) at (\t,0) {};

\node[longResource] (cch) at (0,\v) {};
\node[yshift=-2,above right] at (cch.north west) {\footnotesize
  Authentic classical ch.~$\aA'$};
\node[longResource] (qch) at (0,-\v) {};
\node[yshift=-1.5,above right] at (qch.north west) {\footnotesize
  Noisy quantum ch.~$\aQ'$};




\draw[sArrow,<->] (a1) to node[auto,pos=.08] {$t$} node[auto,pos=.92] {$t$}  (b1);

\draw[sArrow] (a2) to node[auto,pos=.75,swap] {$k_{A},\bot$} (alice.center);
\draw[sArrow] (b2) to node[auto,pos=.75] {$k_{B},\bot$} (bob.center);

\draw[sLine] (a3) to node[pos=.65,auto] {$\rho$} (qch.west);
\draw[nLine] (qch.west) to (qch.east);
\draw[sArrow] (qch.east) to node[pos=.35,auto] {$\rho'$} (b3);


\end{tikzpicture}
}

\caption[QKD system]{\label{fig:qkd.real}The real QKD system \---
  Alice has access to the left interface, Bob to the right interface
  and Eve to the lower interface \--- consists of the protocol
  $(\pi^{\qkd}_A,\pi^{\qkd}_B)$, the insecure quantum channel $\aQ$ in
  \subref{fig:qkd.real.adv} [and a noisy channel $\aQ'$ in
  \subref{fig:qkd.real.noise}] and two-way authentic classical channel
  $\aA$ (or $\aA'$, respectively). As before, arrows represent the
  transmission of (classical or quantum) messages.

  The protocols of Alice and Bob $(\pi^{\qkd}_A,\pi^{\qkd}_B)$ abort
  if they detect too much interference, i.e., if $\rho'$ is not
  similar enough to $\rho$ to obtain a secret key of the desired
  length. They run the classical post\-/processing over the authentic
  channel, obtaining keys $k_A$ and $k_B$. The message $t$ depicted on
  the two-way authentic channel represents the entire transcript of
  the classical communication between Alice and Bob during the
  protocol.}
\end{figure}

The quantum channel is used in the protocol when Alice sends the
qubits she prepared to Bob. This channel may be completely under the
control of Eve, who could apply any operation allowed by physics to
what is sent over the channel. The authentic channel is used during
the next two phases of the protocol, in which Alice and Bob estimate
the noise in their raw keys and perform the post-processing. Such a
channel faithfully transmits messages between Alice and Bob, but
provides Eve with a copy as well.  Since an authentic channel can be
constructed from an insecure channel and a short shared secret key,
QKD is sometimes referred to as a \emph{key expansion}
protocol.\footnote{We model QKD this way in \secref{sec:smt}.}

The second case is modeled by resources which are not controlled by
Eve anymore. Instead, the quantum channel has a fixed noise model and
the authentic channel does not provide copies of the messages to
Eve. This is drawn in \figref{fig:qkd.real.noise}. With these assumed
resources, the ideal resource one wishes to construct is given by
\figref{fig:qkd.resource.probabilistic}.

\subsubsection{Security}
\label{sec:qkd.security}

For the following, we denote by $(\pi^\qkd_A,\pi^\qkd_B)$ the QKD
protocol, with $\pi^\qkd_A$ and $\pi^\qkd_B$ the converters applied by
Alice and Bob, respectively.  We furthermore denote by $\aQ$ the
insecure quantum channel and by $\aA$ the authentic classical channel,
as drawn in \figref{fig:qkd.real.adv}. Their non\-/malicious
counterparts are denoted $\aQ'$ and $\aA'$,
respectively, as in \figref{fig:qkd.real.noise}. Finally,
let $\aK$ be the secret key resource of
\figref{fig:qkd.resource.switch}, and $\aK'$ the secret key
resource of \figref{fig:qkd.resource.probabilistic}. Applying
\defref{def:security}, we find that $(\pi_A^{\qkd},\pi_B^{\qkd})$
constructs $\aK$ from $\aQ$ and $\aA$ within
$\eps$ if 
\begin{equation} \label{eq:qkd.security}
  \exists \sigma_E, \quad \pi_A^{\qkd}\pi_B^{\qkd}(\aQ \| \aA)
  \close{\eps} \aK \sigma_E,
\end{equation}
and  $(\pi_A^{\qkd},\pi_B^{\qkd})$
constructs $\aK'$ from $\aQ'$ and $\aA'$ within
$\eps'$ if 
\begin{equation} \label{eq:qkd.robust}
  \pi_A^{\qkd}\pi_B^{\qkd}(\aQ'
  \| \aA') \close{\eps'} \aK'.
\end{equation}
Note that no simulator is needed in \eqnref{eq:qkd.robust}, because
both the real and ideal system have a blank interface for Eve.  The
left- and right-hand sides of \eqnref{eq:qkd.security} are illustrated
in Figs.~\ref{fig:qkd.real.adv} and \ref{fig:qkd.resource.sim}, and
the left- and right-hand sides of \eqnref{eq:qkd.robust} are
illustrated in Figs.~\ref{fig:qkd.real.noise} and
\ref{fig:qkd.resource.probabilistic}. These two conditions are
decomposed into simpler criteria in \secref{sec:security}.

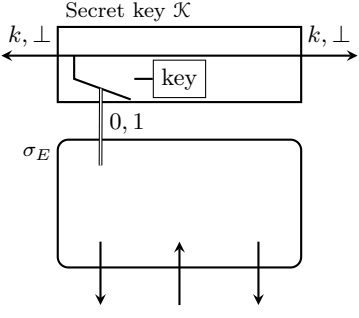
\begin{figure}[tb]
\begin{tikzpicture}[
sArrow/.style={->,>=stealth,thick},
thinResource/.style={draw,thick,minimum width=1.618*2cm,minimum height=1cm},
simulator/.style={draw,rounded corners,thick,minimum width=1.618*2cm,minimum height=1.7cm},
snode/.style={minimum width=1.1cm,minimum height=1.2cm},
innersnode/.style={minimum width=.4cm,minimum height=.3cm}]

\small

\def\t{2.368} 
\def\u{-1.85} 
\def\w{-3.2} 
\def\v{.118} 
\def\s{1.05}

\node[thinResource] (keyBox) at (0,0) {};
\node[draw] (key) at (0,\v/2-.25) {key};
\node (junc) at (-1.4,0 |- key.center) {};
\node[yshift=-1.5,above right] at (keyBox.north west) {\footnotesize
  Secret key $\aK$};
\node (alice) at (-\t,\v) {};
\node (bob) at (\t,\v) {};

\draw[sArrow,<->] (alice.center) to node[pos=.08,auto] {$k,\bot$} node[pos=.92,auto] {$k,\bot$} (bob.center);
\draw[thick] (junc.center |- 0,\v) to (junc.center) to node[pos=.472] (handle) {} +(160:-.8);
\draw[thick] (-.6,0 |- junc.center) to (key);

\node[simulator] (sim) at (0,\u) {};
\node[xshift=1.5,below left] at (sim.north west) {\footnotesize
  $\sigma_E$};
\node[innersnode] (a1) at (-\s,\u+.35) {};
\node[innersnode] (a2) at (-\s,\u-.35) {};
\node[innersnode] (b2) at (\s,\u-.35) {};
\node[innersnode] (c2) at (0,\u-.35) {};

\node (evel) at (-\s,\w) {};
\node (evec) at (0,\w) {};
\node (ever) at (\s,\w) {};

\draw[double] (a1) to node[pos=.55,auto,swap] {$0,1$} (handle.center);
\draw[sArrow] (a2) to (evel.center);
\draw[sArrow] (evec.center) to (c2);
\draw[sArrow] (b2) to (ever.center);

\end{tikzpicture}

\caption[Ideal key \& simulator]{\label{fig:qkd.resource.sim}The key resource from
\figref{fig:qkd.resource.switch} with a simulator $\sigma_E$. This
corresponds to the ideal world in \eqnref{eq:qkd.security}.}
\end{figure}

\subsection{Reduction to the trace distance criterion}
\label{sec:security}

By applying the general AC security definition to QKD, we obtained two
criteria, \eqnsref{eq:qkd.security} and \eqref{eq:qkd.robust},
capturing soundness and completeness, respectively. In this section we
derive the trace distance criterion, \eqnref{eq:d}, introduced at the
beginning of \secref{sec:qkd}, from \eqnref{eq:qkd.security}. We first
show in \secref{sec:security.dist} that the distinguishing advantage
used in the previous sections reduces to the trace distance between
the quantum states gathered by the distinguisher interacting with the
real and ideal systems. Then in \secref{sec:security.simulator}, we
determine the simulator $\sigma_E$ of the ideal system. In
\secref{sec:security.simple} we decompose the resulting security
criterion into a combination of \emph{secrecy} \--- the trace distance
criterion \--- and \emph{correctness} \--- the probability that
Alice's and Bob's keys differ. In the last section,
\ref{sec:security.rob}, we consider the security condition of
\eqnref{eq:qkd.robust}, which captures (security) guarantees in the
absence of a malicious adversary.  We show how this condition can be
used to model the \emph{robustness} of the protocol, i.e., the
probability that the protocol aborts with non\-/malicious noise.

\subsubsection{Trace distance}
\label{sec:security.dist}

The security criteria given in \eqnsref{eq:qkd.security} and
\eqref{eq:qkd.robust} are defined in terms of the distinguishing
advantage between resources. To simplify these equations, we rewrite
them in terms of the trace distance between the states held by the
distinguisher at the end of the protocol in the real and ideal
settings. \textcite{Hel76} proved that the advantage a distinguisher
has in guessing whether it was provided with one of two states with
equal priors, $\rho$ or $\sigma$, is given by the trace distance
between the two, $D(\rho,\sigma)$.\footnote{Actually, \textcite{Hel76}
  solved a more general problem, in which the states $\rho$ and
  $\sigma$ are picked with a priori probabilities $p$ and $1-p$,
  respectively, instead of $1/2$ as in the definition of the
  distinguishing advantage.} A proof of this along with a discussion
of different operational interpretations of the trace distance is
given in \appendixref{app:op}.

We start with the criterion given by~\eqnref{eq:qkd.robust}. The two resources on the left- and right-hand sides of \eqnref{eq:qkd.robust} simply output classical strings (a key or error message) at Alice and Bob's interfaces. Let these pairs of strings be given by the joint probability distributions $P_{AB}$ and $\tilde{P}_{AB}$. The distinguishing advantage between the two resources is thus simply the distinguishing advantage between these probability distributions \--- a distinguisher is given a pair of strings sampled according to either $P_{AB}$ or $\tilde{P}_{AB}$ and has to guess from which distribution it was sampled \--- i.e., 
\[ 
  d\left( \pi_A^{\qkd}\pi_B^{\qkd}(\aQ' \| \aA'), \aK' \right) = d(P_{AB},\tilde{P}_{AB}).
   \] 
 As stated above, the distinguishing advantage between two quantum states is equal to their trace distance, and in the special case where the states are classical \--- i.e., given by two probability distributions \--- the trace distance between the classical states is equal to the total variational distance between the corresponding probability distributions. Thus $d(P_{AB},\tilde{P}_{AB})= D(P_{AB},\tilde{P}_{AB})$, where we use the same notation for both the trace distance and total variational distance, since the latter is a special case of the former. Putting the two together we get
\[ 
d\left( \pi_A^{\qkd}\pi_B^{\qkd}(\aQ' \| \aA')
  , \aK' \right) = D(P_{AB},\tilde{P}_{AB}),
 \] 
  where $P_{AB}$ and $\tilde{P}_{AB}$ are the distributions of the strings output by the real and ideal systems, respectively.

  The resources on the left- and right-hand sides of
  \eqnref{eq:qkd.security} are slightly more complex than those in
  \eqnsref{eq:qkd.robust}. They first output a state $\varphi_C$ at
  the $E$\=/interface, namely the quantum states that Alice sends over
  the insecure quantum channel. Without loss of generality, the
  distinguisher now applies any map $\cE : \lo{C} \to \lo{CE'}$
  allowed by quantum physics to this state, obtaining
  $\rho_{CE'} = \cE(\varphi_C)$ and puts the $C$ register back on the
  insecure channel for Bob, keeping the part in $E'$. Finally, the
  systems output some keys (or error messages) at the $A$ and
  $B$\=/interfaces, and all classical messages exchanged during the
  error estimation and post\-/processing at the $E$\=/interface \---
  this captures the fact that the classical communication is
  public.\footnote{We sometimes refer to the entire sequence of these
    messages as the \emph{classical transcript} of the protocol.} This
  sequence of interactions of the distinguisher with the real or ideal
  QKD systems is illustrated in \figref{fig:qkd.distinguisher}.

  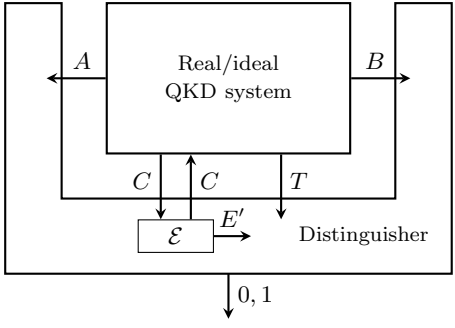
\begin{figure}[tb]
    \centering
    \begin{tikzpicture}[
      sArrow/.style={->,>=stealth,thick},
      largeResource/.style={draw,thick,minimum width=1.618*2cm,minimum
        height=2cm}]

      \node[largeResource,text width=3cm,text centered] (R) at (0,0) {\footnotesize Real/ideal QKD system};

      \draw[thick] (-1.618-1.35,1) -- ++(.75,0) -- ++(0,-2.6)  --
++(1.618*2+1.2,0)  -- ++(0,2.6) -- ++(.75,0) -- ++(0,-3.6) --
++(-1.618*2-2.7,0) -- cycle;

     \node[draw,minimum width=1cm] (E) at (-.7,-2.1) {$\cE$};
     \node at (1.8,-2.1) {\footnotesize Distinguisher};
     
     \draw[sArrow] (R) to node[auto,swap,pos=.4] {$A$} (-1.618-.8,0);
     \draw[sArrow] (R) to node[auto,pos=.4] {$B$} (1.618+.8,0);
     \draw[sArrow] (-.9,0 |- R.south) to node[auto,swap,pos=.4] {$C$}
     (-.9,0 |- E.north);
     \draw[sArrow] (-.5,0 |- E.north) to node[auto,swap,pos=.6] {$C$}
     (-.5,0 |- R.south);
     \draw[sArrow] (.7,0 |- R.south) to node[auto,pos=.4] {$T$}
     (.7,0 |- E.north);
    \draw[sArrow] (E.east |- 0,-2.1) to node[auto] {$E'$}
     (.3,-2.1);

     \draw[sArrow] (0,-2.6) to node[auto] {$0,1$} (0,-3.2);
    \end{tikzpicture}
    \caption[Distinguisher for QKD]{\label{fig:qkd.distinguisher}The
      distinguisher interacting with either the real or ideal QKD
      system first receives a register $C$ containing the quantum
      states sent from Alice to Bob. It applies a map
      $\cE : \lo{C} \to \lo{CE'}$ of its choice, keeps the $E'$
      register and puts $C$ back in the insecure channel. Finally, it
      gets the transcript of the classical communication $T$, and
      Alice's and Bob's outputs $A$ and $B$. It thus holds a state
      $\rho_{ABE'T}$, which it measures to decide if it was
      interacting with the real or ideal system.}
  \end{figure}

  Let $\rho^{\cE}_{ABE}$ be the tripartite state held by a
  distinguisher interacting with the real system, and
  $\tilde{\rho}^{\cE}_{ABE}$ the state held after interacting with the
  ideal system, where the registers $A$ and $B$ contain the final keys
  or error messages, and the register $E$ holds both the state
  $\rho_{E'}$ obtained from tampering with the quantum channel and the
  classical transcript. Distinguishing between these two systems thus
  reduces to maximizing over the distinguisher strategies (the choice
  of $\cE$) and distinguishing between the resulting states,
  $\rho^{\cE}_{ABE}$ and $\tilde{\rho}^{\cE}_{ABE}$:
\[ 
d\left( \pi_A^{\qkd}\pi_B^{\qkd}(\aQ \| \aA) , \aK \sigma_E \right)
= \max_{\cE} d\left( \rho^{\cE}_{ABE},\tilde{\rho}^{\cE}_{ABE} \right).
 \]
 Using again the equality between trace distance and distinguishing
 advantage, we obtain that the advantage a distinguisher has in
 guessing whether it holds the state $\rho^{\cE}_{ABE}$ or
 $\tilde{\rho}^{\cE}_{ABE}$ is given by the trace distance between
 these states, i.e.,
\[ 
  d\left( \pi_A^{\qkd}\pi_B^{\qkd}(\aQ \| \aA) , \aK \sigma_E \right)
= \max_{\cE} D\left(\rho^{\cE}_{ABE},\tilde{\rho}^{\cE}_{ABE} \right). 
\]

The distinguishing advantage between the real and ideal systems of
\eqnref{eq:qkd.security} thus reduces to the trace distance between
the quantum states gathered by the distinguisher. In the following, we
usually omit ${\cE}$ where it is clear that we are maximizing over the
distinguisher strategies, and simply express the security criterion
as \begin{equation} \label{eq:qkd.security.1}
  D(\rho_{ABE},\tilde{\rho}_{ABE}) \leq \eps, \end{equation} where
$\rho_{ABE}$ and $\tilde{\rho}_{ABE}$ are the quantum states gathered
by the distinguisher interacting with the real and ideal systems,
respectively.

\subsubsection{Simulator}
\label{sec:security.simulator}

In the real setting [\figref{fig:qkd.real.adv}], Eve has full control
over the quantum channel and obtains the entire classical transcript
of the protocol. So for the real and ideal settings to be
indistinguishable, a simulator $\sigma^{\qkd}_E$ must generate the
same communication as in the real setting. This can be done by
internally running Alice's and Bob's protocol
$(\pi^{\qkd}_A,\pi^{\qkd}_B)$, producing the same messages at Eve's
interface as the real system. However, instead of letting this
(simulated) protocol decide the value of the key as in the real
setting, the simulator ignores these values and only checks whether a
key is actually produced or whether an error message is generated
instead. It then operates the switch on the secret key resource
accordingly. We illustrate this in \figref{fig:qkd.ideal}.


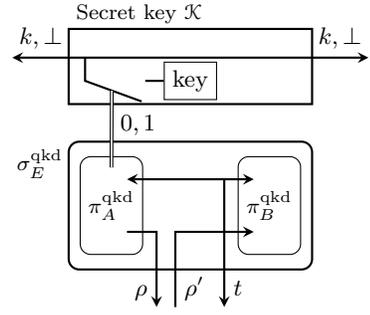
\begin{figure}[tb]

\begin{tikzpicture}[
sArrow/.style={->,>=stealth,thick},
thinResource/.style={draw,thick,minimum width=1.618*2cm,minimum height=1cm},
simulator/.style={draw,rounded corners,thick,minimum width=1.618*2cm,minimum height=1.7cm},
innersim/.style={minimum width=.83cm,minimum height=1.3cm},
innersnode/.style={minimum width=.4cm,minimum height=.3cm}]

\small

\def\t{2.368} 
\def\u{-1.85} 
\def\w{-3.2} 
\def\v{.118} 
\def\s{1.05}

\node[thinResource] (keyBox) at (0,0) {};
\node[draw] (key) at (0,\v/2-.25) {key};
\node (junc) at (-1.4,0 |- key.center) {};
\node[yshift=-1.5,above right] at (keyBox.north west) {\footnotesize
  Secret key $\aK$};
\node (alice) at (-\t,\v) {};
\node (bob) at (\t,\v) {};

\draw[sArrow,<->] (alice.center) to node[pos=.08,auto] {$k,\bot$} node[pos=.92,auto] {$k,\bot$} (bob.center);
\draw[thick] (junc.center |- 0,\v) to (junc.center) to node[pos=.472] (handle) {} +(160:-.8);
\draw[thick] (-.6,0 |- junc.center) to (key);

\node[simulator] (sim) at (0,\u) {};
\node[xshift=1.5,below left] at (sim.north west) {\footnotesize
  $\sigma^{\qkd}_E$};
\node[innersim,rounded corners,draw] (pleft) at (-\s,\u) {\footnotesize $\pi^{\qkd}_A$};
\node[innersim,rounded corners,draw] (pright) at (\s,\u) {\footnotesize $\pi^{\qkd}_B$};
\node[innersnode] (a1) at (-\s,\u+.35) {};
\node[innersnode] (a2) at (-\s,\u-.35) {};
\node[innersnode] (b1) at (\s,\u+.35) {};
\node[innersnode] (b2) at (\s,\u-.35) {};

\node (evel) at (-.45,\w) {};
\node (juncl) at (evel |- a2) {};
\node (evec) at (-.2,\w) {};
\node (juncc) at (evec |- a2) {};
\node (ever) at (.45,\w) {};
\node (juncr) at (ever |- a1) {};

\draw[double] (a1) to node[pos=.55,auto,swap] {$0,1$} (handle.center);
\draw[sArrow,<->] (a1) to (b1);
\draw[sArrow] (juncr.center) to node[pos=.852,auto] {$t$} (ever.center);
\draw[sArrow] (a2) to (juncl.center) to node[pos=.805,auto,swap] {$\rho$} (evel.center);
\draw[sArrow] (evec.center) to node[pos=.25,auto,swap] {$\rho'$} (juncc.center) to (b2);

\end{tikzpicture}

\caption[Simulator for QKD]{\label{fig:qkd.ideal}The ideal QKD system \--- Alice  has access to the left interface, Bob to the right interface and Eve   to the lower interface \--- consists of the ideal secret key   resource and a simulator $\sigma^{\qkd}_E$.}
\end{figure}

The security criterion from \eqnref{eq:qkd.security.1} can now be simplified by noting that with this simulator, the states of the ideal and real systems are identical when no key is produced. The outputs at Alice's and Bob's interfaces are classical elements of the set $\{\bot\} \cup \cK$, where $\bot$ symbolizes an error and $\cK$ is the set of possible keys. The states of the real and ideal systems can be written as
\begin{align*}  \rho_{ABE} & = p^\bot \proj{\bot_A,\bot_B} \otimes
  \rho^\bot_E \\
  & \qquad + \sum_{k_A,k_B \in \cK} p_{k_A,k_B} \proj{k_A,k_B} \otimes
  \rho^{k_A,k_B}_E,\\
  \tilde{\rho}_{ABE} & = p^\bot \proj{\bot_A,\bot_B}
  \otimes \rho^\bot_E \\ & \qquad +\frac{1}{|\cK|} \sum_{k \in \cK} \proj{k,k}
  \otimes \sum_{k_A,k_B \in \cK} p_{k_A,k_B} \rho^{k_A,k_B}_E,
\end{align*}
where $p_{k_A,k_B}$ is the probability of Alice getting the key $k_A$
and Bob getting $k_B$, and $p^\bot$ is the probability of an abort.
Plugging this in \eqnref{eq:qkd.security.1} we get
\begin{equation} \label{eq:qkd.security.2} D\left(
    \rho_{ABE},\tilde{\rho}_{ABE}\right) = (1- p^{\bot})
  D\left(\rho^{\top}_{ABE},\tau_{AB} \otimes \rho^{\top}_{E}\right)
  \leq \eps, 
  \end{equation}
where
 \begin{equation} \label{eq:qkd.security.tmp} \rho^{\top}_{ABE}
  \coloneqq \frac{1}{1- p^{\bot}} \sum_{k_A,k_B \in \cK} p_{k_A,k_B}
  \proj{k_A,k_B} \otimes \rho^{k_A,k_B}_E 
\end{equation}
is the renormalized state of the system conditioned on not aborting
and $\tau_{AB} \coloneqq \frac{1}{|\cK|} \sum_{k \in \cK} \proj{k,k}$
is a perfectly uniform shared key. As previously, the $E$ register
contains the quantum side information that Eve collects about the
states being sent as well as the entire classical transcript of the
error estimation and post\-/processing.

\subsubsection{Correctness \& secrecy}
\label{sec:security.simple}

We now break \eqnref{eq:qkd.security.2} up into two components, often referred to as \emph{correctness} and \emph{secrecy}, and recover the security definition for QKD introduced in \textcite{RK05,BHLMO05,Ren05}. The correctness of a QKD protocol refers to the probability that Alice and Bob end up holding different keys. We say that a protocol is \emph{$\eps_{\corr}$\=/correct} if for all adversarial strategies,
\begin{equation}
  \label{eq:qkd.cor}
  \Pr \left[ K_A \neq K_B \right] \leq \eps_{\corr},
\end{equation}
where $K_A$ and $K_B$ are random variables over the alphabet $\cK \cup \{\bot\}$ describing Alice's and Bob's outputs.\footnote{This can  equivalently be written as $(1-p^\bot)\Pr \left[ K^\top_A \neq     K^\top_B \right] \leq \eps_{\corr}$, where $p^\bot$ is the   probability of aborting and $K^\top_A$ and $K^\top_B$ are Alice and   Bob's keys conditioned on not aborting.} The secrecy of a QKD protocol measures how close the final key is to a distribution that is uniform and independent of the adversary's system. Let $p^\bot$ be the probability that the protocol aborts, and $\rho^\top_{AE}$ be the resulting state of the $AE$ subsystems conditioned on not aborting. A protocol is \emph{$\eps_{\secr}$\=/secret} if for all adversarial strategies,
\begin{equation}
  \label{eq:qkd.sec}
  (1-p^\bot) D\left(\rho^\top_{AE},\tau_A \otimes \rho^{\top}_E\right)
  \leq \eps_{\secr},
\end{equation}
where the distance $D(\cdot,\cdot)$ is the trace distance and $\tau_A$
is the fully mixed state.\footnote{\eqnref{eq:qkd.sec} was already
  introduced at the beginning of \secref{sec:qkd} as \eqnref{eq:d}.}

\begin{thm}
  \label{thm:qkd}
  If a QKD protocol is $\eps_{\corr}$\=/correct and
  $\eps_{\secr}$\=/secret, then \eqnref{eq:qkd.security} is satisfied
  for $\eps = \eps_{\corr} + \eps_{\secr}$.
\end{thm}

This theorem can be proven by using the triangle inequality of the
trace distance to bound \eqnref{eq:qkd.security.2} in terms of the sum
of correctness and secrecy. For completeness, a proof is given in
\appendixref{app:proofs}. This result may also be found in
\textcite{BHLMO05}.

The converse statement can also be shown: if \eqnref{eq:qkd.security} holds for some $\eps$, then the corresponding QKD protocol is both $\eps$\=/correct and $2\eps$\=/secret.\footnote{The factor $2$ is due to the existence quantifier over simulators $\sigma_E$ in the security definition. We cannot exclude that for some specific QKD protocol there exists a  simulator $\bar{\sigma}^\qkd_E$, different from the one used in this proof, that generates a   state $\bar{\rho}_E$ satisfying   $D\left(\rho^{\top}_{AE},\tau_{A} \otimes \bar{\rho}^{\top}_E\right)   \leq D\left(\rho^{\top}_{AE},\tau_{A} \otimes \rho^{\top}_E\right)$.   However, by the triangle inequality we also have that for any   $\bar{\rho}_E$,   $D\left(\rho^{\top}_{AE},\tau_{A} \otimes \bar{\rho}^{\top}_E\right)  \geq \frac{1}{2} D\left(\rho^{\top}_{AE},\tau_{A} \otimes     \rho^{\top}_E\right)$.  Hence the failure $\eps$ of the generic simulator used in this proof   cannot be more than twice as large compared to the optimal one.}

\subsubsection{Robustness}
\label{sec:security.rob}

Correctness and secrecy, as described above, capture the soundness of QKD in the presence of a
malicious Eve, as specified by
\eqnref{eq:qkd.security}.  This is however not sufficient: a QKD protocol which
always aborts without producing any key trivially satisfies
\eqnref{eq:qkd.security} with $\eps=0$, but is obviously not a useful
protocol at all! This is where the second condition, namely
\eqnref{eq:qkd.robust}, is relevant. The real system must not only be
indistinguishable from ideal when an adversary is present and
manipulating the channel, but also when one has a simple noisy
channel, with a blank adversarial interface. In this case, we expect a
secret key to be generated successfully with high probability. This is
captured by considering the strong ideal key resource $\aK'$ from
\figref{fig:qkd.resource.probabilistic} which produces a key with
probability $1-\delta$. If the real system does not generate a key
with the same probability, this immediately results in a gap
noticeable by the distinguisher.

The probability that the real protocol generates a key depends on the
noise introduced by the noisy channel $\aQ'$ [illustrated in
\figref{fig:qkd.real.noise}]. Suppose that this noise is parametrized
by a value $q$, e.g., a depolarizing channel with probability $q$. For
every $q$, the protocol has a probability of aborting, $\delta$, which
is called the \emph{robustness}. Let $\aQ_q$ denote a channel with
this noise model, and let $\aK_\delta$ denote the key resource which
produces an error with a fixed probability
$\delta$. \eqnref{eq:qkd.robust} can thus be phrased as
\begin{equation} \label{eq:robustness} \pi_A^{\qkd}\pi_B^{\qkd}(\aQ_q \|
  \aA') \close{\eps} \aK_\delta
  ,\end{equation} 
  where varying $q$ and $\delta$ results in a family of real and ideal systems.

One can then show that the failure $\eps$ from \eqnref{eq:robustness} is bounded by $\eps_{\corr}+\eps_{\secr}$. Note that this statement is only useful if the probability of aborting, $\delta$, is small for reasonable noise models $q$.

\begin{lem} \label{lem:robustness}
If the resources from \eqnref{eq:robustness} are parametrized such that $\aK_\delta$ aborts with exactly the same probability as the protocol $(\pi_A^{\qkd},\pi_B^{\qkd})$ run on the noisy channel $\aQ_q$, then the completeness of the protocol is bounded by the soundness, i.e.,  
\begin{equation*} d\left( \pi_A^{\qkd}\pi_B^{\qkd}(\aQ_q \| \aA')
  ,\aK_\delta\right)  \leq d\left(
  \pi_A^{\qkd}\pi_B^{\qkd}(\aQ \| \aA),\aK \sigma^{\qkd}_E\right),
\end{equation*} 
where the simulator $\sigma^{\qkd}_E$ is the one used in the previous sections, introduced in \secref{sec:security.simulator}, \figref{fig:qkd.ideal}.
\end{lem}

A proof of this is provided in \appendixref{app:proofs}.

\subsection{Other security criteria}
\label{sec:qkd.other}

\subsubsection{Accessible information}
\label{sec:qkd.other.ai}

As mentioned at the beginning of this section, the trace distance
criterion was only introduced in
2005~\cite{RK05,BHLMO05,Ren05}. Earlier works, e.g.,
\textcite{May96,BBBMR00,SP00}, used a notion of security directly
inspired from classical cryptography, where key techniques such as
\emph{advantage distillation}, \emph{error correction} and
\emph{privacy amplification} were developed
\cite{BBR88,Mau93,AC93,BBCM95}. More concretely, if one denotes an
$n$-bit key random variable by $K$ and the adversary's classical side
information by $Z$, in these works a key was considered secure if the
mutual information per bit between the two is small, i.e.,
$\frac{1}{n}\Ii(K;Z) \leq \eps$, where $\Ii(K;Z) = \Hh(K) -
\Hh(K|Z)$. It was later realized \cite{Mau94,MW00} that the mutual
information per bit is not appropriate in the asymptotic setting,
since $\eps(n) \to 0$ does not imply that the total information about
the key is also small, i.e., one may still have
$n\eps(n) \nrightarrow 0$. It was therefore considered preferable to
directly bound the total information about the key,
$\Ii(K;Z) \leq \eps$.

In the case of QKD, the side information may be quantum, and the
joint system of the key and side information is given by a state
$\rho_{KE}$. The \emph{accessible information} between $K$ and $E$ is
obtained by measuring the $E$ system, and taking the mutual
information between $K$ and the measurement outcome, i.e., 
\begin{equation} \label{eq:localqkd} \Ii_{\text{acc}} (K;E)_\rho
  \coloneqq \max_{\{\Gamma_z\}_z} \Ii(K;\Gamma_Z(E)) \leq \eps,
 \end{equation} 
 where $\Gamma_Z(E)$ is the random variable resulting from measuring
 the $E$ system with the POVM $\{\Gamma_z\}_z$ and as before,
 $\Ii(K;Z) = \Hh(K) - \Hh(K|Z)$ is the mutual information.

 Since measuring a quantum system can only diminish the information it
 provides, one always has $\Ss(K|E) \leq \Hh(K|Z)$ for any random
 variable $Z$ obtained by measuring the $E$ system of a bipartite
 state $\rho_{KE}$, where $\Ss(\cdot)$ is the von Neumann
 entropy. Using the continuity of the conditional von Neumann
 entropy~\cite{AF04}, this can by bounded by its trace distance from
 uniform, namely\footnote{See \corref{cor:AF04} in
   \appendixref{app:op} for a proof of this.}
 \[
 n - \Ss(K|E) \leq 8 \delta n + 2h(2\delta),
 \] 
 where $\delta$ is the trace distance between $\rho_{KE}$ and
 $\tau_K \otimes \rho_E$ and $h(p) = - p \log p - (1-p) \log (1-p)$ is
 the binary entropy. The trace distance criterion thus provides a
 bound on the accessible information.

 Crucially, however, the converse does not hold. As shown in
 \textcite{KRBM07}, it is possible to find a joint state $\rho_{K E}$
 of an $n$-bit key $K$ and the adversary's information~$E$ that
 satisfies \eqnref{eq:localqkd} with $\eps = 2^{-0.18n}$, but knowledge
 of the first $n-1$ bits $K_1$ of $K=K_1K_2$ allow the last bit $K_2$
 to be guessed perfectly.\footnote{This phenomenon is known as
   \emph{information locking} and further examples may be found in
   \textcite{DHLST04,Win17}.} More precisely, if one knows $K_1$ there
 is a way to measure the quantum system $E$ such that the outcome,
 $\Gamma_{Z'}(K_1E)$ is a perfect guess for
 $K_2$,\footnote{\textcite{BHLMO05} show that if $\eps < 2^{-n}$ then
   information locking cannot be exploited, and the adversary's
   advantage in guessing $K_2$ remains exponentially small.} i.e.,
 \[\Ii(K_2;\Gamma_{Z'}(K_1E)) = 1.\]

 To see why this is problematic, suppose that the two parts of the
 key, $K_1$ and $K_2$, are used for One-Time Pad encryption (cf.\
 \secref{sec:ac.otp}) of two messages $M_1$ and $M_2$, respectively,
 of which the first is already known to an eavesdropper (e.g., because
 it contains some publicly available information).  Given $M_1$, the
 eavesdropper can, by listening to the ciphertext, infer $K_1$. This,
 in turn, allows her to apply the measurement yielding $Z'$ to $E$,
 which provides information about $K_2$, and hence also about $M_2$.

 The example emphasises the relevance of \emph{composability}, i.e.,
 the principle that any reasonable notion of \emph{security} should
 have the property that if two cryptographic schemes are considered
 secure then this should also be the case for their
 combination. Criterion~\eqref{eq:localqkd} does not satisfy this
 principle. If a key $K$ generated by a QKD protocol satisfies
 \eqnref{eq:localqkd} then, by definition, it is guaranteed that an
 adversary cannot infer $K$.  But the composition of this QKD protocol
 with the One-Time Pad encryption, which by itself is a perfectly
 secure protocol, is insecure. This is clearly problematic, for such
 compositions of protocols are ubiquitous in cryptography.

 To see how the real\-/world ideal\-/world paradigm avoids this issue,
 imagine a protocol that generates a key consisting of two parts,
 $K_1$ and $K_2$, with the (undesirable) properties as in the example
 from \textcite{KRBM07} described above. The distinguisher could then
 use~$K_1$ to measure~$E$ and check whether the outcome~$Z'$
 determines~$K_2$.  If this is the case then the distinguisher knows
 that it was interacting with the real system ($B=0$), otherwise it
 must have been the ideal one ($B=1$). The distinguisher could thus
 correctly guess the bit~$B$, i.e., the protocol would not meet the
 criterion of being indistinguishable from an ideal system. Hence,
 although the key generation protocol may still satisfy earlier
 criteria such as \eqnref{eq:localqkd}, it would be considered
 insecure, as it should be.

 It is interesting nonetheless to understand what construction a
 security definition like the accessible information corresponds
 to. We discuss this in \secref{sec:alternative.memoryless}, where we
 show that if one assumes that the adversary has no quantum memory,
 then the accessible information is a sufficient security criterion.
 
\subsubsection{Adversarial models}
\label{sec:qkd.other.models}

The definition of cryptographic security introduced in \secref{sec:ac}, \defref{def:security}, does not explicitly mention an adversary. The notion of an adversary is embedded in the distinguisher, which is used to measure the distance between real and ideal systems. The distinguishing metric thus has a dual role: performing the most powerful attack possible and measuring whether this attack was successful \--- i.e., whether it allows real and ideal systems to be distinguished. The reduction to the trace distance criterion discussed in \secref{sec:security} separates these two notions. The trace distance criterion itself, \eqnref{eq:d}, can be seen as the measure of whether the attack resulting in the adversary holding the system $E$ is successful. For this condition to make sense as a security definition, one has to consider all possible adversarial behaviors, i.e., take the maximum of \eqnref{eq:d} over all possible states $\rho_{KE}$ that may occur.

Historically, these two aspects \--- the attack and the criterion for
measuring whether the attack is successful \--- were treated
separately. Early security proofs for QKD, e.g., \textcite{BBBSS92},
did not consider the most powerful attack an eavesdropper could
perform, but only \emph{individual attacks}. These are attack
strategies where the adversary performs an identical operation on each
qubit on the quantum channel and keeps only classical information
$Z$. The information held by Alice, Bob, and Eve is then modeled by
independent and identically distributed (i.i.d.) random variables.

\emph{Collective attacks} \--- a generalization of individual attacks
that allows the eavesdropper to keep quantum information, but still
forces her to perform the same operation on every qubit \--- were
proposed in \textcite{BM97b,BBBvdGM02}. In this setting, one has to
use von Neumann entropy instead of Shannon entropy to measure the
adversary's information about the raw key and compute the achievable
rate \cite{DW05}. Although the adversary's interactions with the
quantum channel are restricted to i.i.d.\ operations, this class of
attacks is particularly important, since proof techniques developed
later \cite{Ren05,Ren07,CKR09,DFR20,AFRV19} show how one can reduce the most
general attack strategies to such a limited one.

The first security proof for QKD that considered a fully general adversary \--- performing \emph{coherent} attacks \--- is attributed to \textcite{May96,May01}. It was then followed by other simpler proofs \cite{BBBMR00,BBBMR06,SP00}. 
In \textcite{BBBMR00,BBBMR06,SP00} the authors point out that security does not hold conditioned on the protocol terminating with a secret key. Instead, one should prove that the probability of the event that the protocol does not abort  \emph{and} that the adversary has non-trivial information about the key is negligible. However, the works discussed above still use basically classical security definitions, such as those based on the accessible information.
 
\subsubsection{Expressing weaker security criteria within the AC framework}
\label{sec:qkd.other.ac}

As discussed above, the early security definitions implicitly imposed
a restriction on the set of possible attack strategies that an
adversary could pursue. Within the modern real-world ideal-world
paradigm, or, more precisely, the AC framework, one can understand
these restrictions as limitations on the distinguisher that tries to
guess whether it is interacting with the real or ideal system. That
is, one does not consider the full set of possible distinguishers, but
only a restricted subset that, e.g., performs i.i.d.\ operations or
takes its final decision by measuring the $E$ system alone, not the
joint $KE$ system.

Alternatively one may also represent these definitions in the AC
framework either by replacing the resources in the ideal setting by
weaker ones, or the resources available in the real setting by
stronger ones. To illustrate the latter, recall that, within the
description of \secref{sec:qkd.overview}, the insecure quantum channel
used by Alice and Bob allows Eve to perform arbitrary operations on
the quantum messages sent. If instead one would provide the players
with a stronger resource that allows Eve to perform only i.i.d.\
operations or allows her to access classical information only, one
would recover weaker security definitions. This is developed in more
detail in \secref{sec:alternative.memoryless}.

Using such an approach, one may still regard the older security
definitions as ``composable'', provided one is aware of the fact that
the real resources are now weaker. In other words, the weaker
definitions do not guarantee that a secret key is obtained from an
authentic classical channel and a completely insecure quantum channel,
but still ensures that a secret key is obtained from some (less
insecure) quantum channel that limits Eve's tampering.

We note that this approach is applicable much more generally, i.e.,
beyond quantum cryptography.  For example, the notion of security
known as \emph{stand\-/alone}~\cite{Gol04} makes the assumption that
the dishonest party does not interact with the environment during the
execution of the protocol. By introducing a resource that restricts
the distinguisher's behavior accordingly, this security definition can
be shown to actually guarantee security, albeit only in a setting
where honest parties have access to such a resource. Similarly, early
definitions of blindness in delegated quantum computing
\cite[DQC;][]{BFK09,FK17} are not known to construct the expected
ideal resource for DQC \--- one which takes the input from the client,
only leaks a bound on the computation size to the server, and returns
a (possibly wrong) computation result to the client
\cite{DFPR14}. However, they do construct a weaker resource which does
not provide the honest player with the result of the computation. This
example is discussed again in \secref{sec:dqc}. A further example are
results in the bounded storage model by \textcite{Unr11}, who obtains
composable security if one limits the number of times a protocol is
run \--- we discuss these further in \secref{sec:open.other}.

\subsubsection{Asymptotic versus finite-size security}

The trace distance criterion, \eqnref{eq:d}, was introduced in
\textcite{RK05,BHLMO05,Ren05} and the relation to composable security
frameworks has been discussed in \textcite{BHLMO05,MR09} \--- see also
\appendixref{app:op}. Security proofs for QKD with respect to this
criterion were developed at the same time
\cite{CRE04,RGK05,Ren05}. Although these new security proofs arguably
use the right security definition, they only prove security
asymptotically. This means that instead of computing the failure
$\eps$ for specific parameters of the protocol, one shows that
$\eps(n) \to 0$ when $n \to \infty$, where $n$ is a parameter that
quantifies some resource, typically the number of quantum signals sent
during the protocol. This does not allow the failure to be evaluated
for any implementation of the protocol, since implementations must
necessarily generate a key in finite time, and hence with finite
resources.

In the asymptotic setting one often demands that the function
$\eps(n)$ be \emph{negligible}, i.e., smaller than $1/p(n)$ for any
polynomial $p(\cdot)$ \--- for example, $\eps(n)$ could be
exponentially small in $n$. The reasoning is usually that honest
players are polynomially bounded, so they will never run the protocol
more than $p(n)$ times and the accumulated error $p(n)\eps(n)$ is then
still negligible. Although such a requirement is standard in
cryptography, it is not directly useful for practical purposes, as
already indicated above.  For example, a protocol with a failure given
by a function $\eps(n)$ which is exponentially small for $n \geq
10^{10^{10}}$ but equal to $1$ otherwise, where $n$ is the number of
signals exchanged between the players, is asymptotically secure, and
yet completely insecure for any realistic parameters. Conversely, the
function $\eps(n) = 10^{-18}$ is considered insecure in the asymptotic
setting, but it guarantees that the protocol can be run once per
second for the lifetime of the universe, and still have an accumulated
error substantially smaller than $1$.

This illustrates that asymptotic security claims can be highly
ambiguous. It is thus necessary to prove finite security bounds if one
wishes to actually use a cryptographic scheme. This has been done for
basic protocols in, e.g.,
\textcite{ILM07,SR08,STS10,TLGR12,HT12,TL17}. For more advanced
protocols, which are specifically designed to be implementable with
imperfect hardware, finite-size security claims can be found, e.g., in
\textcite{LCWXZ14} for decoy\-/state QKD (which will be discussed in
\secref{sec:attacks:countermeasures}), in~\textcite{YinChen} for
twin-field QKD \cite{Lucamarinietal}, and in \textcite{CXCLTL14} for
measurement\-/device\-/independent QKD (which will be discussed in
\secref{sec:alternative.semi}).



\subsubsection{Variations of the trace distance criterion}
\label{sec:qkd.other.variations}

An alternative definition for $\eps$\=/secrecy has been proposed in the literature instead of the trace distance criterion \cite{TSSR10,TLGR12}:
\begin{equation}\label{eq:alternative.sec}
  (1-\pabort)\min_{\sigma_E}D\left(\rho_{KE}, \tau_K \otimes
    \sigma_E\right) \leq \eps.\end{equation}
This alternative notion is equivalent to the standard definition of secrecy [\eqnref{eq:d}] up to a factor $2$, hence any QKD scheme proven secure with one definition is  still secure according to the other, with a minor adjustment of the failure parameter $\eps$. However, we do not know how to derive this alternative notion from a  composable framework. In particular, it is not clear if the failure $\eps$ from \eqnref{eq:alternative.sec} is additive under parallel composition. For example, the concatenation of two keys that each, individually, satisfy \eqnref{eq:alternative.sec}, could possibly have distance from uniform greater than $2\eps$.\footnote{The arXiv version of \textcite{TLGR12} was updated to use \eqnref{eq:d} instead.}


\section{Assumptions for security}
\label{sec:attacks}

The security of a quantum cryptographic protocol relies on assumptions about the physics of the devices that are employed to implement the protocol.  In this section, we discuss these assumptions. For concreteness, we focus on the case of QKD, for which we describe the full set of assumptions in \secref{sec:attacks:assumptionlist}.  We then explain why these assumptions are needed and to what extent they are justified in \secref{sec:attacks:necessity}. Experimental work in QKD has shown however that the assumptions are often very difficult to meet, and are actually not met in many cases. This fact can be exploited by quantum hacking attacks, which are described in \ref{sec:attacks:hacking}. Finally, in Section~\ref{sec:attacks:countermeasures}, we discuss countermeasures against these attacks. 

\subsection{Standard assumptions for QKD} \label{sec:attacks:assumptionlist}

The security of QKD protocols usually relies on the following assumptions.

\begin{enumerate}
  \item \label{item_qm} All devices used by Alice and Bob, as well as the communication channels connecting them, are correctly and completely\footnote{The completeness of quantum theory can be derived from their correctness;  see \secref{sec:completeness}.} described by quantum theory.
    \item \label{item_res} The channel that Alice and Bob use to exchange classical messages is authentic, i.e., it is impossible for an adversary to modify messages or insert new ones. 
  \item \label{item_conv} The devices that Alice and Bob use locally to execute the steps of the protocol, e.g., for preparing and measuring quantum systems, do exactly what they are instructed to do.
\end{enumerate}

As already indicated earlier, due to the lack of proof techniques, additional assumptions had been introduced in the past. A prominent example is the \emph{i.i.d.}\ assumption, which demands that the quantum channel connecting Alice and Bob be described by a sequence of identical and independently distributed maps. Physically, this means that an adversary's interception strategy is such that each signal sent from Alice to Bob is modified in the same manner and independently of the other signals. Security under the i.i.d.\ assumption is called security against \emph{collective attacks}~\cite[see also \secref{sec:qkd.other.models}]{BM97b,BBBvdGM02}. Another assumption, which  usually comes on top of the i.i.d.\ assumption, is that Eve only stores classical data, which she obtains by individually measuring the pieces of information she gained from each signal sent from Alice to Bob. Since it is difficult to argue why an adversary should be restricted in that particular way, the corresponding security guarantee is rather weak. It is usually referred to as security against \emph{individual attacks}~\cite[see \secref{sec:qkd.other.models}]{Fuchsetal1997,Lutkenhaus2000}. 

Most modern security proofs do however not require such additional assumptions, i.e., they are based entirely on Assumptions~\ref{item_qm}--\ref{item_conv} above. This means, in particular, that the quantum channel connecting Alice and Bob can be arbitrary, and may even be entirely controlled by Eve. In this case, one talks about security against \emph{general attacks}, \emph{coherent attacks}, or \emph{joint attacks}. Sometimes the term  \emph{unconditional security} appeared in the literature~\cite{SBCDLP09}, but it is important to keep in mind that the assumptions listed above are still necessary.

\subsection{Necessity and justification of assumptions} \label{sec:attacks:necessity}

Assumption~\ref{item_qm} is often implicit, for it is a prerequisite to even describe the cryptographic scheme. It justifies the use of the formalism of quantum theory to model the different systems, such as the communication channel, including any possible attacks on them. The assumption thus captures the main idea behind quantum cryptography, namely that an adversary is limited by the laws of quantum theory.  The other two assumptions ensure that the experimental implementation follows the theoretical prescription that enters the security definition (Definition~\ref{def:security}), namely the description of the protocol $\pi_{AB}$ and the used resources. In particular, Assumption~\ref{item_res} guarantees that the resources shared between Alice and Bob fulfil the theoretical specifications~$\aR$, which in the case of QKD includes the classical authentic communication channel. Assumption~\ref{item_conv} guarantees that the steps prescribed by the protocol~$\pi_{A B}$ are correctly executed.

Assumption~\ref{item_qm} is widely accepted \--- and proving it wrong would represent a major breakthrough in physics. Nevertheless, it has been shown that there exist QKD protocols that only rely on the weaker assumption of \emph{no-signalling}~\cite{BHK05}.  

Assumption~\ref{item_res} demands that an authentic communication channel is set up between Alice and Bob. There exist information-theoretically secure protocols that achieve this, provided that Alice and Bob share a weak secret key~\cite[see also \secref{sec:smt.auth}]{RW03,DW09,ACLV19}.  Assumption~\ref{item_res} can thus be met by the use of such authentication protocols (see also~\secref{sec:intro} as well as standard textbooks on classical cryptography)

Although Assumption~\ref{item_conv} sounds rather natural, and is in fact required for almost any cryptographic scheme, including any classical one, it is rather challenging to meet.  Numerous quantum hacking experiments, which have been conducted over the past few years, have shown that many implementations of QKD failed to satisfy this assumption. To illustrate this problem, we describe selected examples of such attacks in the following subsection.

\subsection{Quantum hacking attacks} \label{sec:attacks:hacking}          

We start with the \emph{photon number splitting attack}~\cite{Brassardetal2000}, which targets optical implementations of QKD that use individual photons as quantum information carriers. Suppose, for concreteness, that Alice and Bob implement the BB84 protocol~\cite{BB84} by encoding the qubits into the polarisation degree of freedom of individual photons. Specifically, Alice may use a single-photon source that emits photons with a polarisation that she can choose. The BB84 protocol\footnote{This protocol is explained in more detail in \secref{sec:securityproofs}, where a security proof is also sketched.} requires her to send in each round at random a state from one orthonormal basis, say $\{\ket{h}, \ket{v}\}$, where $\ket{h}$ may be realised by a horizontally polarised photon and $\ket{v}$ by a vertically polarised one, or from a complementary basis $\{\ket{d^+}, \ket{d^-}\}$, where $\ket{d^+} = \smash{\frac{1}{\sqrt{2}}} (\ket{h} + \ket{v})$ and $\ket{d^-} = \smash{\frac{1}{\sqrt{2}}} (\ket{h} - \ket{v})$. It may now happen that, in an experimental implementation, the source sometimes accidentally emits two photons at once, which then carry the same polarisation. The states emitted in the four cases are thus $\ket{h} \otimes \ket{h}$, $\ket{v} \otimes \ket{v}$, $\ket{d^+} \otimes \ket{d^+}$, and $\ket{d^-} \otimes \ket{d^-}$.

Before describing the actual attack, we first give a simple information-theoretic argument for why this is problematic. Note first that one single photon carries no information about the choice of the basis made by Alice. Indeed, for either of the basis choices, the density operator describing the photon is maximally mixed, i.e., $\frac{1}{2} \proj{h} + \frac{1}{2} \proj{v} = \frac{1}{2} \proj{d^+} + \frac{1}{2} \proj{d^-} =  \frac{1}{2} \mathbf{1}$. This is however no longer the case for a pulse consisting of two photons, i.e., 
\begin{align}
  \frac{1}{2} \proj{h}^{\otimes 2} + \frac{1}{2} \proj{v}^{\otimes 2} \neq \frac{1}{2} \proj{d^+}^{\otimes 2} + \frac{1}{2} \proj{d^-}^{\otimes 2} \ .
\end{align}
Hence, if the source accidentally emits two equally polarised photons instead of one, it reveals information about Alice's basis choice, which it shouldn't. 

It is therefore not surprising that such two-photon pulses can be exploited by an adversary to attack the system. Eve, who intercepts the channel, may split the two-photon pulse into two, keep one of the photons and send the other one to Bob. The latter thus receives photons in exactly the way prescribed by the protocol, and hence does not notice the interception. Eve, meanwhile, may measure the photons she captured. In principle, if Eve had quantum memory, she could even wait with the measurement until Alice announces the basis choice to Bob, and hence always gain full information about the polarisation state that Alice prepared. 

While the photon number splitting attack exploits an imperfection of the sender (namely that it sometimes emits two identically polarised photons instead of one), many quantum attacks are targeted towards the receiver. An example is the \emph{time-shift attack}~\cite{Makarovetal2006,qi2007time,Zhaoetal2008}, which exploits inaccuracies of the photon detectors. In order to avoid dark counts, the photon detectors are often set up such that they only count photons that arrive within a small time window around the time when a signal is expected to arrive. Furthermore, Bob's receiver device may consist of more than one detector, e.g., one for each possible polarisation state. The time windows of the different detectors are then never perfectly synchronised. This means that there are times at which the receiver is more sensitive to signals with respect to one polarisation than another. Eve may therefore, by appropriately delaying the signals sent from Alice and Bob, bias the detected signals towards one or the other polarisation, and thus gain information about what Bob measures. While this information may be partial, it can, together with the error correction information that is available to Eve, be sufficient to infer the final key. 

Another attack that is targeted towards the receiver is the \emph{detector blinding attack} ~\cite{Makarov2009,WKRFNW11,LWWESM10,GLLSKM11}, where the adversary tries to control the detectors by illuminating them with bright laser light.  In a QKD implementation that uses the encoding of information into the polarisation of individual photons, the detectors are usually configured such they can optimally detect single photon pulses. That is, they should click whenever the incoming pulse contains a photon, and not click if the pulse is empty.  However, the behaviour of such detectors may be rather different in a regime where the incoming pulses contain many photons. For example, it could be that they always click when they are exposed to bright light with a particular intensity, and they may never click for another intensity. Hence, by sending in light with appropriately chosen polarisation and intensity, Eve may gain immediate control over the clicks of Bob's detector. To exploit this for an attack, Eve may mimic Bob's receiver, i.e., intercept the photons sent from Alice and measure them in a randomly chosen basis, as Bob would do. She then sends bright light to Bob to ensure that he obtains the same detector clicks as if he had directly obtained Alice's photons. This works particularly well for implementations that use a \emph{passive basis choice}, i.e., where Bob's measurement basis is not  provided as an input, but rather made by the detection device itself. In this case, an adversary can essentially remote-control Bob and thus get hold of the entire key. 

Yet another hacking strategy are \emph{Trojan-horse attacks}~\cite{Vakhitov2001,GisinFaselKraus2006}. Here the idea is to send a bright laser pulse via the optical fibre into Alice or Bob's component to extract information about its internal settings. Depending on the sender and receiver hardware which is used, measuring the reflection of the pulse can allow Eve, for instance, to determine the basis choices made by Alice and Bob.

In some optical implementations of QKD, e.g., in the \emph{plug-and-play}~\cite{Muller97} or the \emph{circular-type}~\cite{Nishioka2002} system, Alice does not have a photon source but instead encodes information by modulating an incoming signal from Bob before sending it back to him. The signal thus travels twice in opposite directions through the same optical links, which helps reducing fluctuations due to birefringence  and environmental noise. The two-fold use of the (insecure) channel however opens additional possibilities of attacks~\cite{GisinFaselKraus2006}. A prominent example is the \emph{phase-remapping attack}~\cite{FQTL07,XQL10}. It exploits the fact that the modulator used by Alice to encode information into the signal coming from Bob acts on that signal during a particular time interval. In the attack, the adversary slightly advances or delays the signal on its way from Bob to Alice, so that it no longer lies fully within that time interval. The modulation by Alice will then be incomplete, which means that the encoding of the information in the signal differs from what is foreseen by the protocol. This can in turn be exploited by Eve in an intercept-and-resend attack on the signal returned from Alice to Bob.

\subsection{Countermeasures against quantum hacking} \label{sec:attacks:countermeasures}

The attacks described here have in common that they all exploit a breakdown of Assumption~\ref{item_conv}. Specifically, in the case of the photon-splitting attack, the device used by Alice sends out more information than it is supposed to. In the case of the time-shift attack, it is Bob's measurement device  whose measurement operators are not constant over time and can even be partially controlled by Eve. Finally, in the case of the detector blinding attack on systems with passive basis choice, Eve even takes over control of the randomness used to choose the basis.

A seemingly obvious countermeasure to prevent such attacks is to manufacture sources and detectors that meet the theoretical specifications. That is, one would need a perfect single-photon source, as well as detectors that are perfectly efficient and only measure photon pulses in a specified parameter regime. Such requirements are however unrealistic --- the devices used in experiments will always, at least slightly, deviate from these specifications. 

The other possibility is to develop cryptographic protocols and  security proofs that tolerate imperfections of the devices~\cite{GLLP04}.  This has been done in particular for the attacks described above. To prevent photon number splitting attacks, an efficient countermeasure is the \emph{decoy-state} method~\cite{Hwang2003,Wang2005,Loetal2005}. The idea here is that Alice sometimes deliberately sends multi-photon pulses. Alice and Bob can  then check statistically whether an adversary captured them. Another possibility is to use protocols where Alice's encoding of information has the property that, even when one photon is extracted from a pulse, the information about what Alice sent is still partial~\cite{SARG,TamakiLo,SYK14}. In the case of time-shift attacks, it is sufficient to characterise the maximum bias in the detector efficiencies that can be introduced and account for it in the security proofs. Finally, for the detector blinding attacks, a possible countermeasure is to add tests to the protocol, such as a monitoring of the photocurrent, in order to detect those~\cite{Yuanetal2010}.  

The main problem with such countermeasures is however that the space of possible imperfections is hard to characterise. The above are just a few examples of attacks, and many others have been proposed, and sometimes even demonstrated to work successfully in experiments. For example, an adversary may exploit imperfections in the randomness that Alice and Bob use for choosing their measurement basis. To prevent such attacks, one may again extend the protocols such that they can tolerate imperfect randomness (see \secref{sec:alternative.randomness}).

The last decade has thus seen an arms race between designers and attackers of quantum cryptographic schemes. A possible way out of this unsatisfactory situation is \emph{device-independent cryptography}. Here the idea is to replace Assumption~\ref{item_conv} by something much weaker. Namely, one requires that the devices used by Alice and Bob do not unintentionally send information out to an adversary, and that the classical processing of information done by Alice and Bob is correct. Crucially, however,  one does no longer demand that the sources and detectors used by Alice and Bob work according to their specifications. The way this can work is explained in \secref{sec:alternative.di}. 


\section{Security proofs for QKD}
\label{sec:securityproofs}

In this section, we discuss security proofs for QKD. For this we consider a generic protocol as shown in \figref{fig:GenericQKD}. The techniques presented here are however not restricted to QKD. Concepts such as information reconciliation or privacy amplification, which we will describe in this section, also play a role in other protocols, for instance those discussed in \secref{sec:other}.

While the first QKD security proofs such as~\textcite{May01,SP00} treat  the entire QKD protocol as a whole, modern security proofs are modular~\cite{Ren05}. This means that a separate security statement is established for each part of the cryptographic protocol. The overall security statement for QKD then follows by combining these individual statements.  In the case of the protocol shown in \figref{fig:GenericQKD}, one statement concerns the raw key distribution and parameter estimation step (see \secref{sec:RawKeyDistribution}), another one the information reconciliation step (see \secref{sec_infrec}), and yet another one the privacy amplification step (see \secref{sec_pa}). According to the AC framework, each part can be regarded as a constructive statement, asserting that the corresponding subprotocol constructs a particular resource from certain given resources. This modular analysis does not only come with the obvious advantage that the proofs are more versatile and can be adapted to different protocols, but also that the arguments are more transparent and easier to understand and verify.

In the following, we focus on the modular approach to proving security proposed in~\textcite{Ren05}. We note however that there exist various other methods (we discuss these in \secref{sec_othersecurityproofs}). The common feature of all security proofs is that they derive a relation between the information accessible to the legitimate parties and the maximum information that may have been gained by Eve. In the description below, this relation is given by \eqnref{eq_HminBB84}; it lower bounds Eve's uncertainty about the raw key~$\mathbf{X}$ generated by Alice. Crucially, although the statement concerns Eve's knowledge, the bound depends only on data that is accessible to Alice and Bob, in this case the error rate $\eta_0$  between their raw keys $\mathbf{X}$ and~$\mathbf{Y}$.  

There are various different ways to derive and interpret such bounds on Eve's information. In the case of prepare-and-measure schemes, they can be understood as consequences of the \emph{no-cloning principle}~\cite{Wootters82}. According to this principle, if Eve attempts to copy parts of the information transmitted from Alice to Bob into her register $E$,  the transmitted information is disturbed, resulting in a decrease of the correlations between Alice and Bob. This disturbance is larger the more information Eve has gained --- a fact that is known as the \emph{information-disturbance tradeoff}~\cite{Fuchs98}. In the case of entanglement-based protocols, the bounds on Eve's information can be regarded as an instance of the \emph{monogamy of entanglement}. It asserts that the stronger Alice's entanglement with Bob the weaker is her correlation with Eve~\cite{Coffman00,Terhal04,KoashiWinter04}.

\begin{figure}
\noindent\fbox{%
    \parbox{\columnwidth}{%
\begin{enumerate}
  \item \textbf{protocol} $\mathrm{QKD}$
  \item $(\mathbf{X}, \mathbf{Y}) := \mathrm{RawKeyDistribution()}$
  \item \textbf{if} $\mathrm{ParameterEstimation}(\mathbf{X}, \mathbf{Y}) = \mathrm{fail}$ \\ \textbf{then} \textbf{return} $(\perp, \perp)$ and \textbf{abort}
  \item $(\mathbf{X}, \mathbf{X'}) := \mathrm{InformationReconciliation}(\mathbf{X}, \mathbf{Y})$
  \item $(\mathbf{S}, \mathbf{S'}) := \mathrm{PrivacyAmplification}(\mathbf{X}, \mathbf{X'})$
  \item \textbf{return} $(\mathbf{S}, \mathbf{S'})$ 
\end{enumerate}
    }%
}
\caption{Generic QKD protocol \label{fig:GenericQKD}}
\end{figure}

\subsection{Protocol replacement} 

Cryptographic protocols that are optimised for practical use are often not easy to analyse directly. Conversely, protocols that are designed in a way that simplifies their security proofs are usually not easily implementable in practice. For example, building an entanglement-based QKD protocol in practice is technologically more challenging than building a  prepare-and-measure scheme. Conversely, the structure of entanglement-based schemes fits more naturally with the known techniques for proving security.

A first step in a security proof for a practical protocol $\pi_{\mathrm{practical}}$ is thus usually to conceive of another protocol $\pi_{\mathrm{theoretical}}$ that is adapted to the proof techniques at hand. One then argues that, for the purpose of the security proof, $\pi_{\mathrm{practical}}$ can be replaced by $\pi_{\mathrm{theoretical}}$, i.e., that the security of $\pi_{\mathrm{practical}}$ is implied by the security of $\pi_{\mathrm{theoretical}}$. A generic way to achieve this is to show that for any possible attack against $\pi_{\mathrm{practical}}$ there exists a corresponding attack against $\pi_{\mathrm{theoretical}}$.

For a concrete example, suppose that $\pi_{\mathrm{practical}}$ is the BB84 protocol~\cite{BB84}. The protocol follows the generic structure shown in \figref{fig:GenericQKD}, with a particular raw key distribution procedure as shown in \figref{fig:BB84RawKeyDistribution}. The protocol prescribes that Alice and Bob proceed in rounds. In each round~$i$, Alice inputs one qubit $Q_i$ to the quantum channel.  The qubit encodes a random signal bit $X_i$ with respect to a randomly chosen basis $B_i$. Bob measures the output $Q'_i$ of the quantum channel with respect to a randomly chosen basis $B'_i$ to obtain a bit~$Y_i$. This is a prepare-and-measure scheme and in this sense ``practical''.

\begin{figure}[h]
\noindent\fbox{%
    \parbox{\columnwidth}{%
\begin{enumerate}
  \item \textbf{protocol} $\mathrm{RawKeyDistribution}()$ [BB84]
  \item \textbf{parameters} 
     $n$ [number of signals];
     $\phi_{x, 0} := \ket{x}$, $\phi_{x,1} := \frac{1}{\sqrt{2}} (\ket{0} + (-1)^x \ket{1})$, for $x \in \{0,1\}$  [bases for encoding] 
  \item $i:=1$
  \item \textbf{while} $i \leq n$ \textbf{do}
  \item Alice chooses $B_i, X_i \in_R \{0, 1\}$ 
  \item Bob chooses $B'_i \in_R \{0,1\}$
  \item Alice prepares a qubit $Q_i$ in  state $\phi_{X_i,B_i}$ and gives it as input to the quantum channel
  \item Bob measures the output $Q'_i$ of the quantum channel w.r.t.\ basis $\{\phi_{0, B'_i}, \phi_{1, B'_i}\}$ to get $Y_i$
  \item Alice and Bob communicate $B_i$ and $B'_i$ over the classical channel
  \item \textbf{if} $B_i = B'_i$ \textbf{then} $i:=i+1$
  \item \textbf{endwhile}
  \item \textbf{return} $(\mathbf{X} = (X_1, \ldots, X_n), \mathbf{Y} = (Y_1, \ldots, Y_n))$
\end{enumerate}
    }%
}
\caption{Prepare-and-measure raw key distribution \label{fig:BB84RawKeyDistribution}}
\end{figure}

The corresponding ``theoretical'' protocol $\pi_{\mathrm{theoretical}}$ could be an entanglement-based protocol similar to the E91 protocol \cite{Eke91}. This protocol is identical to the BB84 protocol described above, except that the raw key distribution step is replaced by the procedure shown in \figref{fig:EntanglementBasedRawKeyDistribution}. In each round~$i$, Alice creates an entangled state between two qubits $\bar{Q}_i$ and $Q_i$ and sends the latter to Bob, who receives it as $Q'_i$.\footnote{Security is also guaranteed if this entangled state is generated by an untrusted third party and distributed to Alice and Bob.}  Alice and Bob then both select random bases $B_i$ and $B'_i$ and measure their qubits accordingly to obtain bits $X_i$ and $Y_i$, respectively. 

\begin{figure}[h]
\noindent\fbox{%
    \parbox{\columnwidth}{%
\begin{enumerate}
  \item \textbf{protocol} $\mathrm{RawKeyDistribution}()$ [entanglement-based]
  \item \textbf{parameters} 
     $n$ [number of signals];
     $\phi_{x, 0} := \ket{x}$, $\phi_{x,1} := \frac{1}{\sqrt{2}} (\ket{0} + (-1)^x \ket{1})$, for $x \in \{0,1\}$  [bases for encoding]
  \item $i:=1$
  \item \textbf{while} $i \leq n$ \textbf{do}
  \item Alice chooses $B_i \in_R \{0, 1\}$ 
  \item Bob chooses $B'_i \in_R \{0,1\}$
  \item Alice prepares qubits $(\bar{Q}_i, Q_i)$ in state $\smash{\frac{1}{\sqrt{2}} (\ket{0} \ket{0} + \ket{1} \ket{1})}$ and gives $Q_i$ as input to the quantum channel
    \item Alice measures $\bar{Q}_i$ w.r.t.\ basis $\{\phi_{0, B_i}, \phi_{1, B_i}\}$ to get $X_i$ 
  \item Bob measures the quantum channel output $Q'_i$ w.r.t.\ basis $\{\phi_{0, B'_i}, \phi_{1, B'_i}\}$ to get $Y_i$
    \item Alice and Bob communicate $B_i$ and $B'_i$ over the classical channel
  \item \textbf{if} $B_i = B'_i$ \textbf{then} $i:=i+1$
  \item \textbf{endwhile}
  \item \textbf{return} $(\mathbf{X} = (X_1, \ldots, X_n), \mathbf{Y} = (Y_1, \ldots, Y_n))$
\end{enumerate}
    }%
}
\caption{Entanglement-based raw key distribution \label{fig:EntanglementBasedRawKeyDistribution}}
\end{figure}

As first shown in \textcite{BBM92}, these two protocols, $\pi_{\mathrm{practical}}$ and $\pi_{\mathrm{theoretical}}$, are equivalent in terms of their security.\footnote{This statement is only valid in the device-dependent setting, but does not extend to device-independent security proofs [see \textcite{ER14}]. For full device-independent security, it is necessary to distribute entanglement.} Note first that Bob's part of the protocol is obviously the same for $\pi_{\mathrm{practical}}$ and $\pi_{\mathrm{theoretical}}$. To see the correspondence of Alice's part, consider the two bits $B_i$ and $X_i$ together with the qubit $Q_i$ generated by Alice in any round~$i$. It is straightforward to verify that, for both $\pi_{\mathrm{practical}}$ and $\pi_{\mathrm{theoretical}}$, these are described by the same ccq-state of the form
\begin{align}
  \rho_{B_i X_i Q_i}
  = \frac{1}{4} \sum_{b=0}^1 \sum_{x=0}^1 \proj{b} \otimes \proj{x} \otimes \proj{\phi_{x,b}} \ .
\end{align}
This shows in particular that, from the viewpoint of an adversary, who may have access to the quantum channel and hence to~$Q_i$, the two protocols are equivalent. 

The entanglement-based protocol $\pi_{\mathrm{theoretical}}$ described above may be further modified to make it even more suitable for security proofs. One such modification concerns the timing of the steps. Instead of running through $n$ rounds,  in each of which an entangled qubit pair is created and the qubits measured, one may instead consider a first step in which $n$ entangled qubit pairs $(\bar{Q}_i, Q_i)$ are distributed between Alice and Bob and, rather than being measured directly, first stored in quantum memories. Only in a second step Alice and Bob choose bases $B_i = B'_i$ for each of their qubit pairs and measure them accordingly. This is shown in \figref{fig:MemoryEntanglementBasedRawKeyDistribution}. An argument similar to the one above shows that this change has no impact on the security of the protocol. 

\begin{figure}[h]
\noindent\fbox{%
    \parbox{\columnwidth}{%
\begin{enumerate}
  \item \textbf{protocol} $\mathrm{RawKeyDistribution}()$ [with postponed measurement]
  \item \textbf{parameters} 
     $n$ [number of signals];
     $\phi_{x, 0} := \ket{x}$, $\phi_{x,1} := \frac{1}{\sqrt{2}} (\ket{0} + (-1)^x \ket{1})$, for $x \in \{0,1\}$  [bases for encoding]
  \item \textbf{for} $i \in \{1, \ldots, n\}$ \textbf{do}
  \item Alice prepares qubits $(\bar{Q}_i, Q_i)$ in state $\smash{\frac{1}{\sqrt{2}} (\ket{0} \ket{0} + \ket{1} \ket{1})}$ and gives $Q_i$ as input to the quantum channel
  \item Bob stores the quantum channel output $Q'_i$
  \item \textbf{endfor}
    \item \textbf{for} $i \in \{1, \ldots, n\}$ \textbf{do}
      \item Alice chooses $B_i \in_R \{0,1\}$ and communicates $B_i$ to Bob over the classical channel
 \item Alice measures $\bar{Q}_i$ w.r.t.\ basis $\{\phi_{0, B'_i}, \phi_{1, B_i}\}$ to get $X_i$     
  \item Bob measures $Q'_i$ w.r.t.\ basis $\{\phi_{0, B'_i}, \phi_{1, B_i}\}$ to get $Y_i$
  \item \textbf{endfor}
  \item \textbf{return} $(\mathbf{X} = (X_1, \ldots, X_n), \mathbf{Y} = (Y_1, \ldots, Y_n))$
\end{enumerate}
    }%
}
\caption{Entanglement-based raw key distribution with postponed measurement \label{fig:MemoryEntanglementBasedRawKeyDistribution}}
\end{figure}

\subsection{Raw key distribution and parameter estimation} \label{sec:RawKeyDistribution}

The first part of the security proof concerns the raw key distribution and the parameter estimation step. For raw key distribution we consider the particular subprotocol described in \figref{fig:MemoryEntanglementBasedRawKeyDistribution}. Parameter estimation is shown in \figref{fig:ParameterEstimation}. It essentially calculates an estimate for the fraction~$\eta$ of positions~$i$ in which the bit strings $\mathbf{X}$ and $\mathbf{Y}$ differ, i.e., $|X_i - Y_i| = 1$, and returns the value ``$\mathrm{fail}$'' if this fraction exceeds a given threshold~$\eta_0$. 

\begin{figure}[h]
\noindent\fbox{%
    \parbox{\columnwidth}{%
\begin{enumerate}
  \item \textbf{protocol} $\mathrm{ParameterEstimation}(\mathbf{X}, \mathbf{Y})$
  \item \textbf{parameters} $s$ [sample size]; $\eta_0$ [threshold]
  \item Alice chooses a subset $S \subset_R \{1, \ldots, n\}$, \\ with $n=|\mathbf{X}|$ [length of $\mathbf{X}$] and $s=|S|$ [size of $S$]
  \item Alice communicates $\{(i, X_i): \, i \in S\}$ over the classical channel
  \item Bob computes $\eta = \frac{1}{s} \sum_{i \in S} |X_i - Y_i|$
  \item \textbf{if} $\eta \leq \eta_0$ \textbf{then} \textbf{return} $\mathrm{ok}$ \textbf{else} \textbf{return} $\mathrm{fail}$
\end{enumerate}
    }%
}
\caption{Parameter estimation \label{fig:ParameterEstimation}}
\end{figure}

To run the raw key distribution and parameter estimation protocol, one needs as initial resources   an insecure quantum channel $\aQ$ together with an authentic classical channel $\aA$, as shown for example in \figref{fig:qkd.real}. The target is a \emph{raw key} resource $\aR$, which can be understood as a weak version of a shared secret key resource as shown in \figref{fig:qkd.resource.switch}. The resource~$\aR$ is equipped with a switch controlled by Eve~\cite{Portmann2017}. If the switch is in position~$1$, the resource merely outputs $\perp$ to Alice and Bob. If the switch is in position~$0$, the resource outputs bit strings $\mathbf{X}$ and $\mathbf{Y}$ of length~$n$ to Alice and Bob, but at the same time enables Eve to interact with the resource, allowing her to gain information~$E$. The latter is bounded by a secrecy condition, which may be expressed in terms of a lower bound~$t$ on the \emph{smooth min-entropy}~\cite{Ren05} of  Alice's output $\mathbf{X}$ conditioned on~$E$,
\begin{align} \label{eq:Hminboundgeneral}
  H_{\min}^\varepsilon(\mathbf{X} | E) \geq t \ .
\end{align}
Here $\varepsilon > 0$ is a small parameter that will contribute additively to the failure probability of the protocol. The choice of this particular measure for entropy will be relevant for the further proof steps below, especially privacy amplification. Intuitively, one may think of $H_{\min}^\varepsilon(\mathbf{X} | E)$ as the minimum number of bits that can be extracted from $\mathbf{X}$ that are uniform and uncorrelated to~$E$, except with probability~$\varepsilon$. 

The desired statement is that running the raw key distribution protocol followed by the parameter estimation protocol on $\aQ$ and $\aA$ constructs the raw key resource $\aR$  for appropriately chosen parameters. One may view this as the core of security proofs in QKD. It shows that a criterion on the statistics of the  data $\mathbf{X}$ and $\mathbf{Y}$ measured by Alice and Bob, as tested by the parameter estimation protocol, is sufficient to imply a certain level of secrecy of $\mathbf{X}$ towards Eve. 

To illustrate the idea behind the argument, let us for the moment focus on collective attacks (see \secref{sec:qkd.other.models}). Under this assumption, each of the qubit pairs $(\bar{Q}_i, Q'_i)$ held by Alice and Bob when they execute the raw key distribution protocol of \figref{fig:MemoryEntanglementBasedRawKeyDistribution}, prior to the measurement, is in the same state $\rho_{\bar{Q}_i, Q'_i}$. Recall, however, that the second qubit, $Q'_i$, is what Bob received. Since Eve may corrupt the quantum communication channel, it is not guaranteed that this qubit coincides with the qubit $Q_i$ that Alice sent. The state $\rho_{\bar{Q}_i, Q'_i}$ may thus be different from the entangled state $\smash{\frac{1}{\sqrt{2}} (\ket{0} \ket{0} + \ket{1} \ket{1})}$ that Alice prepared. 

To gain some intuition, it may be useful to consider the special case where the threshold in the subprotocol for parameter estimation is small, say even $\eta_0 = 0$. If the subprotocol returns the value ``$\mathrm{ok}$'' then this means that the bit strings $\mathbf{X}$ and $\mathbf{Y}$ largely coincice. This yields a constraint on the state  $\rho_{\bar{Q}_i, Q'_i}$, namely that if both Alice and Bob measure it  with respect to the basis $\{\ket{0}, \ket{1}\}$ or with respect to the basis $\{\frac{1}{\sqrt{2}} (\ket{0} \pm \ket{1} )\}$ they obtain identical outcomes, except with some small probability that is due to the finite sample size used for parameter estimation. 

It is now straightforward to verify that the only states $\rho_{\bar{Q}_i Q'_i}$ that can pass the test with $\eta_0 = 0$ are those that are close to the pure state $\smash{\frac{1}{\sqrt{2}} (\ket{0} \ket{0} + \ket{1} \ket{1})}$ that Alice prepared. Next, one may consider the joint state $\rho_{\bar{Q}_i Q'_i E}$ that includes Eve. But because the state of the first two qubits is almost pure, one can conclude that this state must be of the form
\begin{align}
  \rho_{\bar{Q}_i Q'_i E} \approx \rho_{\bar{Q}_i Q'_i} \otimes \rho_E \ .
\end{align}
That is, Eve's information $E$ is almost uncorrelated to $\bar{Q}_i$ and $Q'_i$. But because each of the $n$ bits $X_i$ of $\mathbf{X}$ is obtained from a measurement of $\bar{Q}_i$, it is as well almost uncorrelated to $E$. This proves that each bit $X_i$ is almost uniformly random and independent of $E$. The smooth min-entropy of the entire sequence $\mathbf{X}$ of bits is thus almost maximal, i.e., $H_{\min}^\varepsilon(\mathbf{X} | E) \approx n$. 

If, instead of $\eta_0 = 0$, one inserts an arbitrary value for the tolerated noise tolerance $\eta_0$, which is also known as the \emph{Quantum Bit Error Rate (QBER)}, a refinement of the argument we just sketched gives~\cite{RGK05,Ren05}
\begin{align} \label{eq_HminBB84}
  H_{\min}^{\varepsilon}(\mathbf{X}|E) \geq n (1-h(\eta_0)) + O(\sqrt{n})
\end{align}
where $h(x) = -x\log_2(x) -(1-x)\log_2(1-x)$ denotes the binary entropy function. 

We also note that the argument can be adapted to the case of device-independent security. In this case  the parameter estimation tests whether the outcome statistics of Alice and Bob violates a Bell inequality. The lower bound on the entropy then depends on the degree of this violation; see \textcite{ABGMPS07} for the example of the CHSH Bell inequality~\cite{CHSH69}.

The assumption of collective attacks is necessary    to sensibly talk about the state $\rho_{\bar{Q}_i Q'_i}$ of the individual systems. However, there are no good reasons why an adversary should be restricted to such attacks (see~\secref{sec:attacks:assumptionlist}). Modern security proofs therefore usually consist of an additional step, in which it is shown that general attacks cannot be more powerful than collective attacks. 

There exist various techniques to achieve this. The most widely one used so far is based on the \emph{exponential de Finetti theorem}~\cite{Ren05,Ren07,RC09}. The theorem states that, if a state over many subsystems, such as $\rho_{\bar{Q}_1 Q'_1 \cdots \bar{Q}_n Q'_n}$, is symmetric under reorderings, i.e., the state remains the same if one permutes the subsystems $\bar{Q}_i Q'_i$, then it is well approximated by a mixture of i.i.d.\ states, i.e., states of the form $\rho_{\bar{Q}_1 Q'_1} \otimes \cdots \otimes \rho_{\bar{Q}_n Q'_n}$. The latter corresponds to the structure one has if one assumes collective attacks. 

To apply the exponential de Finetti theorem, it is sufficient to argue that the rounds of the protocol, in which the individual signals are sent, could be reordered arbitrarily. Like in the example of the BB84 protocol described above, this is the case for most protocols that have been proposed in the literature. A notable exception are the \emph{Coherent One-Way (COW)} protocol~\cite{SBGSZ05} and the \emph{Differential Phase Shift (DPS)} protocol~\cite{IWY02}, where information is encoded in the correlations between signals. 

Another method, which is related to the de Finetti theorem, is the  \emph{post-selection technique}~\cite{CKR09}. Like the former, it can be used to lift security proofs against collective attacks to security proofs against general attacks, provided that the protocol satisfies the symmetry assumptions described above. 

Under certain conditions, it is also possible to establish bounds of the form of \eqnref{eq_HminBB84} directly for general attacks, i.e., without first restricting to collective attacks. This is the case for the approaches presented in~\textcite{CRE04} and in~\textcite{RGK05}, which are both applicable to the device\-/dependent setting, as well as the techniques proposed in~\textcite{TR11,TLGR12}, which include semi\-/device\-/independent scenarios, and in~\textcite{RUV13,VV14,MS14}, which applies to particular device\-/independent protocols. 

The most recent approach to directly prove security against general attacks relies on the  \emph{Entropy Accumulation Theorem (EAT)}~\cite{DFR20}. This approach, in contrast to methods based on the de Finetti theorem, gives rather tight min-entropy bounds even when the number~$n$ of protocol rounds is relatively small. It is furthermore applicable to the semi\-/device\-/independent and the device-independent setting~\cite{AFRV19}, which will be discussed in~\secref{sec:alternative}.

\subsection{Information reconciliation} \label{sec_infrec}

The goal of information reconciliation is to ensure that Alice and Bob have the same (raw) key. The most common way to achieve this is to regard Alice's bit string $\mathbf{X}$ as the key, and to let Bob infer this key from the information $\mathbf{Y}$ he has. To this end, Alice sends partial information about $\mathbf{X}$ to Bob over the classical channel. 

\begin{figure}[h]
\noindent\fbox{%
    \parbox{\columnwidth}{%
\begin{enumerate}
  \item \textbf{protocol} $\mathrm{InformationReconciliation}(\mathbf{X}, \mathbf{Y})$
  \item \textbf{parameters} $\mathrm{enc}, \mathrm{dec}$ [coding scheme]
  \item Alice sends $C = \mathrm{enc}(\mathbf{X})$ over the classical channel
  \item Bob computes $\mathbf{X'} = \mathrm{dec}(C, \mathbf{Y})$
\item \textbf{return} $(\mathbf{X}, \mathbf{X'})$
\end{enumerate}
    }%
}
\caption{Information reconciliation \label{fig:InformationReconciliation}}
\end{figure}

The protocol shown in \figref{fig:InformationReconciliation} uses as resources a raw key $\aR$, as described in the previous section, as well as, again, an authentic classical communication channel~$\aA$. Its purpose is to generate a weak key resource $\aR'$,  which provides a guarantee of the form of \eqnref{eq:Hminboundgeneral} on the secrecy of the key, and, in addition, ensures that  Alice and Bob's values, $\mathbf{X}$ and $\mathbf{X'}$, are identical. 

We note that information reconciliation is a purely classical subprotocol. It is also largely independent of the other parts of the QKD protocol, and hence  works in both the device-dependent and the device-independent setting. The choice of the \emph{coding scheme}, i.e., the functions $\mathrm{enc}$ and $\mathrm{dec}$ that the protocol invokes,  merely depends on the \emph{noise model}. The latter describes how  Alice and Bob's inputs to the protocol, $\mathbf{X}$ and $\mathbf{Y}$, are correlated with each other. 

The noise model is most generally specified in terms of a joint probability distribution of $\mathbf{X}$ and $\mathbf{Y}$. The coding scheme must then be chosen such that
\begin{align} \label{eq:decodingsuccess}
  \Pr\bigl[ \mathrm{dec}(\mathrm{enc}(\mathbf{X},\mathbf{Y}) = \mathbf{X}  \bigr] \geq 1-\varepsilon
\end{align}
The parameter $\varepsilon>0$  bounds the failure probability of the subprotocol and will hence, similarly to the parameter  $\varepsilon$ used in the step above,  contribute additively to the total failure probability of the QKD protocol. Furthermore, to maintain as much secrecy as possible for $\mathbf{X}$, the function $\mathrm{enc}$ should be chosen such that $C=\mathrm{enc}(\mathbf{X})$ does not reveal too much information about $\mathbf{X}$.  (Recall that the classical channel is accessible to Eve, so she may get hold of~$C$.) This may be achieved by making $C$ as small as possible. It can be shown using classical techniques from information theory that any coding scheme that satisfies \eqnref{eq:decodingsuccess} requires a communication~$C$ of
\begin{align} \label{eq:communicationbound}
  k \geq H^{\varepsilon}_{\max}(\mathbf{X} | \mathbf{Y}) \ ,
\end{align}
bits, where $H^{\varepsilon}_{\max}$ denotes the smooth max-entropy~\cite{RW05}. Furthermore, there exist coding schemes that saturate this bound (up to a small additive constant).  

In the case of an i.i.d.\ noise model, $H^{\varepsilon}_{\max}(\mathbf{X} | \mathbf{Y})$ is approximated by the Shannon entropy, up to terms of order $\sqrt{n}$, where $n$ is the length of $\mathbf{X}$. For a protocol such as BB84, which uses single qubits, and assuming that the QBER is $\eta_0$, one thus has
\begin{align}
  k \approx n h(\eta_0) + O(\sqrt{n}) \ .
\end{align}

Letting $E$ be the initial information that Eve has about the raw key $\mathbf{X}$ before information reconciliation, the secrecy after information reconciliation with communication~$C$ consisting of $k$ bits is given by
\begin{align} \label{eq:secreduction}
  H_{\min}^{\varepsilon}(\mathbf{X} | E C) \gtrapprox H_{\min}^{\varepsilon}(\mathbf{X} | E) - k - O(1) \ .
\end{align}
Hence, for an optimal information reconciliation protocol, we have 
\begin{align}
  H_{\min}^{\varepsilon}(\mathbf{X} | E C) \gtrapprox H_{\min}^{\varepsilon}(\mathbf{X} | E) - H_{\max}(\mathbf{X} | \mathbf{X'}) - O(\sqrt{n}) \ .
\end{align}
In particular, for the case of the BB84 protocol, we get 
\begin{align}
  H_{\min}^\varepsilon(\mathbf{X} | E C)  \geq n (1- 2 h(\eta_0)) - O(\sqrt{n}) \ .
\end{align}

As is clear from \eqnref{eq:secreduction}, the amount of secrecy that is left after information reconciliation  depends on the amount~$k$ of  communication required. The design of coding schemes  $(\mathrm{enc}$, $\mathrm{dec})$ that optimise this parameter is a main subject of classical information theory~\cite{CT12}. While the bound in \eqnref{eq:communicationbound} can already be saturated with randomly constructed encoding functions, a main challenge is to develop schemes for which the encoding and decoding functions are efficiently computable~\cite{LABZG08,ELAB09,EMM11,JK14}.

While the information reconciliation protocol of \figref{fig:InformationReconciliation} invokes only one-way communication from Alice to Bob, one may also consider two-way schemes. In fact, the first proposals for QKD implementations used a procedure to correct errors that required multiple rounds of communication between Alice and Bob~\cite{BBBSS92}.\footnote{Despite its two-way nature, the particular method proposed in~\textcite{BBBSS92} did not achieve the information-theoretic bounds described above. It was only realised later in~\textcite{BBCS92}, in the context of oblivious transfer, that one-way error correction is sufficient and can be made (asymptotically) optimal.} Furthermore, one may also include \emph{advantage distillation}~\cite{Mau93}. Here the idea is that Alice and Bob group their data into small blocks. They then try to distinguish blocks that are likely to contain few errors from those that are likely to contain many errors. The ones with many errors are then discarded. It has been shown that this technique can be advantageous compared to standard error correction~\cite{GL03,Ren05,TLR20}.

\subsection{Privacy amplification} \label{sec_pa}

The aim of privacy amplification is to turn the weakly secret key~$\mathbf{X}$, which after information reconciliation is known to Alice and Bob, into  a strong secret key $K$, i.e., a bit string that is essentially uniform and independent of the information held by an adversary~\cite{BBR88,BBCM95}. This is typically achieved with a protocol as in  \figref{fig:PrivacyAmplification}. Apart from the weak key resource $\aR$, which satisfies a secrecy bound of the form of \eqnref{eq:Hminboundgeneral} and which is assumed to output the same string~$\mathbf{X}$ to Alice and Bob,  the protocol requires an authentic communication channel~$\aA$. From these resources, the protocol constructs a secret key resource as shown in \figref{fig:qkd.resource.switch}.

\begin{figure}[h]
\noindent\fbox{%
    \parbox{\columnwidth}{%
\begin{enumerate}
  \item \textbf{protocol} $\mathrm{PrivacyAmplification}(\mathbf{X}, \mathbf{X'})$
  \item \textbf{parameters} $\{\mathrm{ext}_s\}_{s \in \cS}$ [randomness extractor]
  \item Alice chooses $S \in_R \cS$ and sends it over the classical channel
  \item Alice computes $K = \mathrm{ext}_S(\mathbf{X})$ 
  \item Bob computes $K' = \mathrm{ext}_S(\mathbf{X'})$ 
  \item \textbf{return} $(K, K')$
\end{enumerate}
    }%
}
\caption{Privacy amplification \label{fig:PrivacyAmplification}}
\end{figure}

The protocol makes use of a \emph{randomness extractor}~\cite{Zuc90,Shaltiel04}. This is a family of functions $\mathrm{ext}_s$ parameterised by a \emph{seed} $s \in \cS$, which take as input a bit string, such as $\mathbf{X}$, and output a bit string of a fixed length~$\ell$. In the classical literature, a \emph{strong $(k, \varepsilon)$-extractor} is defined by the property that, for any input $\mathbf{X}$ whose min-entropy satisfies the lower bound $H_{\min}(\mathbf{X}) \geq k$, the output $\mathrm{ext}_s(\mathbf{X})$ is $\varepsilon$-close to uniform. More precisely, the expectation over a randomly chosen seed $s \in \cS$ of the variational distance between the distribution of the output $\mathrm{ext}(\mathbf{X})$ and a uniform string $U$ of $\ell$ bits must be upper bounded by $\varepsilon$,
\begin{align}
  \mathrm{Exp}_s\bigl[D(P_{\mathrm{ext}_s(\mathbf{X})}, P_{U})\bigr] \leq \varepsilon \ .
\end{align}
This definition does however not take into account the quantum nature of information that an adversary may have about $\mathbf{X}$~\cite{KMR05,GKKRD07}. It is hence not sufficient for use in the context of quantum key distribution, unless one restricts to security against individual attacks, which corresponds to forcing the adversary to store classical information only (see \secref{sec:attacks:assumptionlist}). 

To be able to prove general security, it is necessary to demand that the randomness extractor $\{\mathrm{ext}_s\}_{s \in \cS}$ be \emph{quantum-proof}, for parameters $k$ and $\varepsilon$ as above. This means that, for any $\mathbf{X}$ and any quantum system $E$ such that $H_{\min}(\mathbf{X} | E) \geq k$ one has
\begin{align}
 \mathrm{Exp}_s\bigl[D(\rho_{\mathrm{ext}_s(\mathbf{X}) E}, \rho_{U} \otimes \rho_E)\bigr] \leq \varepsilon \ .
\end{align}
Note that this criterion refers to min-entropy $H_{\min}(\mathbf{X} | E) = H_{\min}^{\varepsilon'}(\mathbf{X} | E)$ with smoothness parameter $\varepsilon' = 0$. However, a straightforward application of the triangle inequality for the distance between states implies that a corresponding criterion  also holds if $\varepsilon'>0$~\cite{Ren05}. 

A number of constructions for quantum-proof extractors have been proposed in the literature~\cite{RK05,Ren05,FS08,KT08,DPVR12,BT12,MPS12,BFS17}. In the context of QKD, the most widely used extractors are based on two-universal hashing~\cite{CW79,WC81}. As shown in~\textcite{RK05,Ren05,TSSR10}, these  can achieve an output length of  $\ell = k -  2 \log_2(1/\varepsilon)$ while still being quantum-proof $(k, \varepsilon)$ extractors. Using them within the protocol of \figref{fig:PrivacyAmplification}, it generates a key of length 
\begin{align}
   \ell = H_{\min}^\varepsilon(\mathbf{X} | E C) - O(1) \ ,
 \end{align}
 with a failure probability of the order~$\varepsilon$.
Combining this with the results of the previous sections, with optimal information reconciliation and privacy amplification, it is possible to generate a key of length
\begin{align}
  \ell = H_{\min}^\varepsilon(\mathbf{X} | E) - H_{\max}(\mathbf{X} | \mathbf{Y}) - O(1) \ .
\end{align}
In particular, in the case of the BB84 protocol, we obtain
\begin{align}
  \ell = n (1-2 h(\eta_0)) - O(\sqrt{n})
\end{align}
where $\eta_0$ is the QBER. The asymptotic key rate is thus $1-2 h(\eta_0)$. 

\subsection{Other approaches to prove security} \label{sec_othersecurityproofs}

The generic security proof described above follows the approach proposed in \textcite{Ren05}. It is sometimes termed ``information-theoretic'', as its core part consists of bounds on entropic quantities, such as \eqnref{eq_HminBB84}. Such bounds have first been proposed in \textcite{DW05}. They were further developed in \textcite{RenRen12} and used in \textcite{PhysRevLett.95.080501,RGK05}; see also \textcite{christandl2007unifying} for related work. However, as already mentioned, there exist a variety of other proof strategies. 

Early proofs~\cite{LoChau99,SP00} used a reduction to the problem of entanglement distillation. For this, one rearranges the key distribution protocol such that all measurements are postponed to the very last step. If one now omits these final measurements, Alice and Bob end up with correlated quantum registers rather than classical keys. One may thus regard the protocol as an entanglement distillation protocol~\cite{Bennett96,Benett96b} and prove that the registers held by Alice and Bob are almost maximally entangled. If this is the case then, by the monogamy of entanglement, the information in these registers is uncorrelated to Eve, and hence secret.\footnote{The following is a quantitative version of this statement. If the entanglement distilled by Alice and Bob has fidelity~$F$ to a maximally entangled state then it follows from Theorem~1 of \textcite{FuchsvanGraaf} that the corresponding security parameter~$\varepsilon$ according to~\eqnref{eq:qkd.security} is bounded by $\epsilon \leq \sqrt{1-F^2}$.}

This approach may be more generally understood as follows. Assuming that Alice and Bob's start with quantum correlation stored in individual qubits equipped with a computational basis, the entanglement distillation protocol can be regarded as a quantum error correction scheme~\cite{CS96,Steane96} that corrects both for bit and phase flip errors. The correction of bit flip errors ensures that Alice and Bob end up with the same key. The correction of phase flip errors ensures that the two registers are not only classically correlated but maximally entangled. Since, as indicated above, the latter implies secrecy, one can understand the correction of phase flip errors as a kind of privacy amplification~\cite{Renes2013}.

The technique has  been used originally to prove the security of the BB84 protocol, including variants with imperfect devices~\cite{GLLP04}, but can also be applied to other quantum key distribution protocols~\cite{Tamaki03,Koashi04,Boileau}. While the correspondence to entanglement distillation requires that error correction and privacy amplification be treated as a single quantum error correction step, it is under certain conditions possible to achieve a separation in a way similar to the modular description above~\cite{Lo03}. Furthermore, as shown in \textcite{HHHLO08}, the method also works if the registers of Alice and Bob merely contain \emph{bound entanglement}, i.e., entanglement from which no maximally entangled states can be distilled~\cite{HHH98}. 

A somewhat related strategy, proposed originally in \textcite{May01}, is the use of complementarity~\cite{Koa09}. Specifically, one uses the fact that if Alice is able to accurately predict the outcomes of a measurement in one basis, say the computational basis, then by the uncertainty principle any predictions for the outcomes of measurements in a complementary basis, e.g., the diagonal basis in the case of single qubits, will be inaccurate.  This technique has been refined in a series of works and made applicable to the study of finite-size effects~\cite{HT12,TLGR12} and to measurement-device independent cryptography~\cite{TamakiLo12}. The complementarity approach is also related to the use of  \emph{entropic uncertainty relations}~\cite{Berta10,Coles,TR11}.


\section{Alternative modeling of QKD}
\label{sec:alternative}

So far we discussed QKD as protocols that start with an insecure
quantum channel and an authentic classical channel and generate, as
the desired ideal resource, a key of fixed length. In this section we
discuss other variants of QKD protocols, where these resources are
chosen differently. In \secref{sec:alternative.adaptive} we
consider an ideal key resource with adaptive key length. In
\secref{sec:alternative.entanglement} we discuss protocols which
use a source of entanglement instead of an insecure quantum
channel. In \secref{sec:alternative.randomness} we show how to
model a situation in which no perfect randomness is available.  In
\secref{sec:alternative.di} we model device\-/independent
QKD. Relaxations of this known as semi\-/device\-/independence are
discussed in \secref{sec:alternative.semi}. Finally, in
\secref{sec:alternative.memoryless} we consider adversaries that have
no quantum memory.

\subsection{Adaptive key length}
\label{sec:alternative.adaptive}

For a protocol to construct the shared secret key resource of \figref{fig:qkd.resource.switch}, it must either abort or produce a key of a fixed length. A more practical protocol could adapt the secret key length to the noise level of the quantum channel. This provides the adversary with the functionality to control the key length (not only whether it gets generated or not), and can be modeled by allowing the key length to be input at Eve's interface of the ideal key resource, as illustrated in \figref{fig:qkd.resource.adaptive}.

\begin{figure}[tb]

\begin{tikzpicture}[
sArrow/.style={->,>=stealth,thick},
largeResource/.style={draw,thick,minimum width=1.618*2cm,minimum height=2cm}]

\small

\def\u{.236} 

\node[largeResource] (keyBox) at (0,0) {};
\node (alice) at (-2.5,\u) {Alice};
\node (bob) at (2.5,\u) {Bob};
\node (eve) at (0,-1.7) {Eve};
\node[draw] (key) at (0,-.5) {key};

\draw[sArrow,<->] (alice) to node[pos=.22,auto] {$k(m)$} node[pos=.78,auto] {$k(m)$} (bob);
\draw[thick] (0,\u) to (key);
\draw[sArrow] (eve) to node[pos=.3,auto] {$m$} (key);

\end{tikzpicture}

\caption[Secret key resource with adaptive
length]{\label{fig:qkd.resource.adaptive}A secret key resource with   adaptive key length. This resources allows Eve to choose the length $m$   of the final key $k$, which is then output at Alice's and Bob's interfaces.}
\end{figure}
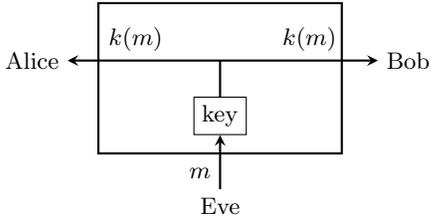

Such an ideal resource has been considered in \textcite{BHLMO05,HT12}. The reduction from the corresponding security definition in AC to a trace distance criterion still goes through. But instead of \eqnref{eq:d}, we get 
\begin{equation} \label{eq:adpative.trace.d}
\sum_m p_m D \left( \rho^m_{KE},\tau^m_K \otimes \rho^m_E \right) \leq
\eps,
\end{equation}
where $p_m$ is the probability of obtaining a key of length $m$, $\rho^m_{KE}$ is the joint state of the key and Eve's system conditioned on the key having length $m$,  and $\tau^m_K$ is a fully mixed state of dimension $2^m$.

\subsection{Source of entanglement}
\label{sec:alternative.entanglement}

In contrast to \emph{prepare\-/and\-/measure} protocols,
\emph{entanglement\-/based} protocols, e.g., \textcite{Eke91,BBM92},
use a source of entanglement, instead of a quantum communication
channel. It is also pretty standard in security proofs to first
transform a given prepare\-/and\-/measure protocol into an
entanglement\-/based one, and then prove the security of the
latter~\cite{SP00}. In \figref{fig:qkd.real.ent} we draw the system
consisting of a QKD protocol $(\pi^{\qkd}_A,\pi^{\qkd}_B)$, the
authentic channel $\aA$ and a source  $\aE$ of entangled states, which may be
controlled by Eve. To specify the completeness property, we also consider
a source of entanglement $\aE'$ that produces a fixed bipartite
entangled state instead of allowing Eve to decide.

\begin{figure}[tb]

\begin{tikzpicture}[
sArrow/.style={->,>=stealth,thick},
thinResource/.style={draw,thick,minimum width=2.4cm,minimum height=1cm},
protocol/.style={draw,rounded corners,thick,minimum width=1.2cm,minimum height=2.5cm},
pnode/.style={minimum width=.8cm,minimum height=.5cm}]

\small

\def\t{4} 
\def\u{2.8} 
\def\v{.75}
\def\w{.6} 

\node[pnode] (a1) at (-\u,\v) {};
\node[pnode] (a2) at (-\u,0) {};
\node[pnode] (a3) at (-\u,-\v) {};
\node[protocol] (a) at (-\u,0) {};
\node[yshift=-2,above right] at (a.north west) {\footnotesize
  $\pi^{\qkd}_A$};
\node (alice) at (-\t,0) {};

\node[pnode] (b1) at (\u,\v) {};
\node[pnode] (b2) at (\u,0) {};
\node[pnode] (b3) at (\u,-\v) {};
\node[protocol] (b) at (\u,0) {};
\node[yshift=-2,above right] at (b.north west) {\footnotesize $\pi^{\qkd}_B$};
\node (bob) at (\t,0) {};

\node[thinResource] (cch) at (\w,\v) {};
\node[yshift=-2,above right] at (cch.north west) {\footnotesize
  Authentic ch.~$\aA$};
\node[thinResource] (qch) at (-\w,-\v) {};
\node[yshift=-1.5,above right] at (qch.north west) {\footnotesize
  Source of states $\aE$};
\node (eveq1) at (-\w-.4,-1.75) {};
\node (junc1) at (eveq1 |- a3) {};
\node (eveq2) at (-\w+.4,-1.75) {};
\node (junc2) at (eveq2 |- a3) {};
\node (evec) at (\w+\w,-1.75) {};
\node (junc3) at (evec |- b1) {};

\draw[sArrow,<->] (a1) to node[auto,pos=.08] {$t$} node[auto,pos=.92] {$t$}  (b1);
\draw[sArrow] (junc3.center) to node[auto,pos=.9] {$t$} (evec.center);

\draw[sArrow] (a2) to node[auto,pos=.75,swap] {$k_{A},\bot$} (alice.center);
\draw[sArrow] (b2) to node[auto,pos=.75] {$k_{B},\bot$} (bob.center);

\draw[sArrow,<-] (a3) to (junc1.center) to node[pos=.8,auto,swap] {$\rho_A$} (eveq1.center);
\draw[sArrow] (eveq2.center) to node[pos=.2,auto,swap] {$\rho_B$} (junc2.center) to (b3);

\end{tikzpicture}

\caption[QKD system with source of
entanglement]{\label{fig:qkd.real.ent}A real QKD system that uses a
  source of entangled states. Instead of having access to an insecure
  channel as in \figref{fig:qkd.real.adv}, Alice and Bob use a source
  of entanglement $\aE$ that is controlled by Eve. This means that Eve may
  generate an arbitrary state $\rho_{ABE}$ of which the $A$ register goes to Alice and
  the $B$ register to Bob.}
\end{figure}
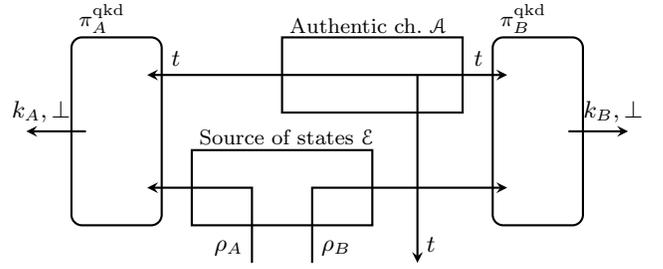

The reduction from the AC security definition to the trace distance criterion described in \secref{sec:security} works here, too, with the source of entanglement replacing the insecure channel, resulting in the same conditions for $\eps$\=/secrecy and $\eps$\=/correctness.

One can also show that any protocol designed for a distributed source
of entanglement can be transformed into one where a state is prepared
locally and sent over an (insecure) channel. To explain this, we first
decompose Alice's QKD protocol in two parts.  In the first she carries
out a subprotocol $\alpha$ that performs a measurement
$\bM^a = \{M^a_x\}_x$ on the state received from the source of
entangled states, where $\bM^a$ is chosen with some probability $p_a$
from a set $\{\bM^a\}_a$. The second part consists of the rest of her
QKD protocol. We illustrate this in \figref{fig:qkd.ent.protocol}.

\begin{figure}[tb]

\begin{tikzpicture}[
sArrow/.style={->,>=stealth,thick},
thinResource/.style={draw,thick,minimum width=1.618*2cm,minimum height=1cm},
protocol/.style={draw,rounded corners,thick,minimum width=1.545cm,minimum height=2.5cm},
pnode/.style={minimum width=1cm,minimum height=.5cm},
sqResource/.style={draw,rounded corners,thick,minimum width=1cm,minimum height=1cm}]

\small

\def\t{5.92} 
\def\u{4.39} 
\def\um{2.5} 
\def\ub{2.37} 
\def\v{.75}

\node[pnode] (a1) at (-\u,\v) {};
\node[pnode] (a2) at (-\u,0) {};
\node[pnode] (a3) at (-\u,-\v) {};
\node[protocol] (a) at (-\u,0) {};
\node[yshift=-2,above right] at (a.north west) {\footnotesize
  $\pi^{\qkd}_A$};
\node (alice) at (-\t,0) {};

\node (b1) at (-.4,0 |- a1) {};
\node (b3) at (\ub,-\v) {};

\node[sqResource] (m) at (-\um,-\v) {};
\node[inner sep=1] (mInner) at (-\um,-\v) {$M^a_x$};
\node[yshift=-2,above right] at (m.north west) {\footnotesize
  $\alpha$};

\node[thinResource] (qch) at (0,-\v) {};
\node[yshift=-1.5,above right] at (qch.north west) {\footnotesize
  Source of states $\aE$};
\node (eveq1) at (-.4,-1.75) {};
\node (junc1) at (eveq1 |- a3) {};
\node (eveq2) at (.4,-1.75) {};
\node (junc2) at (eveq2 |- a3) {};

\draw[sArrow] (b1) to node[auto,pos=.85,swap] {$t$}  (a1);

\draw[sArrow] (a2) to node[auto,pos=.75,swap] {$k_{A},\bot$} (alice.center);

\draw[sArrow] (mInner) to node[swap,auto,pos=.48] {$a,x$} (a3);
\draw[sArrow,<-] (mInner) to (junc1.center) to node[pos=.8,auto,swap] {$\rho_A$} (eveq1.center);
\draw[sArrow] (eveq2.center) to node[pos=.264,auto,swap] {$\rho_B$} (junc2.center) to (b3);

\end{tikzpicture}

\caption[Entanglement based QKD
protocol]{\label{fig:qkd.ent.protocol}We split Alice's part of an
  entanglement\-/based QKD protocol in two parts, the measurement of
  the incoming states (denoted by $\alpha$) and the rest of the
  protocol (denoted by $\pi^{\qkd}_A$).}
\end{figure}
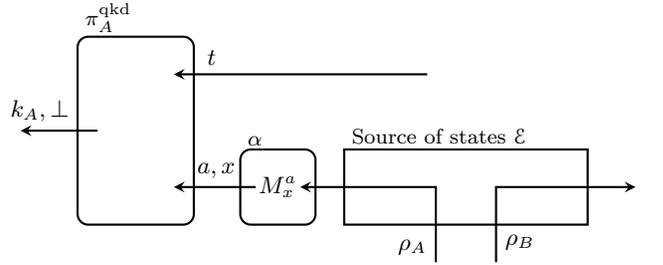

We now need to argue that there exists a converter $\gamma$ which constructs
$\alpha \aE$ from an insecure channel $\aQ$ and $\alpha \aE'$ from a
noiseless channel $\aQ'$. For this, we must
establish the two following conditions.
\begin{enumerate}[label=(\roman*), ref=\roman*]
\item \label{eq:ent.sec} There exists a simulator $\sigma_E$ such that
  \[ 
   \gamma \aQ = \alpha \aE \sigma_E.
\]
\item \label{eq:ent.cor} The following equality holds,
  \[\gamma \aQ' = \alpha \aE'.\] 
\end{enumerate}
Once we have established these conditions, it follows immediately from
the composition theorem of the AC framework~\cite{MR11} that any QKD
protocol which is sound when using $\alpha\aE$ and complete when using
$\alpha\aE'$ is also sound and complete when using $\gamma\aQ$ and
$\gamma\aQ'$, respectively.

Let $\rho_{AB}$ be the bipartite entangled state that is generated by
$\aE'$. Let
$\tilde{\varphi}^{x,a}_B \coloneqq \trace[A]{M^a_x \rho_{AB}
  \hconj{\left(M^a_x\right)}}$,
$p_{x|a} \coloneqq \tr \tilde{\varphi}^{x,a}_B$ and
$\varphi^{x,a}_B \coloneqq \tilde{\varphi}^{x,a}_B/p_{x|a}$. We define
the converter $\gamma$ to prepare the state $\varphi^{x,a}_B$ with
probability $p_ap_{x|a}$, which it sends on the insecure
channel. Furthermore, we define the simulator $\sigma_E$ to prepare
$\rho_{AB}$, input the $A$\=/part on the entanglement resource for
Alice and output the $B$\=/part at the outer interface. It is then
straightforward to check from \figref{fig:qkd.ent} that this satisfies
the conditions \eqref{eq:ent.sec} and \eqref{eq:ent.cor} described above.

\begin{figure*}[htb]
\centering
\subfloat[Soundness][\label{fig:qkd.ent.soundness}When modeling
  soundness, the adversary can modify the messages on the insecure
  channel $\aQ$. The simulator $\sigma_E$ generates the entangled state
  $\rho_{AB}$ that is expected from of a non\-/adversarial source of
  entangled states, and outputs the $B$ part at the outer interface,
  making the two systems on the left and right indistinguishable.]{
\begin{tikzpicture}[
sArrow/.style={->,>=stealth,thick},
thinResource/.style={draw,thick,minimum width=1.618*2cm,minimum height=1cm},
sqResource/.style={draw,rounded corners,thick,minimum width=1cm,minimum height=1cm}]

\small

\def\ua{3.85} 
\def\um{2.6} 
\def\ub{2.37} 
\def\t{2.5}

\node (a) at (-\ua,0) {};
\node (b) at (\ub,0) {};

\node[sqResource] (m) at (-\um,0) {};
\node[inner sep=1] (mInner) at (-\um,0) {$\varphi^{x,a}$};
\node[yshift=-2,above right] at (m.north west) {\footnotesize
  $\gamma$};

\node[thinResource] (qch) at (0,0) {};
\node[yshift=-1.5,above right] at (qch.north west) {\footnotesize
  Insecure channel $\aQ$};

\node (le) at (-.4,-\t) {};
\node (re) at (.4,-\t) {};

\node (junc1) at (le |- a) {};
\node (junc2) at (re |- a) {};

\draw[sArrow] (mInner) to node[swap,auto,pos=.6] {$a,x$} (a);

\draw[sArrow] (mInner) to node[auto,pos=.2] {$\varphi$} (junc1.center) to (le);
\draw[sArrow] (re) to (junc2.center) to (b);

\end{tikzpicture}  \hspace{2cm}
\begin{tikzpicture}[
sArrow/.style={->,>=stealth,thick},
thinResource/.style={draw,thick,minimum width=1.618*2cm,minimum height=1cm},
sqResource/.style={draw,rounded corners,thick,minimum width=1cm,minimum height=1cm}]

\small

\def\ua{3.85} 
\def\um{2.6} 
\def\ub{2.37} 
\def\v{.75}
\def\t{2.5}

\node (a) at (-\ua,0) {};
\node (b) at (\ub,0) {};

\node[sqResource] (m) at (-\um,0) {};
\node[inner sep=1] (mInner) at (-\um,0) {$M^a_x$};
\node[yshift=-2,above right] at (m.north west) {\footnotesize
  $\alpha$};

\node[thinResource] (qch) at (0,0) {};
\node[yshift=-1.5,above right] at (qch.north west) {\footnotesize
  Source of states $\aE$};

\node[sqResource] (sim) at (-.4,-2*\v) {};
\node[inner sep=5] (simInner) at (-.4,-2*\v) {$\rho_{AB}$};
\node[xshift=1,below left] at (sim.north west) {\footnotesize $\sigma_E$};

\node (le) at (-.4,-\t) {};
\node (re) at (.4,-\t) {};

\node (junc1) at (le |- a) {};
\node (junc2) at (re |- a) {};

\draw[sArrow] (mInner) to node[swap,auto,pos=.6] {$a,x$} (a);

\draw[sArrow,<-] (mInner) to (junc1.center) to (simInner);
\draw[sArrow] (simInner) to (le);
\draw[sArrow] (re) to (junc2.center) to (b);

\end{tikzpicture}}

\vspace{12pt}

\subfloat[Completeness][\label{fig:qkd.ent.completeness}When modeling
completeness, the source of entanglement $\aE'$ prepares the state
$\rho_{AB}$. The systems on the left and right are
indistinguishable.]{
\begin{tikzpicture}[
sArrow/.style={->,>=stealth,thick},
thinResource/.style={draw,thick,minimum width=1.618*2cm,minimum height=1cm},
sqResource/.style={draw,rounded corners,thick,minimum width=1cm,minimum height=1cm}]

\small

\def\ua{3.85} 
\def\um{2.6} 
\def\ub{2.37} 
\def\v{.75}

\node (a) at (-\ua,0) {};
\node (b) at (\ub,0) {};

\node[sqResource] (m) at (-\um,0) {};
\node[inner sep=1] (mInner) at (-\um,0) {$\varphi^{x,a}$};
\node[yshift=-2,above right] at (m.north west) {\footnotesize
  $\gamma$};

\node[thinResource] (qch) at (0,0) {};
\node[yshift=-1.5,above right] at (qch.north west) {\footnotesize
  Noiseless channel $\aQ'$};



\draw[sArrow] (mInner) to node[swap,auto,pos=.6] {$a,x$} (a);

\draw[sArrow] (mInner) to node[auto,pos=.93] {$\varphi$} (b);


\end{tikzpicture} \hspace{2cm}
\begin{tikzpicture}[
sArrow/.style={->,>=stealth,thick},
thinResource/.style={draw,thick,minimum width=1.618*2cm,minimum height=1cm},
sqResource/.style={draw,rounded corners,thick,minimum width=1cm,minimum height=1cm}]

\small

\def\ua{3.85} 
\def\um{2.6} 
\def\ub{2.37} 
\def\v{.75}

\node (a) at (-\ua,0) {};
\node (b) at (\ub,0) {};

\node[sqResource] (m) at (-\um,0) {};
\node[inner sep=1] (mInner) at (-\um,0) {$M^a_x$};
\node[yshift=-2,above right] at (m.north west) {\footnotesize
  $\alpha$};

\node[thinResource] (qch) at (0,0) {};
\node[yshift=-1.5,above right] at (qch.north west) {\footnotesize
  Source of states $\aE'$};

\node[draw] (boxState) at (0,0) {$\rho_{AB}$};


\draw[sArrow] (mInner) to node[swap,auto,pos=.6] {$a,x$} (a);

\draw[sArrow] (boxState) to (mInner);
\draw[sArrow] (boxState) to (b);


\end{tikzpicture}}

\caption[Using an entanglement protocol with an insecure
channel]{\label{fig:qkd.ent}Pictorial proof for the security of the
  construction of $\alpha\aE$ from $\aQ$ and $\alpha \aE'$ from
  $\aQ'$. Any protocol designed to run with a source of entangled
  states $\aE$ and which measures the incoming states on Alice's side
  as does $\alpha$ can be equivalently used with an insecure channel
  $\aQ$ and a converter $\gamma$ that generates the states to be sent
  on the channel.}
\end{figure*}
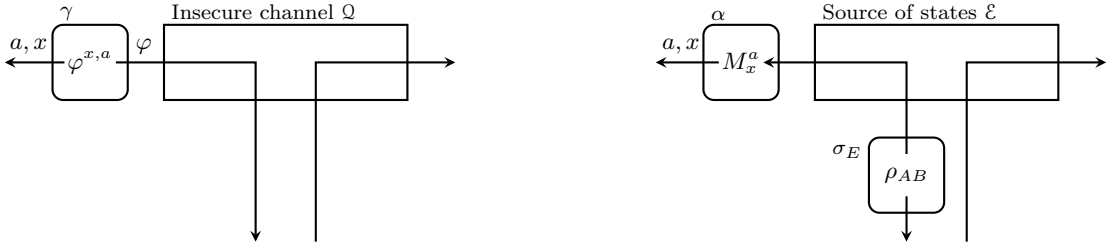
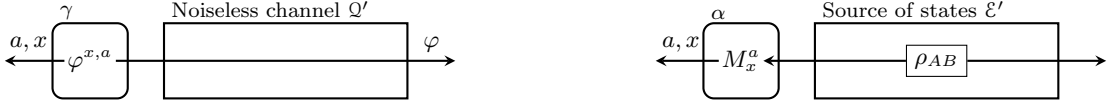

\subsection{Imperfect randomness}
\label{sec:alternative.randomness}

QKD protocols usually assume that the honest parties have (arbitrary) access to perfect random numbers. This is however never the case in practice. A more realistic model of a QKD system would consider randomness as a resource that is available in limited and imperfect quantities to Alice and Bob. The real QKD setting drawn in \figref{fig:qkd.real} needs to be changed to take this into account. In \figref{fig:qkd.real.randomness} we depict a QKD protocol that \--- additionally to the insecure quantum channel and authentic classical channel \--- has access to resources producing (local) randomness, $\aR_A$ and $\aR_B$, at Alice's and Bob's interfaces, respectively. A different model of randomness resources might also provide some partial (quantum) information about the randomness to the eavesdropper. For simplicity, however, we chose to draw the simpler case in which $\aR_A$ and $\aR_B$ have an empty interface for the dishonest party.

\begin{figure}[tb]

\begin{tikzpicture}[
sArrow/.style={->,>=stealth,thick},
thinResource/.style={draw,thick,minimum width=2.4cm,minimum height=1cm},
pnode/.style={minimum width=.8cm,minimum height=.5cm},
sqResource/.style={draw,thick,minimum width=1cm,minimum height=1cm},
longProtocol/.style={draw,rounded corners,thick,minimum width=1.2cm,minimum height=4cm}]

\small

\def\t{4.1} 
\def\u{2.9} 
\def\v{1.5}
\def\w{.6} 
\def\x{.8} 

\node[pnode] (a1) at (-\u,\v) {};
\node[pnode] (a2) at (-\u,0) {};
\node[pnode] (a3) at (-\u,-\v) {};
\node[longProtocol] (a) at (-\u,0) {};
\node[yshift=-2,above right] at (a.north west) {\footnotesize
  $\pi^{\qkd}_A$};
\node (alice) at (-\t,0) {};

\node[pnode] (b1) at (\u,\v) {};
\node[pnode] (b2) at (\u,0) {};
\node[pnode] (b3) at (\u,-\v) {};
\node[longProtocol] (b) at (\u,0) {};
\node[yshift=-2,above right] at (b.north west) {\footnotesize $\pi^{\qkd}_B$};
\node (bob) at (\t,0) {};

\node[sqResource] (ra) at (-\x-\w,\v) {};
\node[yshift=-2,above right] at (ra.north west) {\footnotesize $\aR_A$};
\node[sqResource] (rb) at (\x+\w,\v) {};
\node[yshift=-2,above right] at (rb.north west) {\footnotesize $\aR_B$};
\node[thinResource] (cch) at (\w,0) {};
\node[yshift=-2,above right] at (cch.north west) {\footnotesize
  Authentic ch.~$\aA$};
\node[thinResource] (qch) at (-\w,-\v) {};
\node[yshift=-1.5,above right] at (qch.north west) {\footnotesize
  Insecure ch.~$\aQ$};
\node (eveq1) at (-\w-.4,-\v-1) {};
\node (junc1) at (eveq1 |- a3) {};
\node (eveq2) at (-\w+.4,-\v-1) {};
\node (junc2) at (eveq2 |- a3) {};
\node (evec) at (\w+\w,-\v-1) {};
\node (junc3) at (evec |- b2) {};

\draw[sArrow] (ra.center) to node[auto,swap,pos=.65] {$r_A$} (a1);
\draw[sArrow] (rb.center) to node[auto,pos=.65] {$r_B$}  (b1);

\draw[sArrow,<->] (a2) to node[auto,pos=.08] {$t$} node[auto,pos=.92] {$t$}  (b2);
\draw[sArrow] (junc3.center) to node[auto,pos=.9] {$t$} (evec.center);

\draw[sArrow] (a2) to node[auto,pos=.75,swap] {$k_{A},\bot$} (alice.center);
\draw[sArrow] (b2) to node[auto,pos=.75] {$k_{B},\bot$} (bob.center);

\draw[sArrow] (a3) to (junc1.center) to node[pos=.8,auto,swap] {$\rho$} (eveq1.center);
\draw[sArrow] (eveq2.center) to node[pos=.264,auto,swap] {$\rho'$} (junc2.center) to (b3);

\end{tikzpicture}

\caption[QKD system with explicit randomness
resource]{\label{fig:qkd.real.randomness} A real QKD system with a
  deterministic protocol $(\pi^{\qkd}_A,\pi^{\qkd}_B)$ and explicit
  sources of randomness $\aR_A$ and $\aR_B$.}
\end{figure}
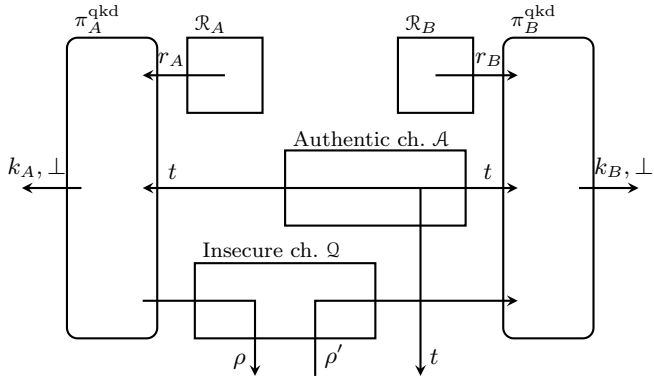

In such a setting, the converters $\pi^{\qkd}_A$ and $\pi^{\qkd}_B$
are deterministic systems. A QKD protocol would then construct an
ideal key resource given access to these three resources. It remains
an open problem to minimize the assumptions on the sources of
randomness in QKD. Recent results on device\-/independent randomness
amplification~\cite{CR12} show that under certain minimal assumptions\footnote{One
  generally has to assume that no messages leave or enter the quantum
  devices unless authorized by the protocol. Some papers make
  additional assumptions to simplify the protocols and proofs.} about
the workings of an unknown quantum system, one can transform a single
(public) weak source of randomness into a fully (private) random
source~\cite{CSW14,BRGHHHSW16,KAF20}. Alternatively, if two (or more)
sources of weak randomness are available to a player (under certain
strict conditions on the correlations between these different
sources), these can be combined to obtain (approximately) uniform
randomness \cite{CLW14,AFPS16}. Composing this with a standard QKD
protocol would allow secret keys to be distributed when only weak
randomness is available to the honest parties.

\subsection{Device-independent QKD}
\label{sec:alternative.di}

In this review we have so far always considered scenarios for which it
is assumed that the players have trusted quantum devices, which work
exactly according to their specifications. For instance, if a player
instructs the device to generate a $\zero$ state, then it is assumed
that the device generates precisely this state. This assumption is
however not met in any actual implementation with realistic devices,
as these are never perfect. Indeed, there have been numerous
demonstrations of successful attacks against implementations of
quantum cryptographic protocols that exploited deviations of the
devices' functionality from the specifications, as discussed in
\secref{sec:attacks}. Crucially, this problem cannot be solved only by
a more careful design of the devices, for it appears to be impossible
to guarantee their perfect working under all possible environmental
conditions.

A theoretical solution to this problem is to devise protocols whose
security does not rely on the assumption that devices are perfect.
Ideally, they should provide security guarantees even if the devices
are untrusted, meaning that their behavior may deviate arbitrarily
from the specification. Remarkably, using quantum devices, this is
possible (with certain caveats described below). The idea is to use a
phenomenon called \emph{(Bell) non\-/locality} \cite{Bell64} \--- see
also \textcite{Sca13,BCPSW14} for review articles on the topic. The
subfield of cryptography that studies the use of non\-/locality to
design protocols that work with untrusted devices is termed
\emph{device\-/independent cryptography}.
  
In a nutshell, a Bell inequality is a bound on the probability of
observing certain values in an experiment involving measurements of two isolated (and hence non-communicating) systems. The bound characterises classical locality: it cannot be violated if
the two isolated systems are described by classical physics. However, the bound can be violated by measurements on entangled quantum systems.   One of the most commonly used Bell inequalities is the
CHSH inequality \cite{CHSH69}. It states that, if two players each
hold non\-/communicating systems, and each performs one out of two
binary measurements chosen uniformly at random on their respective
system, where the choice of the measurement is given by
$x,y \in \{0,1\}$ and the outcome is given by $a,b \in \{0,1\}$,
respectively, then the probability that $xy = a \xor b$ should be less than
or equal to $3/4$.\footnote{An alternative formulation of the
  inequality is
  $\left| E(0,0) + E(0,1) + E(1,0) - E(1,1) \right| \leq 2$, where
  $E(x,y)$ is the expected value of the product of the outcomes of the
  systems when measured with settings $x$ and $y$, respectively, and
  the outcomes are values in $\{-1,+1\}$.} But if the systems are
quantum, it is possible to observe this outcome with probability up to
$\approx .85$ \--- this is achieved if the systems are in a perfectly
entangled state and the players perform an optimal measurement.

An observation of a violation of a Bell inequality implies that the
measurement outcomes contain some genuine randomness
\cite{Col06,PAM10,AMP12,CR12}, even conditioned on the knowledge of
the person who set up and programmed the devices used in the
experiment \--- the only assumptions being that no information other
than the measurement result leaves the devices, and that these devices
never fall in the hands of an adversary, since their internal memory
may contain a copy of the measurement outcomes. This randomness may
then be used to generate uniform random numbers
\cite{VV12,CSW14,MS14,BRGHHHSW16,KAF20} or a shared secret key
\cite{BHK05,PABGMS09,VV14,ADFRV18,AFRV19}.

For a review of different results and techniques in
device\-/independent cryptography, we refer to \textcite{ER14}. In
this section we show how to model device\-/independent quantum key
distribution (DI-QKD) in the AC framework. It then follows from the
composition theorem of AC, that the resulting key can be safely used
in applications requiring one.

The converters $\pi^\qkd_A$ and $\pi^\qkd_B$ modeling Alice's and
Bob's parts of the protocol in \secref{sec:qkd} are systems which
generate quantum states and perform measurements. In DI-QKD, exactly
these operations cannot be trusted. So instead, the DI protocol
$(\pi^\diqkd_A,\pi^\diqkd_B)$ will only involve \emph{classical}
operations. Everything \emph{quantum} is moved into a resource, a
device $\aD$. The honest players can send bits to these devices, and
receive bits back from them \--- this corresponds to choosing a
measurement $x,y \in \{0,1\}$ and receiving the outcome
$a,b \in \{0,1\}$, described a few paragraphs earlier. The adversary is
permitted to ``program'' these devices by providing some initial state
$\rho$ as input \--- depending on the model, Eve may be allowed to
provide further inputs to the device at some later point, e.g., to
provide more EPR pairs so that it may continue running. The
corresponding real world is drawn in
\figref{fig:alternatives.diqkd}. The ideal world will be identical to
that of standard QKD, since we wish to construct the same key
resource, i.e., \figref{fig:qkd.resource.sim}.

\begin{figure}[tb]

\begin{tikzpicture}[
sArrow/.style={->,>=stealth,thick},
thinResource/.style={draw,thick,minimum width=2.4cm,minimum height=1cm},
sqResource/.style={draw,thick,minimum width=1cm,minimum height=1cm},
protocol/.style={draw,rounded corners,thick,minimum width=1.2cm,minimum height=2.5cm},
pnode/.style={minimum width=.6cm,minimum height=.5cm}]

\small

\def\t{4} 
\def\u{2.8} 
\def\v{.75}
\def\w{1.3} 

\node[pnode] (a1) at (-\u,\v) {};
\node[pnode] (a2) at (-\u,0) {};
\node[pnode] (a3) at (-\u,-\v) {};
\node[protocol] (a) at (-\u,0) {};
\node[yshift=-2,above right] at (a.north west) {\footnotesize
  $\pi^{\diqkd}_A$};
\node (alice) at (-\t,0) {};

\node[pnode] (b1) at (\u,\v) {};
\node[pnode] (b2) at (\u,0) {};
\node[pnode] (b3) at (\u,-\v) {};
\node[protocol] (b) at (\u,0) {};
\node[yshift=-2,above right] at (b.north west) {\footnotesize $\pi^{\diqkd}_B$};
\node (bob) at (\t,0) {};

\node[thinResource] (cch) at (0,\v) {};
\node[yshift=-2,above right] at (cch.north west) {\footnotesize
  Authentic ch.~$\aA$};
\node[sqResource] (da) at (-\w,-\v) {$\aD_A$};
\node[yshift=-1.5,above] at (da.north) {\footnotesize
  Device};
\node[pnode] (dan) at (-\w,-\v) {};
\node[sqResource] (db) at (\w,-\v) {$\aD_B$};
\node[yshift=-1.5,above] at (db.north) {\footnotesize
  Device};
\node[pnode] (dbn) at (\w,-\v) {};

\node (eveq1) at (-\w,-1.75) {};
\node (eveq2) at (\w,-1.75) {};
\node (evec) at (0,-1.75) {};
\node (junc3) at (evec |- b1) {};

\draw[sArrow,<->] (a1) to node[auto,pos=.2] {$t$} node[auto,pos=.8] {$t$}  (b1);
\draw[sArrow] (junc3.center) to node[auto,pos=.9] {$t$} (evec.center);

\draw[sArrow] (a2) to node[auto,pos=.75,swap] {$k_{A},\bot$} (alice.center);
\draw[sArrow] (b2) to node[auto,pos=.75] {$k_{B},\bot$} (bob.center);

\draw[sArrow] (eveq1.center) to node[pos=.3,auto,swap] {$\rho_A$} (dan);
\draw[sArrow] (eveq2.center) to node[pos=.3,auto,swap] {$\rho_B$} (dbn);

\draw[sArrow,bend left] (a3) to node[pos=.5,auto] {$x$} (dan);
\draw[sArrow,bend left] (dan) to node[pos=.5,auto] {$a$} (a3);
\draw[sArrow,bend right] (b3) to node[pos=.5,auto,swap] {$y$} (dbn);
\draw[sArrow,bend right] (dbn) to node[pos=.5,auto,swap] {$b$} (b3);

\end{tikzpicture}

\caption[DI-QKD]{\label{fig:alternatives.diqkd}The real world setting
  for  a DI-QKD protocol. Eve can program the devices $\aD$, but cannot
  receive any output from them.}
\end{figure}
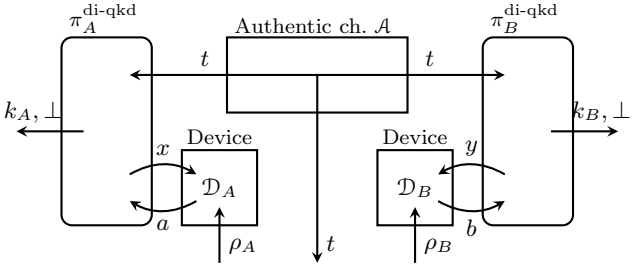

Applying \defref{def:security}, this means that the protocol
$(\pi^\diqkd_A,\pi^\diqkd_B)$ constructs $\aK$ from $\aA$, $\aD_A$ and
$\aD_B$ within $\eps$ if
\begin{equation} \label{eq:diqkd}
  \exists \sigma_E, \quad \pi_A^{\diqkd}\pi_B^{\diqkd}(\aD_A \| \aD_B \| \aA)
  \close{\eps} \aK \sigma_E.
\end{equation}
Note that we have not specified the behaviors of the devices $\aD_A$
and $\aD_B$ at all. In fact, we need \eqnref{eq:diqkd} to hold for all
devices $\aD_A$ and $\aD_B$.\footnote{The simulator may depend on
  these devices, i.e., $\forall \aD_A, \aD_B, \exists \sigma_E$ such
  that \eqnref{eq:diqkd} holds.} This is exactly the
\emph{device\-/independant} guarantee, namely that security holds
regardless of how the (quantum) devices work. Alternatively one can
consider fixed devices $\aD_A$ and $\aD_B$ that are universal
computers, and have their program be part of the inputs at the $E$ interface.

As usual, completeness is captured by specific devices $\aD'_A$ and
$\aD'_B$ that work honestly \--- e.g., they share perfectly entangled
states and perform the correct measurements as specified by the
protocol \--- as well as the same honest resources $\aA'$ and $\aK'$
as in \secref{sec:qkd}. Additionally to
\eqnref{eq:diqkd}, we also need
\begin{equation*} 
  \pi_A^{\diqkd}\pi_B^{\diqkd}(\aD'_A \| \aD'_B
  \| \aA') \close{\eps'} \aK'.
\end{equation*}

The same reduction as for (normal) QKD goes through, and one can show
that \eqnref{eq:diqkd} is satisfied if for all behaviors of the
devices (and their inputs), \eqnsref{eq:qkd.cor} and \eqref{eq:qkd.sec}
hold.

Note however that the construction outlined in this section only
allows the devices $\aD_A$ and $\aD_B$ to be accessed during the
protocol. No access is granted after the protocol ends, meaning that
we make no security statement about what happens if the devices are
reused. It is an open question how to reuse devices in DI
cryptography, which we discuss in \secref{sec:open.di}.

Proving security of device-independent QKD is more challenging than in
the device-dependent case. One of the difficulties is that the
measurement operators that describe Alice and Bob's measurement can be
arbitrary. In particular it cannot be assumed, for instance, that two
subsequent measurement outcomes by Bob are obtained by two separate
measurement processes. While some of the techniques described
in~\secref{sec:securityproofs}, such as entropy accumulation, are
still applicable to the device-independent setting, others, like de
Finetti-type theorems, are not or must be adapted,
[see~\textcite{ArnonThesis} for details].

\subsection{Semi-device-independent QKD}
\label{sec:alternative.semi}

The only assumption made about the devices in DI-QKD is that no
information leaves these devices unless allowed by the protocols (see
\secref{sec:alternative.di}). But achieving the violation of Bell
inequalities needed for this is challenging because it requires high
detector efficiency and tolerates only low noise on the
channel~\cite{BCPSW14}. Protocols that are easier to implement can be
achieved by making additional assumptions about the quantum devices
used by Alice and Bob. These are generally called
\emph{semi\-/device\-/independent} (semi-DI).

Many different assumptions may be labeled semi-DI. For example, in a
one-sided device-independent setting the protocol is DI for Bob but
not for Alice \cite{BCWSW12}. One may also assume dimension bounds on
the quantum systems generated by untrusted devices as in
\textcite{PB11}. Alternatively, \textcite{LPTRG13} assume that each
use of the devices are causally independent \--- i.e., the states
generated and measurements performed are in product form \--- to
analyze a protocol where the Bell violation is measured locally in
Alice's lab, thus avoiding the noise introduced by the channel between
Alice and Bob. Similar ideas have been used for other protocols than
QKD, e.g., semi-DI quantum money~\cite{HS20,BDG19}.

One of the most promising forms of semi-DI QKD, which has already been
implemented over large distances
\cite{Liu13,Tang2014,Pirandola2015,Yin2016} is
\emph{measurement\-/device\-/independent} (MDI) QKD
\cite{LCQ12,BP12,MR12,CXCLTL14}. Here, one assumes that players can
generate the states they desire, but one does not trust measurement
devices at all. This model is motivated by the attacks on the
detectors, e.g., the time-shift attacks or detector blinding attacks
discussed in \secref{sec:attacks:hacking}.

To understand how such protocols work, we will start from an
entanglement based protocol as in
\secref{sec:alternative.entanglement}, then modify it step by
step, until we achieve a prepare\-/and\-/measure protocol, in which
all measurements are performed by the adversary. Since the final
protocol is as secure as the original one, and the original one is
secure for all adversaries, the final MDI QKD protocol is secure
for all adversaries as well. In particular, it is secure for
adversaries that completely control the measurement apparatus.

In an entanglement based protocol, Alice and Bob receive the $A$ and
$B$ parts of a state $\psi_{ABR}$, and measure these systems in either
the computational or diagonal basis, obtaining a raw key. This key is
then processed as in a prepare\-/and\-/measure protocol (see
\secref{sec:qkd.protocol} and \secref{sec:securityproofs}).  If the
source gave them a state which is (close to) a tensor product of EPR
pairs, such a protocol is guaranteed to terminate with a shared secret
key. Equivalently, the source could generate any of the Bell states,
and notify Alice and Bob which one it gave them. They then perform bit
or phase flips to change it to an EPR pair.

Instead of the source distributing an entangled state, Alice and Bob
could each generate an EPR pair $AA'$ and $BB'$, respectively. They
then send $A'$ and $B'$ to a third party, Charlie, who measures this
in the Bell basis, and tells them the measurement outcome. If
performed correctly, the $AB$ system will be in a Bell state, and the
measurement outcome will tell them which one. By flipping bits or
phases, Alice and Bob can turn this into an EPR pair, and continue
with the protocol as above. Crucially, if Charlie does not perform the
correct measurement, then Alice and Bob will end up holding the $A$
and $B$ parts of some unknown state $\psi_{ABR}$. But this does not
compromise security: if it is too far from the expected state, the
protocol will just abort.

Instead of first performing a bit or phase flip, and then measuring,
Alice and Bob could first measure their systems $A$ and $B$, and then
flip the value of the measurement outcome if needed. And instead of
generating EPR pairs $AA'$ and $BB'$, then measuring $A$ and $B$, they
could pick the measurement outcome at random, then generate the
corresponding reduced state in $A'$ and $B'$ respectively, send these
to Charlie, and when they obtain the measurement outcome from Charlie,
they flip their bits if needed.

The only (trusted) quantum operations that Alice and Bob need to
perform in the protocol described in the paragraph above are
generating the states in the systems $A'$ and $B'$. All measurements
have now been delegated to Charlie, who may deviate arbitrarily from
the protocol without compromising security.

The real world for such a MDI-QKD protocol is drawn in
\figref{fig:alternatives.mdi}, where one can see that the converters
$\pi^{\mdi}_A$ and $\pi^{\mdi}_B$ do not have any incoming quantum
states, i.e., they do not need to perform any measurement.

Security proofs for MDI-QKD protocols can be based on the same techniques as those for fully device-independent protocols, as discussed in \secref{sec:alternative.di}. The comments on security proofs made in that section thus also apply here. 

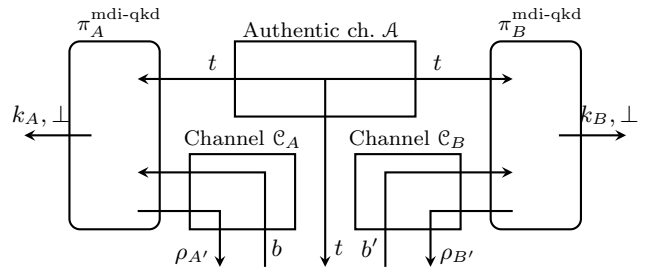
\begin{figure}[tb]

\begin{tikzpicture}[
sArrow/.style={->,>=stealth,thick},
thinResource/.style={draw,thick,minimum width=2.4cm,minimum height=1cm},
sqResource/.style={draw,thick,minimum width=1.4cm,minimum height=1cm},
protocol/.style={draw,rounded corners,thick,minimum width=1.2cm,minimum height=2.5cm},
pnode/.style={minimum width=.6cm,minimum height=.5cm}]

\small

\def\t{4} 
\def\u{2.8} 
\def\v{.75}
\def\w{1.1} 
\def\a{.3}

\node[pnode] (a1) at (-\u,\v) {};
\node[pnode] (a2) at (-\u,0) {};
\node[pnode] (a3) at (-\u,-\v) {};
\node[protocol] (a) at (-\u,0) {};
\node[yshift=-2,above right] at (a.north west) {\footnotesize
  $\pi^{\mdi}_A$};
\node (alice) at (-\t,0) {};

\node[pnode] (b1) at (\u,\v) {};
\node[pnode] (b2) at (\u,0) {};
\node[pnode] (b3) at (\u,-\v) {};
\node[protocol] (b) at (\u,0) {};
\node[yshift=-2,above right] at (b.north west) {\footnotesize $\pi^{\mdi}_B$};
\node (bob) at (\t,0) {};

\node[thinResource] (cch) at (0,\v) {};
\node[yshift=-2,above right] at (cch.north west) {\footnotesize
  Authentic ch.~$\aA$};
\node[sqResource] (da) at (-\w,-\v) {};
\node[yshift=-1.5,above] at (da.north) {\footnotesize
  Channel $\aC_A$};
\node[pnode] (dan) at (-\w,-\v) {};
\node[sqResource] (db) at (\w,-\v) {};
\node[yshift=-1.5,above] at (db.north) {\footnotesize
  Channel $\aC_B$};
\node[pnode] (dbn) at (\w,-\v) {};

\node (eveq11) at (-\w-\a,-1.75) {};
\node (eveq12) at (-\w+\a,-1.75) {};
\node (eveq21) at (\w-\a,-1.75) {};
\node (eveq22) at (\w+\a,-1.75) {};
\node (evec) at (0,-1.75) {};
\node (junc3) at (evec |- b1) {};

\draw[sArrow,<->] (a1) to node[auto,pos=.2] {$t$} node[auto,pos=.8] {$t$}  (b1);
\draw[sArrow] (junc3.center) to node[auto,pos=.9] {$t$} (evec.center);

\draw[sArrow] (a2) to node[auto,pos=.75,swap] {$k_{A},\bot$} (alice.center);
\draw[sArrow] (b2) to node[auto,pos=.75] {$k_{B},\bot$} (bob.center);

\draw[sArrow] (a3.south east) to (eveq11 |- a3.south east) to node[pos=.8,auto,swap] {$\rho_{A'}$} (eveq11.center);
\draw[sArrow] (b3.south west) to (eveq22 |- b3.south west) to node[pos=.8,auto] {$\rho_{B'}$} (eveq22.center);
\draw[sArrow] (eveq12.center) to node[pos=.2,auto,swap,xshift=-1] {$b$} (eveq12 |- a3.north east) to (a3.north east);
\draw[sArrow] (eveq21.center) to node[pos=.2,auto,xshift=1] {$b'$} (eveq21 |- b3.north west) to (b3.north west);

\end{tikzpicture}

\caption[MDI-QKD]{\label{fig:alternatives.mdi}The real world setting
  for a MDI-QKD protocol. The only quantum operations performed by
  $\pi^{\mdi}_A$ and $\pi^{\mdi}_B$ are the generation of quantum
  states. The communication resources $\aC$ send quantum systems from
  Alice or Bob to Eve, and classical bits from Eve to Alice and Bob..}
\end{figure}

\subsection{Memoryless adversaries}
\label{sec:alternative.memoryless}

The previous sections analyzed different models of QKD, in which we
changed the capabilities and resources of the honest players running
the protocol. Similar techniques may also be used to model
limitations on adversaries. In this last section we consider as example QKD
protocol with an adversary that has no (long-term) quantum memory, and is
forced to measure the quantum states exchanged between Alice and Bob
during the QKD protocol and store the classical information.

The insecure channel resource, $\aQ$, modeled as part of the real QKD
system in \figref{fig:qkd.real.adv} gives complete control over the
states sent on this channel to the adversary. Since this may include
storing them and measuring them at a later point, we need to limit the
adversary's access to this channel as part of the insecure channel
resource. We may thus define a different channel $\tilde{\aQ}$, which
requires Eve to input some measurement specification and then obtains
the measurement outcome at her interface. The resulting
post\-/measurement state is then output at Bob's interface.

The model of $\tilde{\aQ}$ described above is just one possible way
one may imagine limiting Eve's access to the states sent during
QKD. The result is a change in the requirements of the QKD
protocol. Instead of constructing a secure key $\aK$ from an authentic
channel $\aA$ and an insecure channel $\aQ$, it is now sufficient if
$\aK$ can be constructed from $\aA$ and $\tilde{\aQ}$. It is not hard
to see that, since the accessible information (see
\secref{sec:qkd.other.ai}) measures the information an adversary has
\emph{after} measuring their quantum states, a QKD protocol with low
accessible information would satisfy such a construction \--- namely,
$\aA \| \tilde{\aQ} \rightarrow \aK$. The accessible information
security measure is thus a composable criterion under the assumption
that the adversary has such a physical limit on their memory.

Since QKD protocols are secure against general adversaries as modeled
in \secref{sec:qkd}, there does not seem to be much incentive to
consider adversaries with limitations on their memory (unlike for some
two-party protocols discussed in \secref{sec:open.other}). It is
however noteworthy that, as already mentioned in
\secref{sec:qkd.other.ac}, by explicitly limiting the adversary's
capabilities we capture weaker security notions that appeared in the literature.


\section{Secure classical message transmission}
\label{sec:smt}

One of the main tasks in cryptography is to securely transmit a confidential message from one player to another. \emph{Securely} transmitting the message means that the adversary does not learn anything about the message \--- except unavoidable leaks such as the message length \--- and cannot modify the message either. 
We also want to achieve this with minimal assumptions on the available resources. In \figref{fig:construction} we depict the steps necessary to construct such a secure channel from nothing but insecure channels and an initial short key. The aim of this section is to explain this construction in detail. 

In \secref{sec:smt.auth} we first show how to construct an authentic
channel, which is used both by QKD and the OTP. Then in
\secref{sec:smt.qkd} and \secref{sec:smt.otp} we revisit the notions
of a secure key and secure channel resources introduced earlier, and
discuss a modification used here.\footnote{We take into account that
  Eve may prevent the honest players from obtaining the key or the
  transmitted message, which was ignored earlier for simplicity.} In
\secref{sec:smt.together} we put the individual parts of the
construction together and show how this gives a construction of a
secure channel from a short secret key and insecure communication
channels only.

\subsection{Authentication}
\label{sec:smt.auth}

The use of an authentic channel is essential for many cryptographic
protocols, including quantum key distribution, as we have seen
earlier. It allows the players Alice and Bob to be sure that they
are communicating with each other, and not with an adversary Eve. The
authentic channel we used in \secref{sec:qkd} \--- e.g.,
\figref{fig:qkd.real} \--- is however idealized: it guarantees that
the recipient always receives the message that was sent. In a
realistic situation, one has to assume that an adversary may jumble or
cut the communication, and prevent messages from arriving. What still
can be constructed is a channel which guarantees that Bob does not
receive a corrupted message. He either receives the correct message
sent by Alice, or an error, which symbolizes an attempt by Eve to
change the message. This can be modeled by giving Eve's idealized
interface two controls: the first provides her with Alice's message,
the second allows her to input one bit which specifies whether Alice's
message should be delivered or whether Bob gets an error instead. We
illustrate this in \figref{fig:auth.resource}.

\begin{figure}[tb]
\begin{tikzpicture}[
sArrow/.style={->,>=stealth,thick},
largeResource/.style={draw,thick,minimum width=1.618*2cm,minimum height=2cm}]

\small

\node[largeResource] (keyBox) at (0,0) {};
\node (alice) at (-2.6,0) {Alice};
\node (bob) at (2.6,0) {Bob};
\node (eve) at (0,-1.7) {Eve};
\node (ajunc) at (eve.north west |- alice) {};

\draw[thick] (alice) to node[pos=.12,auto] {$x$} (0,0) to node[pos=.5]
(ejunc) {} +(160:-.8);
\draw[sArrow] (ajunc.center) to node[pos=.85,auto,swap] {$x$} (eve.north west);
\draw[sArrow] (.8,0) to node[pos=.9,auto] {$x,\bot$} (bob);
\draw[double] (ejunc.center |- eve.north) to node[pos=.15,auto,swap] {$0,1$} (ejunc.center);

\end{tikzpicture}

\vspace{6pt}






\caption[Authentic channel resource]{\label{fig:auth.resource}An
  authentic channel resource. The message input at Alice's interface
  is visible to Eve, who gets to decide if Bob receives it or not. But
  it guarantees that, if Bob does receive a message, then it
  corresponds to the one sent by Alice.}
\end{figure}
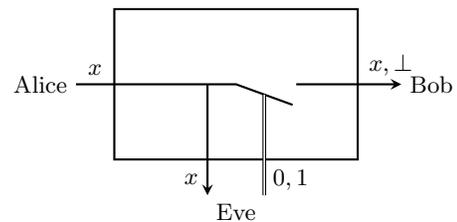

As we are going to explain, such an authentic channel can be
constructed from a completely insecure channel together with a shared
secret key. Although this may be done using a non\-/uniform secret key
[see \textcite{RW03,DW09,ACLV19}], we review here a simpler
construction, originally proposed by \textcite{WC81}, which still only
needs a short key, but which however has to be (close to) uniform: one
computes a hash $h_k(x)$ of the message $x$, and sends the string
$x \| h_k(x)$ to Bob, where $k$ is the short shared secret key and
$\{h_k\}_{k \in \cK}$ a family of strongly universal hash
functions.\footnote{See the formal definition later in this work in
  \footnoteref{fn:universalhash}.} Alice's part of the authentication
protocol $\pi^{\auth}_A$ thus gets as input a key $k$ from an ideal
key resource as well as a message $x$ from Alice, and sends
$x \| h_k(x)$ over the insecure channel. When Bob receives a string
$x' \| y'$, he needs to check whether $y' = h_k(x')$. His part of the
protocol hence gets as input the key $k$ from the ideal key resource,
the message $x' \| y'$ delivered by the channel, and outputs $x'$ if
$y' = h_k(x')$, and otherwise a symbol $\bot$ to indicate an
error. This is depicted in \figref{fig:auth.real}.

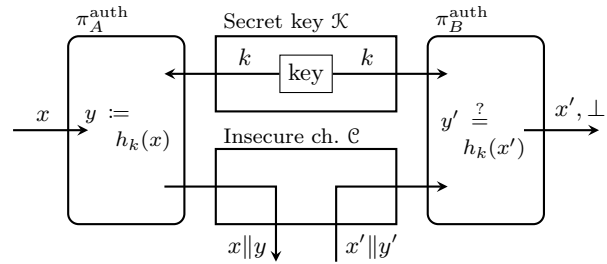
\begin{figure}[tb]

\begin{tikzpicture}[
sArrow/.style={->,>=stealth,thick},
thinResource/.style={draw,thick,minimum width=2.4cm,minimum height=1cm},
protocol/.style={draw,rounded corners,thick,minimum width=1.545cm,minimum height=2.5cm},
pnode/.style={minimum width=1cm,minimum height=.5cm}]

\small

\def\t{3.895} 
\def\u{2.39} 
\def\v{.75}

\node[pnode] (a1) at (-\u,\v) {};
\node[pnode] (a2) at (-\u,0) {};
\node[pnode] (a3) at (-\u,-\v) {};
\node[protocol,text width=1.1cm] (a) at (-\u,0) {\footnotesize $y
  \coloneqq $\\$\quad \ h_k(x)$};
\node[yshift=-2,above right] at (a.north west) {\footnotesize
  $\pi^{\auth}_A$};
\node (alice) at (-\t,0) {};

\node[pnode] (b1) at (\u,\v) {};
\node[pnode] (b2) at (\u,0) {};
\node[pnode] (b3) at (\u,-\v) {};
\node[protocol,text width=1.2cm] (b) at (\u,0) {\footnotesize $y' \stackrel{?}{=}$\\$\quad h_k(x')$};
\node[yshift=-2,above right] at (b.north west) {\footnotesize $\pi^{\auth}_B$};
\node (bob) at (\t,0) {};

\node[thinResource] (keyBox) at (0,\v) {};
\node[draw] (key) at (0,\v) {key};
\node[yshift=-2,above right] at (keyBox.north west) {\footnotesize
  Secret key $\aK$};
\node[thinResource] (channel) at (0,-\v) {};
\node[yshift=-1.5,above right] at (channel.north west) {\footnotesize
  Insecure ch.~$\aC$};
\node (eveleft) at (-.4,-1.75) {};
\node (everight) at (.4,-1.75) {};
\node (ajunc) at (eveleft |- a3) {};
\node (bjunc) at (everight |- b3) {};

\draw[sArrow] (key) to node[auto,swap,pos=.3] {$k$} (a1);
\draw[sArrow] (key) to node[auto,pos=.3] {$k$} (b1);

\draw[sArrow] (alice.center) to  node[auto,pos=.4] {$x$} (a2);
\draw[sArrow] (b2) to node[auto,pos=.75] {$x',\bot$} (bob.center);

\draw[sArrow] (a3) to (ajunc.center)
to node[pos=.8,auto,swap] {$x \| y$} (eveleft.center);
\draw[sArrow] (everight.center) to node[pos=.2,auto,swap] {$x' \| y'$}
(bjunc.center) to (b3);

\end{tikzpicture}

\caption[Real authentication system]{\label{fig:auth.real}The real
  authentication system consists of the authentication protocol
  $(\pi^{\auth}_A,\pi^{\auth}_B)$ as well as the secret key and insecure
  channel resources, $\aK$ and $\aC$. As in previous illustrations,
  Alice has access to the left interface, Bob to the right interface
  and Eve to the lower interface.}
\end{figure}

To capture completeness of this protocol, one considers instead of the insecure channel $\aC$ as in
\figref{fig:auth.real} a noiseless channel with a
blank interface for Eve\footnote{Unlike in the case of QKD, we do not
  consider noisy channels between Alice and Bob, as such noise could
  be removed easily by encoding the communication with an appropriate classical error
  correcting code. Therefore, the assumed and constructed channels
  faithfully deliver the message from Alice to Bob.}  [as illustrated
in \figref{fig:security.channel.noiseless}], and the constructed
channel is also a perfect noiseless channel instead of the channel
from \figref{fig:auth.resource}. These real and ideal systems are
indistinguishable as they are both identity channels which faithfully
transmit $x$ from Alice to Bob. This proves completeness, and we can therefore focus in the
following on the other part, namely proving soundness of the protocol.

In the ideal setting, the authentic channel
(\figref{fig:auth.resource}) has the same interface on Alice's and
Bob's sides as the real setting (\figref{fig:auth.real}): Alice can
input a message, and Bob receives either a message or an
error. However, Eve's interface looks quite different: in the real
setting she can modify the transmission on the insecure channel,
whereas in the ideal setting the adversarial interface provides only
controls to read the message and interrupt the transmission. From
\defref{def:security} we have that an authentication protocol
constructs the authentic channel if there exists a simulator
$\sigma^{\auth}_E$ that can recreate the real interface while
accessing just the idealized one. An obvious choice for the simulator
is to first generate its own key $k$ and output $x \| h_k(x)$. Then
upon receiving $x' \| y'$, it checks if $x' \| y' = x \| h_k(x)$ and
presses the switch on the authentic channel to output an error if this
does not hold. We illustrate this in \figref{fig:auth.ideal}.

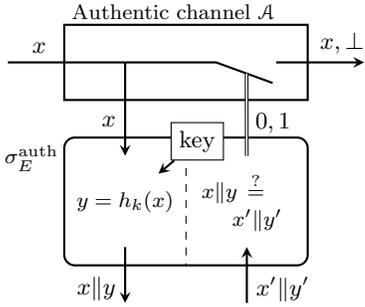
\begin{figure}[tb]

\begin{tikzpicture}[
sArrow/.style={->,>=stealth,thick},
thinResource/.style={draw,thick,minimum width=1.618*2cm,minimum height=1cm},
simulator/.style={draw,rounded corners,thick,minimum width=1.618*2cm,minimum height=1.7cm},
snode/.style={minimum width=1.1cm,minimum height=1.2cm}]

\small

\def\t{2.368} 
\def\u{-1.1}
\def\v{.75}
\def\w{-2.45} 

\node[thinResource] (channel) at (0,\v) {};
\node[yshift=-1.5,above right] at (channel.north west) {\footnotesize
  Authentic channel $\aA$};
\node (alice) at (-\t,\v) {};
\node (bob) at (\t,\v) {};

\node[simulator] (sim) at (0,\u) {};
\node[xshift=1.5,below left] at (sim.north west) {\footnotesize
  $\sigma^{\auth}_E$};
\node[snode,ellipse] (sleft) at (-.809,\u) {};
\node[snode] (sright) at (.809,\u) {};
\draw[dashed] (sim.north) to (sim.south);

\node (ajunc) at (sleft |- alice) {};
\node (bjunc) at (sright |- bob) {};

\draw[thick] (alice.center) to node[pos=.15,auto] {$x$} (.4,\v) to node[pos=.54] (ejunc) {} +(160:-.8);
\draw[sArrow] (ajunc.center) to node[pos=.63,auto,swap] {$x$} (sleft);
\draw[sArrow] (1.2,\v) to node[pos=.75,auto] {$x,\bot$} (bob.center);
\draw[double] (sright) to node[pos=.4,auto,swap] {$0,1$} (ejunc.center);

\node (sltext) at (-.809,\u) {\footnotesize $y = h_k(x)$};
\node[text width=1.2cm] (srtext) at (.809,\u) {\footnotesize $x \| y
  \stackrel{?}{=}$\\$\quad \ x' \| y'$};

\node (eveleft) at (sleft |- 0,\w) {};
\node (everight) at (sright |- 0,\w) {};
\draw[sArrow] (sleft) to node[pos=.75,auto,swap] {$x \| y$} (eveleft.center);
\draw[sArrow] (everight.center) to node[pos=.25,auto,swap] {$x' \| y'$}
(sright);

\node[draw,fill=white] (key) at (.15,\u+.8) {key};
\draw[sArrow] (key) to (sleft);

\end{tikzpicture}

\caption[Ideal authentication system]{\label{fig:auth.ideal}The ideal authentication system \--- Alice has access to the left interface,  Bob to the right interface and Eve to the lower interface \---  consists of the ideal authentication resource and a simulator  $\sigma^{\auth}_E$.}
\end{figure}

In this case, an authentication protocol is $\eps$\=/secure if Figures~\ref{fig:auth.real} and \ref{fig:auth.ideal} are $\eps$\=/close, i.e.,
\begin{equation} \label{eq:security.auth} \pi^{\auth}_A\pi^{\auth}_B
  \left(\aK \| \aC \right) \close{\eps} \aA\sigma^{\auth}_E.
\end{equation}

Original works defining authentication \--- e.g.,
\textcite{WC81,Sim85,Sim88,Sti90,Sti94} \--- did not use such a
composable security definition. Instead, they considered two kinds of
attacks. In the first, the adversary obtains a pair of a valid message
and authentication tag, and tries to find a pair of a different
message and corresponding valid authentication tag \--- this is called
a \emph{substitution attack}. In the second, the adversary directly
tries to find a pair of message and corresponding valid authentication
tag \--- this is called an \emph{impersonation attack}. It was then
shown that if the family of hash functions used are $\eps$\=/almost
strongly $2$\-/universal,\footnoteremember{fn:universalhash}{A family
  of functions $\{h_k : \cX \to \cY\}_k$ is said to be $\eps$\=/almost
  strongly $2$\-/universal if any two different messages are almost
  uniformly mapped to all pairs of tags, i.e.,
  $\forall x_1,x_2,y_1,y_2,x_1 \neq x_2, \Pr_k \left[ h_k(x_1) = y_1
    \text{ and } h_k(x_2) = y_2\right] \leq
  \frac{\eps}{|\cY|}$~\cite{Sti94}. The family of functions is said to
  be strongly $2$\-/universal if $\eps = 1/|\cY|$.}  the probability
of either of these attacks being successful is bounded by $\eps$. A
composable security proof for these schemes was given by
\textcite{Por14}, who showed that \eqnref{eq:security.auth} is
satisfied, again under the condition that the family of hash functions
used are $\eps$\=/almost strongly
$2$\-/universal.\footnote{\textcite{Por14} additionally shows that
  part of the secret key $k$ can be recycled, since only a number of
  bits corresponding to the length of the hash $h_k(x)$ are
  leaked. This is discussed in \secref{sec:recycle}.}

Note that a distinguisher interacting with either of the real or ideal
systems has the choice between providing messages in two different
orders. It can first provide Alice with a message, receive the
ciphertext\footnote{By \emph{ciphertext} we denote the pair of the
  message and authentication tag generated by the sender.} at Eve's
interface, then input a modified ciphertext, and finally learn whether
the ciphertext is accepted or not at Bob's interface. Or, it can first
input a forged ciphertext at Eve's interface, then learn if it is
accepted, and finally provide Alice with a message and obtain the
corresponding ciphertext at Eve's interface. These two orders of
messages roughly correspond to the substitution and impersonation
attacks.

The secret key resource used so far in this section assumes that both
players always get a copy of the key. However, in \secref{sec:qkd} we
modeled a secret key resource with a switch at Eve's interface, giving
her the possibility to prevent the players from getting the key. If
such a switch is present and Eve flips it, the honest players will not
be able to run the authentication protocol at all. This does however
not change the ideal resource constructed, because not running the
protocol or running the protocol but Eve preventing the message from
being delivered are essentially equivalent. However, the proof in this
case requires a different simulator \--- one which receives the bit
deciding whether the players get a key or not and then acts accordingly.

In \secref{sec:smt.qkd} an even weaker secret key resource is
considered, one which allows Eve to decide if only one of the two
players gets a secret key, but not the other \--- this is drawn in
\figref{fig:key.resource}. Running a similar reasoning as in the
paragraph above, one can see that this does not change the outcome of
the protocol either: it still constructs the authentic channel from
\figref{fig:auth.resource}.

\begin{figure}[tb]

\begin{tikzpicture}[
sArrow/.style={->,>=stealth,thick},
largeResource/.style={draw,thick,minimum width=1.618*2cm,minimum height=2cm}]

\small

\def\u{1.3}
\def\s{.8}

\node[largeResource] (keyBox) at (0,0) {};
\node (alice) at (-3,0) {Alice};
\node (bob) at (3,0) {Bob};
\node (eve) at (0,-1.7) {Eve};
\node[draw] (key) at (0,0) {key};
\node (juncL) at (-\u,0) {};
\node (juncR) at (\u,0) {};

\draw[sArrow] (juncR.center) to node[pos=0.7,auto] {$k,\bot$} (bob);
\draw[sArrow] (juncL.center) to node[pos=0.7,auto,swap] {$k,\bot$} (alice);
\draw[thick] (key) to (-\u+\s,0) to node[pos=.5] (handleL) {} +(210:\s);
\draw[thick] (key) to (\u-\s,0) to node[pos=.5] (handleR) {} +(330:\s);
\draw[double] (handleL |- eve.north) to node[pos=.15,auto,swap] {$0,1$} (handleL.center);
\draw[double] (handleR |- eve.north) to node[pos=.15,auto] {$0,1$} (handleR.center);

\end{tikzpicture}

\caption[Secret key resource with two
switches]{\label{fig:key.resource}A secret key resource allowing Eve
  to control who gets the key: the two bits Eve inputs control whether
  Alice and Bob, respectively, obtain a key or an error message from
  the resource.}
\end{figure}
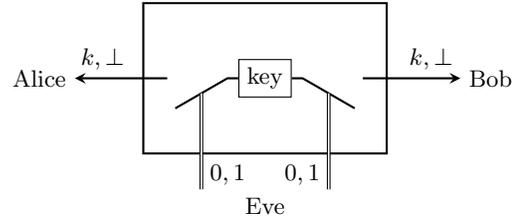


\subsection{Quantum key distribution}
\label{sec:smt.qkd}

In \secref{sec:qkd} we analyzed QKD protocols that use an insecure quantum channel and an authentic channel with the guarantee that the message is always delivered, as indicated in \figref{fig:qkd.real}. The motivation behind this (standard) choice was that, if the message is not delivered, then the players abort and the scheme is trivially secure. In other words, the non-trivial case that needs to be analyzed to prove that a QKD scheme constructs  the ideal key resource is the one in which the adversary does not  use her switch and allows the messages to be delivered on the  authentic channel.

Nonetheless, if we do replace the authentic channel with the one that
can actually be constructed from \figref{fig:auth.resource}, then we
also have to weaken the ideal key resource that is constructed. If the
players get an error message from the authentic channel instead of the
intended message, they will simply abort the protocol and not produce
a key. In the version from \secref{sec:qkd}, Eve already has the power
to prevent the players from getting a secret key. The difference is
that now Eve can let one player get the secret key, but not the other,
e.g., by jumbling the last message between the players. The resulting
ideal key resource is drawn in \figref{fig:key.resource}.

The analysis carried out in \secref{sec:qkd} goes through with only
minor changes with the weaker authentic channel and secret key
resources, since the only differences are the abort conditions, which
now may additionally occur because of failed authentication. In
particular, the reduction from the constructive statement,
\eqnref{eq:qkd.security}, to the trace distance criterion,
\eqnref{eq:qkd.sec}, is unaffected by these changes of resources.

Note that if Eve prevents one player from getting the key, but not the
other, the players are generally unaware of this fact and may end up
with mismatching key lengths.\footnote{The protocol may be designed in
  such a way that the round in which Eve needs to jumble the
  communication so that one player accepts the key but not the other
  is unknown to her. Thus her probability of success will be $O(1/n)$,
  where $n$ is the number of rounds of communication.}  This is not a
problem which is specific to QKD, but happens in general with any key
distribution scheme \cite{Wol99}.

\subsection{One-time pad}
\label{sec:smt.otp}

In \secref{sec:ac.otp} we analyzed the OTP and showed that it
constructs a secure channel given a secret key and an authentic
channel. But again, we used an authentic channel that guarantees that
the message is transmitted, as well as a secret key that is guaranteed
to be delivered to the players. It is however easy to convince oneself
that these extra assumptions about the resources do not affect the
security of the protocol, since if the players do not get a key or a
message they just abort the protocol, in which case security holds
trivially. Therefore, if we plug in the authentic channel from
\figref{fig:auth.resource} and the secret key resource from
\figref{fig:key.resource}, we merely need to weaken the secure channel
that is constructed, so that it may abort as well. We can model this
by adding a switch at Eve's interface that when pressed, delivers an
error message at Bob's interface instead of Alice's message, as
depicted in \figref{fig:secure.resource}. This is very similar to the
authentic channel resource from \figref{fig:auth.resource}, except
that Eve receives only the length of the message rather than the
message itself.

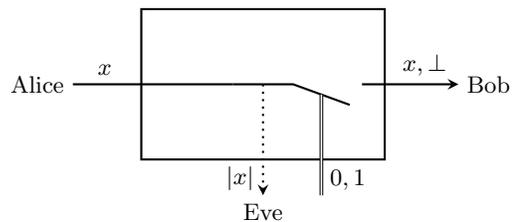
\begin{figure}[tb]

\begin{tikzpicture}[
sArrow/.style={->,>=stealth,thick},
largeResource/.style={draw,thick,minimum width=1.618*2cm,minimum height=2cm}]

\small

\def\z{.4}

\node[largeResource] (keyBox) at (0,0) {};
\node (alice) at (-3,0) {Alice};
\node (bob) at (3,0) {Bob};
\node (eve) at (0,-1.7) {Eve};
\node (juncLL) at (-3*\z,0) {};
\node (juncLR) at (-\z,0) {};
\node (juncRL) at (\z,0) {};
\node (juncRR) at (3*\z,0) {};

\draw[thick] (alice) to node[pos=.2,auto] {$x$} (juncLR.center);
\draw[thick] (juncLR.center) to (juncRL.center) to node[pos=.5]
(handle2) {} +(340:2*\z);
\draw[sArrow] (juncRR) to node[pos=.65,auto] {$x,\bot$} (bob);
\draw[sArrow,dotted] (0,0) to node[pos=.85,auto,swap] {$|x|$} (eve);
\draw[double] (handle2.center |- eve.north) to node[pos=.15,auto,swap] {$0,1$} (handle2.center);

\end{tikzpicture}

\caption[Secure channel with a switch]{\label{fig:secure.resource}A
  secure channel that allows Eve to learn the length of the message
  and prevent Bob from receiving it: when Alice inputs a message $x$
  at her interface, information about the length of the message is
  given to Eve, who can additionally press a switch that either delivers
  Alice's message to Bob or provides him with an error message $\bot$
  instead.}
\end{figure}

The analysis of the OTP with these altered resources is identical to
the one of \secref{sec:ac.otp}, since the only difference is that both
the real and ideal systems might abort or output an error at Bob's
interface instead of a message if Eve operates the corresponding
switch at her interface.

\subsection{Combining the subprotocols}
\label{sec:smt.together}

Let $\aA$ denote an authentic channel resource, as illustrated in
\figref{fig:auth.resource}, and let $\aK^{\ell}$ denote a secret key
resource of length $\ell$, as drawn in
\figref{fig:key.resource}. Furthermore, let $\aC$ be an insecure
classical, and $\aQ$ an insecure quantum channel. Finally, we denote by $\aS$ a secure classical channel, as depicted in
\figref{fig:secure.resource}. If we summarize the results presented so
far, we have that an authentication protocol constructs an authentic
channel from an insecure channel and a secret key, i.e.,
\[ 
\aK^{a} \| \aC \xrightarrow{\pi^\auth_{AB},\eps_\auth}
\aA;\
\]
a QKD protocol constructs a shared secret key resource from an authentic channel and an insecure quantum channel, i.e., 
\[ 
\aA \| \aQ \xrightarrow{\pi^\qkd_{AB},\eps_\qkd}
\aK^n;\]
and a OTP constructs a secure channel from an authentic channel and a secret key, i.e., 
\[
 \aA \| \aK^m \xrightarrow{\pi^\otp_{AB},0} \aS.
\]
Using also the fact that a key can be split, i.e.,  
\[ 
  \aK^{a+b} \xrightarrow{\id,0} \aK^{a} \| \aK^b ,
  \]
denoting by $a_\qkd$ the length of the key used by the authentication subroutines for QKD,\footnote{A QKD protocol usually authenticates   many messages, which may be going in both directions between Alice and Bob. For   simplicity, we write this here as just one round of authentication, which uses a key of length $a_\qkd$ and has error   $\eps^\qkd_\auth$.} and by $a_\otp$ the length of the key used to authenticate the message for constructing the secure channel, we obtain
\[ 
\aC \| \aC \| \aQ \| \aK^{a_\qkd} \xrightarrow{\pi_{AB},\eps} \aS
\| \aK^{n-m-a_\otp} ,
\]
where $\pi_{AB}$ is the composition of all the protocols and $\eps = \eps_\qkd + \eps_\auth^\qkd + \eps_\auth^\otp$ is the sum of the errors of the individual protocols. We depict this in \figref{fig:secure.complete}, where for simplicity we have drawn only one round of authentication as a subroutine of QKD.

\begin{figure*}[tb]

\begin{tikzpicture}[
resource/.style={draw,thick,minimum width=3cm,minimum height=1cm},
sArrow/.style={->,>=stealth,thick},
sLine/.style={-,thick},
protocol2/.style={draw,thick,minimum width=.6cm,minimum height=2.8cm,rounded corners},
protocol3/.style={draw,thick,minimum width=.6cm,minimum height=4.6cm,rounded corners}]

\small

\def\v{1.8}
\def\x{.5} 
\def\p{2.4} 
\def\c{1.2} 
\def\t{7.3}

\node[resource,red] (keyLong) at (0,3*\v) {};
\node[yshift=-1.5,above right] at (keyLong.north west) {\footnotesize
  Secret key $\aK^{a_\qkd}$};
\node[resource,red] (ch1) at (0,2*\v) {};
\node[yshift=-1.5,above] at (ch1.north) {\footnotesize
  Insecure class.\ ch.\ $\aC$};
\node[resource,violet] (chQ) at (0,\v) {};
\node[yshift=-1.5,above] at (chQ.north) {\footnotesize
  Insecure quant.\ ch.\ $\aQ$};
\node[resource,blue] (ch2) at (0,0) {};
\node[yshift=-1.5,above] at (ch2.north) {\footnotesize
  Insecure class.\ ch.\ $\aC$};

\node[protocol2,red] (authqkdA) at (-\p,2.5*\v) {};
\node[yshift=-1.5,above] at (authqkdA.north) {\footnotesize $\pi^\auth_A$};
\node[protocol2,red] (authqkdB) at (\p,2.5*\v) {};
\node[yshift=-1.5,above] at (authqkdB.north) {\footnotesize $\pi^\auth_B$};
\node[protocol3,violet] (qkdA) at (-\p-\c,2*\v) {};
\node[yshift=-1.5,above] (labelA) at (qkdA.north) {\footnotesize $\pi^\qkd_A$};
\node[protocol3,violet] (qkdB) at (\p+\c,2*\v) {};
\node[yshift=-1.5,above] (labelB) at (qkdB.north) {\footnotesize $\pi^\qkd_B$};
\node[protocol2,blue] (authotpA) at (-\p-2*\c,\v/2) {};
\node[yshift=-1.5,above] at (authotpA.north) {\footnotesize $\pi^\auth_A$};
\node[protocol2,blue] (authotpB) at (\p+2*\c,\v/2) {};
\node[yshift=-1.5,above] at (authotpB.north) {\footnotesize $\pi^\auth_B$};
\node[protocol3,green] (otpA) at (-\p-3*\c,\v) {};
\node[yshift=-1.5,above] at (otpA.north) {\footnotesize $\pi^\otp_A$};
\node[protocol3,green] (otpB) at (\p+3*\c,\v) {};
\node[yshift=-1.5,above] at (otpB.north) {\footnotesize $\pi^\otp_B$};

\node[draw,thick,dashed,fit=(authqkdA)(labelA)(otpA),inner sep=8] (compA) {};
\node[yshift=-1.5,above right] at (compA.north west) {\footnotesize
  Composed protocol $\pi_A$};
\node[draw,thick,dashed,fit=(authqkdB)(labelB)(otpB),inner sep=8] (compB) {};
\node[yshift=-1.5,above right] at (compB.north west) {\footnotesize
  Composed protocol $\pi_B$};

\node (p1) at (0,5*\v/2) {};
\node (p2) at (0,3*\v/2) {};
\node (p3) at (0,\v/2) {};
\node (p4) at (0,-\v/2) {};
\node (aliceUp) at (-\t,3*\v) {};
\node (aliceDown) at (-\t,\v) {};
\node (bobUp) at (\t,3*\v) {};
\node (bobDown) at (\t,\v) {};

\node (eveLeft) at (-\x,0) {};
\node (eveRight) at (\x,0) {};

\draw[sArrow] (aliceDown.center) to node[auto,pos=.35] {$x$} (otpA.west |- aliceDown);
\draw[sArrow] (otpA.east |- p3) to (authotpA.west |- p3);
\draw[sArrow] (qkdA.west |- aliceUp) to node[auto,pos=.9,swap] {$k_1$} (aliceUp.center);
\draw[sArrow] (qkdA.west |- 0,2*\v) to node[auto,pos=.5,swap] {$k_2$} (otpA.east |- 0,2*\v);
\draw[sArrow] (qkdA.west |- 0,\v) to node[auto,pos=.5,swap] {$k_3$} (authotpA.east |- 0,\v);
\draw[sArrow] (authotpA.east |- 0,0) to (eveLeft |- 0,0) to (eveLeft |- p4);
\draw[sArrow] (qkdA.east |- 0,\v) to (eveLeft |- 0,\v) to (eveLeft |- p3);
\draw[sArrow] (qkdA.east |- p1) to (authqkdA.west |- p1);
\draw[sArrow] (authqkdA.east |- 0,2*\v) to (eveLeft |- 0,2*\v) to (eveLeft |- p2);
\draw[sArrow] (keyLong.west |- 0,3*\v) to (authqkdA.east |- 0,3*\v);

\draw[sArrow] (keyLong.east |- 0,3*\v) to (authqkdB.west |- 0,3*\v);
\draw[sArrow] (eveRight |- p2) to (eveRight |- 0,2*\v) to (authqkdB.west |- 0,2*\v);
\draw[sArrow] (eveRight |- p3) to (eveRight |- 0,\v) to (qkdB.west |- 0,\v);
\draw[sArrow] (eveRight |- p4) to (eveRight |- 0,0) to (authotpB.west |- 0,0);
\draw[sArrow] (authqkdB.east |- p1) to (qkdB.west |- p1);
\draw[sArrow] (qkdB.east |- aliceUp) to node[auto,pos=.9] {$k'_1$} (bobUp.center);
\draw[sArrow] (qkdB.east |- 0,2*\v) to node[auto,pos=.5] {$k'_2$} (otpB.west |- 0,2*\v);
\draw[sArrow] (qkdB.east |- 0,\v) to node[auto,pos=.5] {$k'_3$} (authotpB.west |- 0,\v);
\draw[sArrow] (authotpB.east |- p3) to (otpB.west |- p3);
\draw[sArrow] (otpB.east |- bobDown) to node[auto,pos=.65] {$x'$} (bobDown.center);

\end{tikzpicture}

\caption[Constructing a secure channel]{\label{fig:secure.complete}
  (Color online)
  Composition of QKD, authentication and OTP protocols. For
  simplicity, we have drawn only one round of authentication as a
  subroutine of QKD as \textcolor{red}{$\pi^\auth_{AB}\left(
      \aK^{a_\qkd} \| \aC \right)$}. The \textcolor{violet}{QKD
    protocol $\pi^{\qkd}_{AB}$} constructs a shared key resource that
  produces the long key $(k_1,k_2,k_3)$. The second
  \textcolor{blue}{authentication protocol $\pi^\auth_{AB}$} then uses
  part of this key to construct another authentic channel, and the
  \textcolor{green}{OTP protocol $\pi^{\otp}_{AB}$} uses another part
  of this key to encrypt and decrypt the message sent on the channel.}
\end{figure*}
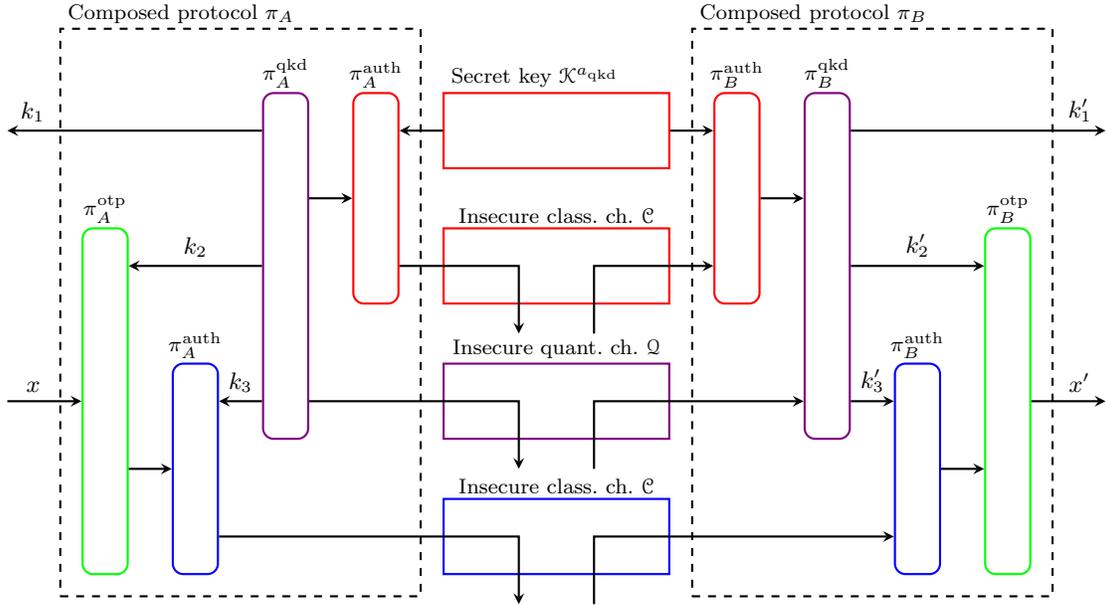


\section{Other cryptographic tasks}
\label{sec:other}

In our description so far we have adopted the composable view and
regarded cryptographic protocols as constructions, i.e., protocols
construct some resources given other resources.\footnote{See also
  Footnote~\ref{fn:resourcetheory} for other uses of resource theories
  in quantum mechanics.} Cryptographic protocols proposed in the
literature have not always been defined in this way, but instead are
specified in terms of particular security\-/like properties, e.g.,
that an adversary is unable to guess the content of an encrypted
message. Sometimes these properties can be rephrased as an ideal
system within the real\-/world ideal\-/world paradigm, as discussed in
\secref{sec:qkd.other.ac}. In this section we review some of the major
results in quantum cryptography from this perspective, i.e., we
present them as constructive statements, defining the resources
constructed and used by the protocols. For a broader review of quantum
cryptography, we refer to a recent survey by \textcite{BS16}.

\subsection{Secure quantum message transmission}
\label{sec:qmt}

From a theory of resources perspective, the task of securely
transmitting a quantum message from Alice to Bob is nearly identical
to the corresponding classical task, analyzed in \secref{sec:smt}. Here too, we
require the players to share a secret key resource and an insecure
channel, and the goal is to construct a secure channel \--- the only
difference being that the insecure and secure channels are both
quantum channels. We have already encountered insecure quantum
channels in \secref{sec:qkd}, where they were used for QKD. A secure
quantum channel is modeled analogously to a secure classical channel
as drawn in \figref{fig:secure.resource}, except that the messages
sent are quantum. We depict this in
\figref{fig:secure.quantum.resource}.

\begin{figure}[tb]

\begin{tikzpicture}[
sArrow/.style={->,>=stealth,thick},
largeResource/.style={draw,thick,minimum width=1.618*2cm,minimum height=2cm}]

\small

\def\z{.4}

\node[largeResource] (keyBox) at (0,0) {};
\node (alice) at (-3,0) {Alice};
\node (bob) at (3,0) {Bob};
\node (eve) at (0,-1.7) {Eve};
\node (juncLL) at (-3*\z,0) {};
\node (juncLR) at (-\z,0) {};
\node (juncRL) at (\z,0) {};
\node (juncRR) at (3*\z,0) {};

\draw[thick] (alice) to node[pos=.2,auto] {$\rho$} (juncLR.center);
\draw[thick] (juncLR.center) to (juncRL.center) to node[pos=.5]
(handle2) {} +(340:2*\z);
\draw[sArrow] (juncRR) to node[pos=.65,auto] {$\rho,\bot$} (bob);
\draw[sArrow,dotted] (0,0) to node[pos=.85,auto,swap] {$|\rho|$} (eve);
\draw[double] (handle2.center |- eve.north) to node[pos=.15,auto,swap] {$0,1$} (handle2.center);

\end{tikzpicture}

\caption[Secure quantum channel]{\label{fig:secure.quantum.resource}A
  secure quantum channel that allows Eve to learn the length of the
  quantum message \--- informally denoted as $|\rho|$ \--- and prevent
  Bob from receiving it: when Alice inputs a message $\rho$ at her
  interface, information about the length of the message is given to
  Eve, who can additional press a switch that either delivers Alice's
  message to Bob or provides him with an error message $\bot$
  instead.}
\end{figure}

The first protocols that construct such a secure quantum channel from a shared secret key and an insecure channel were proposed by \textcite{BCGST02}. They follow the same pattern as classical message transmission: one first encrypts the quantum message with a quantum OTP, then encodes it in a larger space so as to detect any errors that may be introduced by an adversary. However, contrary to the case of classical messages, there is no known way to view these two steps as two distinctive constructive statements, i.e., as a construction of an authentic channel from an insecure channel, and a second construction of a secure channel from an authentic channel. This means that the analysis has to include both aspects at the same time.

In \secref{sec:qmt.protocol} we explain this construction and in
\secref{sec:qmt.related} we review additional work on the topic. At
the end of this section, in \secref{sec:computational}, we 
revisit this topic from a computational perspective.

\subsubsection{Generic protocol}
\label{sec:qmt.protocol}

The classical OTP introduced in \secref{sec:ac.otp} can be seen as
randomly flipping each bit of the message. The quantum OTP
\cite{AMTW00,BR03} follows the same principle: one flips the bits and
phases of the message at random. For $x,z \in \{0,1\}^n$, let $X^x$
and $Z^z$ denote operators on $\left(\complex^2\right)^{\otimes n}$
which perform bit and phase flips in positions indicated by the
strings $x$ and $z$, respectively. The quantum OTP consists of
choosing $x$ and $z$ uniformly at random, and applying the
corresponding operation to the message. For any state $\rho_{MR}$
where $M$ is a register of size $2^n$, we thus have
\[ 
\frac{1}{2^{2n}}\sum_{x,z} Z^z X^x \rho_{MR} X^x Z^z = \tau_M
\otimes \rho_R,
\]
where $\tau_M$ is the fully mixed state and $\rho_R$ is the reduced
density operator of $\rho_{MR}$. We call the operators $Z^zX^x$
defined this way \emph{Pauli operators}.\footnote{This notation
  simplifies the presentation here, but deviates from the more
  commonly used definition of the Pauli operators as
  $i^{z \cdot x} Z^z X^x$, where $z \cdot x = \sum_j z_jx_j$ and $i$
  is the imaginary unit.}

The second ingredient needed to construct a secure quantum channel is
an error correcting code that is going to be used to detect errors in
the transmission, i.e., tampering by an adversary. Generally, an error
correcting code may be seen as a map from a message space $\hilbert_M$
to a larger (physical) space $\hilbert_C$. For simplicity, we model
the encoding for a code $\cC_k$ as first appending a state
$\ket{0} \in \hilbert_T$ to the message $\rho_M$, where
$\hilbert_C = \hilbert_M \otimes \hilbert_T$, followed by applying a
unitary $U_k$ to the resulting state, e.g.,
$\sigma_{C} = U_k \left( \rho_M \otimes \proj{0} \right)
\hconj{U}_k$.

To test whether an error occurred, one decodes the received state
$\tilde{\sigma}_{C}$ by applying the inverse operation $\hconj{U}_k$
and measures the $T$ register in the computational basis. If the
result is not $0$, then this is evidence for noise. We say that a code
detects an error $V$ if after decoding a message to which this error
was applied \--- i.e., $\tilde{\sigma}_{C} = V \sigma_{C} \hconj{V}$
\--- one always gets a measurement outcome different from
$0$. Furthermore, we call an error \emph{trivial} if it never affects
the code word, i.e., for any $\rho_M$,
\[ 
  \hconj{U}_k V U_k \left( \rho_M \otimes \proj{0} \right)
\hconj{U}_k \hconj{V} U_k = \rho_M \otimes \proj{0}.
\]

For an error $V$ to modify a message and yet not be caught, it must be
non-trivial and not detected by the code used. For the purpose of
constructing a secure channel according to the method of
\textcite{BCGST02}, it is sufficient to detect all Pauli
errors. \textcite{BCGST02} define a set of codes, which they call
\emph{purity testing codes}. They guarantee that with high probability
over the choice of code, all Pauli errors are either caught or
trivial. More precisely, a set $\{\cC_k\}_k$ of codes forms a family
of $\eps$\=/purity testing codes if, for any Pauli error $X^xZ^z$, the
probability over a uniformly random choice of $k$ that this error is
neither caught nor trivial is less than $\eps$.

The protocol for constructing a secure channel then works as
follows. The sender first encrypts the message with a quantum OTP,
i.e., a Pauli operator $Z^zX^x$ chosen uniformly at random according
to the secret key. Then a state $\ket{s}$ is appended to the message,
where $s$ is also chosen uniformly at random according to the
key. Finally, the resulting state is encoded with a unitary $U_k$
corresponding to the encoding operation of an element of a purity
testing code family $\{\cC_{k}\}_{k}$, where again $k$ is chosen
uniformly at random according to the secret key. Decryption works in
the obvious way: the receiver applies the inverse operator
$\hconj{U}_k$ and then measures the $T$ register. If the outcome is
not $s$, the message was jumbled and the player outputs an error
symbol. Otherwise, the receiver decrypts the message with the operator
$Z^zX^x$.

\subsubsection{Concrete schemes}
\label{sec:qmt.related}

\textcite{BCGST02} introduced the general family of protocols
described in \secref{sec:qmt.protocol}, and also provided a concrete
construction of a purity testing code family that has good
parameters. Following this seminal work, a variety of further
protocols for authentication of quantum messages have been proposed in
the literature, which are based on different codes. Authentication
using the signed polynomial code \cite{BCGHS06,ABE10}, the trap code
\cite{BGS13,BW16}, the Clifford code \cite{ABE10,DNS12,BW16} \---
which is a unitary $3$\-/design~\cite{Web15,Zhu17} \--- a unitary
$8$\-/design \cite{GYZ17} and a unitary $2$\-/design\footnote{Because
  any unitary $t$\-/design for $t \geq 2$ is a unitary $2$\-/design,
  it follows that any $t$\-/design constructs a secure quantum
  channel.}  \cite{Por17,AM17} are all instances of the family from
\textcite{BCGST02}, with alternative purity testing
codes.\footnote{Most of these works consider a construction where the
  message is first encoded with a purity testing code, and then
  encrypted. But as shown in \textcite{Por17}, this is equivalent to
  the original scheme of \textcite{BCGST02}, which reverses the order
  of these two operations.} To the best of our knowledge, only the
Auth-QFT-Auth scheme from \textcite{GYZ17} is not known to follow the
model of \textcite{BCGST02}.

Although most of these works provide some kind of security proof for
the protocol, only two papers consider a composable security
definition, namely \textcite{HLM11,Por17}. Both works show that the
family of protocols from \textcite{BCGST02} construct a secure quantum
channel from a shared secret key and an insecure quantum channel. It
should be noted however that \textcite{HLM11} considers a restricted
class of distinguishers \--- those that perform a substitution attack
(see \secref{sec:smt.auth}) \--- and \textcite{Por17} only analyzes a
subset of this family, for which the purity testing code family
detects all (rather than only non-trivial) errors with high
probability.\footnote{One refers to this as a \emph{strong} purity
  testing codes.} In fact both papers prove that one may additionally
recycle part of the key, which is discussed further in the following
section.

\subsection{Key reuse in classical and quantum message transmission}
\label{sec:recycle}

As already mentioned in Secs.~\ref{sec:smt.auth} and
\ref{sec:qmt.related}, part of the key used in the constructions of
secure channels may be recycled, i.e., at the end of the protocol it
can be added back to a pool of secret key bits. For example, in the
case of classical message authentication analyzed in
\secref{sec:smt.auth}, the sender appends a tag $y = h_k(x)$ to the
message $x$. The value of the tag $y$ depends on the shared secret key
$k$, and every bit of the tag leaks (at most) a bit of the secret key
to the adversary. But the key is longer than the tag, so the bits
which are not leaked may be reused. It is however vital that they are
not recycled too soon: if the sender reuses part of the key before the
receiver obtains the (authenticated) message, the adversary may learn
these bits and use this information to successfully change the message
and authentication tag.

To recycle key bits in classical or quantum message transmission, the
real system is changed as follows. Firstly, the players need an extra
resource, a $1$-bit backwards authentic channel, allowing the receiver
to tell the sender whether the message was successfully received or
not. Once this confirmation is sent, the receiver may recycle part of
the key, i.e., it is output by the corresponding converter. And once
this confirmation is received by the sender, she also recycles the
same part of the key. The ideal resource constructed in this way
corresponds to the parallel composition of a secure (or authentic)
channel and a secret key resource, drawn in
\figref{fig:keyrecycle.ideal}. As previously, the adversary may
control whether the message is delivered on the secure channel. And
since the amount of key recycled may depend on the adversary's
behavior as well---namely by allowing or preventing the message from
being delivered, we model the resource as being equipped with a switch
to control how much key the players get.

\begin{figure}[tb]

\begin{tikzpicture}[
sArrow/.style={->,>=stealth,thick},
thinResource/.style={draw,thick,minimum width=3.2cm,minimum height=1cm}]

\small

\def\t{2.5} 
\def\v{.8}
\def\e{2.2} 
\def\z{.4}

\node (a1) at (-\t,\v) {};
\node (a2) at (-\t,-\v) {};
\node (b1) at (\t,\v) {};
\node (b2) at (\t,-\v) {};
\node (eve) at (0,-\e) {};

\node[thinResource] (keyBox) at (0,\v) {};
\node[draw] (key) at (0,\v) {key};
\node[yshift=0,above right] (keyname) at (keyBox.north west) {\footnotesize
  Secret key $\aK$};

\draw[sArrow] (key) to node[auto,swap,pos=.85] {$k$} (a1);
\draw[sArrow] (key) to node[auto,pos=.85] {$k$} (b1);
\draw[double] (eve.center) to node[pos=.08,auto,swap] {$0,1$} (key);

\node[thinResource] (channel) at (0,-\v) {};

\node (juncLL) at (-3*\z,-\v) {};
\node (juncLR) at (-\z,-\v) {};
\node (juncRL) at (\z,-\v) {};
\node (juncRR) at (3*\z,-\v) {};

\draw[thick] (a2) to node[pos=.15,auto] {$\rho$} (juncLR.center);
\draw[thick] (juncLR.center) to (juncRL.center) to node[pos=.5]
(handle2) {} +(340:2*\z);
\draw[sArrow] (juncRR) to node[pos=.6,auto] {$\rho,\bot$} (b2);
\draw[sArrow,dotted] (-2*\z,-\v) to node[pos=.8,auto,swap] {$|\rho|$}
(-2*\z,-\v |- eve.center);
\draw[double] (handle2.center |- eve.center) to node[pos=.2,auto,swap] {$0,1$} (handle2.center);

\node[yshift=0,above right] at (channel.north west) {\footnotesize
  Secure channel $\aS$};

\node[inner sep=0] (justadot) at (0,1.5) {};

\node[draw,thick,dashed,fit=(channel)(keyBox)(justadot),inner sep=8] (alltogether) {};

\end{tikzpicture}

\caption[Ideal secure channel with
key]{\label{fig:keyrecycle.ideal}The ideal system for a secure channel
  with key recycling: it consists of a secure channel $\aS$ and key
  resource $\aK$. The adversary controls the length of the 
  recycled key (through her input to $\aK$) as well as whether the
  receiver obtains the message (through her input to $\aS$).}
\end{figure}
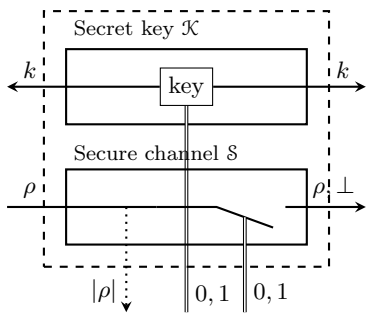

In the case of authentication of classical messages, \textcite{WC81}
already proposed that part of the key can be safely recycled. Here, if
the $2$-universal hash function has the special form
$h_{k_1,k_2}(x) = {f_{k_1}(x) \xor k_2}$, where $k_2$ is a bit string
of the same length as the tag, then $k_1$ may be recycled, but a new
$k_2$ is needed for every message. It was proven by \textcite{Por14}
that this scheme is composable and constructs the ideal resource
described above.

In the case of quantum messages, roughly the same holds in the case
where the message fails the authentication: the number of bits of key
leaked depend on the length of the ciphertext and the rest can be
recycled.\footnote{If the ciphertext is $n$ qubits long, about $2n$
  bits of key are lost, see \textcite{Por17} for the exact
  parameters.} But in the case where the message is accepted, the
players can recycle more key. This holds because of the no\-/cloning
principle of quantum mechanics: if the receiver holds the correct
ciphertext, then the adversary cannot have a copy of it, and thus does
not hold any information about the key either. It was first shown by
\textcite{HLM11} that nearly all of the key could be recycled in the
case where the message is accepted. Then \textcite{Por17} showed that
every bit of the key can indeed be recycled. This is not known to hold
for all schemes that construct a secure quantum channel, but so far
only for those that use strong purity testing codes \cite{Por17}.

\subsection{Delegated quantum computation}
\label{sec:dqc}

The setting in which a client, typically with bounded computational
resources, asks a server to perform some computation for her is called
\emph{delegated computation}. The client might not want the server to
learn what computation it is performing for her, and wish to run a
protocol that hides the underlying computation \--- this property is
called \emph{blindness} in the literature. Furthermore, the client
might want to verify that the server correctly performed the
computation she required \--- this is known as \emph{verifiability}.

The task of delegating a quantum computation was first studied by
\textcite{Chi05}, with the goal of achieving blindness. In follow-up
works, the requirements on the client's information-processing
abilities were reduced. \textcite{BFK09} proposed the first protocol
for blind delegated quantum computation that does not require the
client to have quantum memory, but only the ability to prepare
different pure states. This result was extended in \textcite{FK17} to
include verifiability as well.

Delegated quantum computation (DQC) was formalized as a constructive
statement by \textcite{DFPR14}. The authors modeled a DQC protocol
that achieves both blindness and verifiability as constructing a
resource $\aS^{\blind}_{\verif}$ that works as follows. It first
receives a description of the required computation as a state $\psi$
from the client. Every computation necessarily leaks some information
to the server, e.g., an upper bound on the computation size, so the
resource computes this leaked information $\ell$ and outputs it at the
server's interface. The server can then decide if it will cheat \---
in which case the client will however get an error message \--- or
output the correct result of the computation, which is evaluated by
applying an operator $\cU$ to the input. This is depicted in
\figref{fig:dqc.ideal}. A DQC protocol constructs such a resource from
nothing more than a shared communication channel between client and
server.

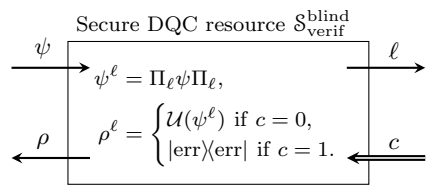
\begin{figure}[tb]

\begin{tikzpicture}[
sArrow/.style={->,>=stealth,thick},
blocResource/.style={draw,minimum width=4cm,minimum height=1.9cm},
brnode/.style={minimum width=3.4cm,minimum height=.2cm}]

\small

\def\u{2.75}
\def\v{.6}

\node[brnode] (n1) at (0,\v) {};
\node[brnode] (n3) at (0,-\v) {};
\node[blocResource] (S) at (0,0) {};
\node[yshift=-1.5,above right] at (S.north west) {{\footnotesize
  Secure DQC resource} $\aS^{\blind}_{\verif}$};
\node[text width=3.3cm] at (0,.1) {\footnotesize
  \begin{align*}
   \psi^\ell & = \Pi_\ell \psi \Pi_\ell, \\
   \rho^\ell & = \begin{cases}
    \cU(\psi^\ell) \text{ if $c=0$,} \\ \proj{\err}
    \text{ if $c=1$.} \end{cases} \end{align*}};

\node (a1) at (-\u,\v) {};
\node (a3) at (-\u,-\v) {};
\node (b1) at (\u,\v) {};
\node (b3) at (\u,-\v) {};

\draw[sArrow] (a1.center) to node[auto,pos=.4] {$\psi$} (n1);
\draw[sArrow] (n3) to node[auto,swap,pos=.6] {$\rho$} (a3.center);

\draw[sArrow] (n1) to node[auto,pos=.6] {$\ell$} (b1.center);
\draw[sArrow,double] (b3.center) to node[auto,swap,pos=.4] {$c$} (n3);

\end{tikzpicture}

\caption[Ideal DQC resources]{\label{fig:dqc.ideal}Ideal DQC
  resource. The client has access to the left interface, and the
  server to the right interface. The server obtains some information
  $\ell$ about the input, and can decide if the client gets the
  correct outcome or an error by inputting a bit $c$.}
\end{figure}

A weaker resource that provides only blindness but not verifiability
can be obtained by increasing the power of the server at its interface
of the resource: instead of inputing a bit that decides if the client
gets the correct outcome, the server can decide what output the client
gets, but still only receives the leaked information $\ell$
\cite{DFPR14}.

In \textcite{DFPR14} the protocols from \textcite{BFK09,FK17} were
shown to satisfy the corresponding constructive definitions. These
protocols still require the client to prepare a few different single
qubit quantum states and send them to the server. In order to better
analyze this requirement, \textcite{DK16} decomposed the construction
of a DQC resource in two steps. First, they consider a resource which
provides the server with the random states it needs (and the client
with a description of which state was given to the server), then the
DQC protocol constructs the DQC resource given this state preparation
resource. This decomposition then allowed \textcite{GV19} to design a
protocol that constructs the required state preparation resource for
an entirely classical client \--- this was achieved by sacrificing
information\-/theoretic security for computational security. Composing
this with a DQC protocol, one gets DQC for an entirely classical
client, albeit with computational security. This is believed not to be
possible with information\-/theoretic security~\cite{ACGK19}. We note
however that \textcite{GV19} make a non\-/standard assumption about
available resources, without which such a result does not seem
possible \cite{BCCKLMW20}.

It is instructive to compare this to early definitions of blindness,
e.g., those from \textcite{BFK09,FK17}. There, the requirement is that
the server learns nothing about the computation except for the allowed
leaked information $\ell$. Roughly speaking, this means that the state
$\rho^{\psi^\ell}$ held by the server at the end of the protocol \---
where $\psi^\ell$ is an input that leaks information $\ell$ \--- must
be such that \begin{equation}\label{eq:dqc.standalone}
  \rho^{\psi^\ell} \approx \rho^\ell,\end{equation} i.e., it can only
depend on $\ell$, but not on any other part of the input. If we
compare this to the constructive definition from \textcite{DFPR14}, in
which the distinguisher has access to both the server's interface and
the client's interface of the resources, \eqnref{eq:dqc.standalone}
corresponds to the special case where we do not maximize over all
distinguishers, but only those that ignore the output received by the
client. Following \secref{sec:qkd.other.ac} one may express this as a
restriction on the resource constructed instead of a restriction on
the distinguisher: requiring a DQC protocol to satisfy
\eqnref{eq:dqc.standalone} is (mathematically) equivalent to requiring
it to construct an ideal resource that does not provide the client
with the result of the computation.

\subsection{Multi-party computation}
\label{sec:mpc}

In this section we consider a setting where multiple mutually
distrustful parties wish to evaluate a (possibly randomized) function
to which each of them provides an input \--- or they wish to jointly
evaluate a CPTP map on shared quantum inputs. They however do not want
the other parties to learn anything about their input other than
what can be learnt from the output. Furthermore, they also want to
guarantee that if they get an output, then this is the correct
output. An example is a function that outputs which player $i$ has the
largest input $x_i$ \--- e.g., the players want to know who earns more
without revealing their salary to the others. Another example is
generating a random coin flip, in which case no input is
required. Generally speaking, multi-party computation corresponds to constructing an ideal resource
which first takes the inputs from all parties and then provides them
with the correct output.

\subsubsection{Bit commitment}
\label{sec:mpc.BC}

A bit commitment resource is a system in a two-party setting, which
forces a player (say, Alice) to commit to a value, but without
revealing this value to the other player (say, Bob). At some later
point, the commitment is ``opened'', so that Bob may know to what
value Alice committed. More precisely, Alice sends a bit $b$ to the
resource, and Bob is notified that Alice is committed to a
value. Alice may then send an open command to the resource, at which
point $b$ is delivered to Bob. In the classical setting, such a
resource cannot be constructed from communication channels
alone~\cite{CF01,MR11}, but extra resources such as a common reference
string \--- a random string shared by all parties \--- are
needed~\cite{CF01}.

The argument from \textcite{MR11} has been extended in
\textcite{VPdR19} to prove that even if the players use quantum
protocols and even if the adversary is computationally bounded, has
only bounded or noisy storage and is restricted by relativistic
constraints,\footnote{This means that the adversary cannot send
  information between two points faster than light takes to travel
  between the two points, see \secref{sec:relativistic}.} it is still
impossible to construct a bit commitment resource without further
setup assumptions than communication channels.

It has been suggested that one could construct bit commitment if one
takes relativity into account, i.e., that messages cannot be sent
faster than the speed of light \cite{Ken99,Ken12,KTHW13}. However,
these protocols do not construct a bit commitment resource:
\textcite[Appendix A]{Kan15}\footnote{The proof from
  \textcite[Appendix A]{Kan15} uses the same attack as
  \textcite{BCMS98}, where it is shown that some non\-/composable
  definitions of bit commitment appearing the classical literature
  cannot be used to force a quantum player to commit to a measurement
  outcome.} proves that if one composes these relativistic bit
commitment protocols with the protocol from \textcite{Unr10} to
construct oblivious transfer from bit
commitment,\footnote{\emph{Oblivious transfer} and the construction
  from \textcite{Unr10} are discussed in \secref{sec:mpc.OT}.} then
the result is not a secure oblivious transfer protocol. It has now
been proven that taking relativity into account is not sufficient to
achieve bit commitment \cite{VPdR19}, which we discuss in more detail
in \secref{sec:relativistic}.

\subsubsection{Coin flipping}
\label{sec:mpc.CF}

Another well studied resource is that of coin flipping, which flips a
random coin and provides both players with the result. The
impossibility proof for bit commitment from \textcite{MR11} can be
adapted to show that coin flipping and biased coin flipping (where a
player is allowed to partially bias the flip) are also impossible
without further assumptions.  Note that this proof is valid
independently of whether one considers classical or quantum
strategies. A direct proof for the impossibility of coin flip in the
quantum and relativistic setting \--- even in the case of
computational- and memory\-/bounded adversaries \--- is given in
\textcite{VPdR19}.

\emph{Coin expansion} is a weaker task, in which one constructs a
resource that produces a sequence of coin flips from a weaker resource
that produces fewer coin flips. This has been shown to be impossible
for classical players with information\-/security \cite{HMU06,SM16},
but is possible with computational security \cite{HMU06} and remains
open in the quantum case.


\subsubsection{Two-party function evaluation and oblivious transfer}
\label{sec:mpc.OT}

It has been shown by \textcite{IPS08} that a resource that evaluates
any classical probabilistic polynomial\-/time function with two inputs
can be constructed given an \emph{oblivious transfer} resource, i.e.,
a system that receives two strings $s_0,s_1$ from one player, Alice, a
bit $c$ from the second player, Bob, and sends Bob $s_c$.

In the quantum setting, it is possible to construct an oblivious
transfer resource from a \emph{bit commitment} resource. The
construction of oblivious transfer from bit commitment was first
proposed by \textcite{CK88}, adapted to noisy channels in
\textcite{BBCS92}, and proven secure by \textcite{Unr10}. Combining
this result with \textcite{IPS08}, it follows that bit commitment is
universal for classical two-party computation \cite{Unr10}.

It is however not possible to construct an oblivious transfer resource
from nothing but communication channels, even if the adversary is
computationally bounded, has only bounded or noisy storage and is
restricted by relativistic constraints~\cite{LdR21}.

\subsubsection{Everlasting security}
\label{sec:mpc.ever}

\textcite{Unr13} studied multi-party computation in the setting of
\emph{everlasting security}. This means that one relies upon a
computational assumption, but this assumption has to be broken
\emph{during} the execution of the protocol for an adversary to break
the scheme. If this is not the case, then even a computationally
unbounded adversary cannot get any significant advantage after the
protocol has terminated. This is generally not satisfied by
computational encryption schemes, because an adversary could obtain a
ciphertext and wait for an advancement in algorithms to break the
scheme and obtain the message. But if a computational authentication
scheme is executed, then the adversary must be able to perform the
hard computation before the message is received and authenticated.

Composition in such a setting is not straightforward, and
\textcite{Unr13} provides the necessary definition a scheme must
satisfy to be composable. They show how to perform
authentication given a signature card, which, when composed with QKD
and secure encryption as in \secref{sec:smt}, results in a secure
channel. They also provide a way to perform bit commitment based on signature
cards. Composing this with the protocols from \secref{sec:mpc.OT}
allows one to perform any multi-party computation with everlasting
security.

\subsubsection{Multi-party quantum computation}
\label{sec:mpc.MPQC}

The tasks studied so far in this section are concerned with
multi-party evaluation of a classical function \--- so\-/called
multi-party computation (MPC) \--- but using quantum communication and
computation to possibly achieve what cannot be done classically. The
problem of multi-party \emph{quantum} computation (MPQC) generalises this to the case where the
inputs and outputs are quantum \cite{CGS02,BCGHS06,DNS12,DGJMS20,LRW20,ACCHLS21}. The relation between inputs and outputs is then most generally described by a CPTP. It is standard to
use a composable framework for analyzing classical MPC
\cite{CDN15}. But to the best of our knowlege, the only work on MPQC
that mentions that the results hold in a composable framework is
\textcite{BCGHS06} \--- and they only provide a proof sketch. All
other works assume that the dishonest party performs their attack in
an isolated way, only interacting with the environment (the
distinguisher) before the protocol starts and after the protocol
ends. This is the so\-/called \emph{stand\-/alone} security model, and
protocols proven secure in such a model do not necessarily compose
concurrently with other protocols, in particular, they might not be
secure if two instances of the same protocol are run in
parallel. Nonetheless, for MPQC we do not know of any
attacks on protocols run concurrently, and it is plausible that
exactly the same result go through in a composable security framework.

The ideal resource one wishes to construct in MPQC receives the inputs
from all parties, performs the quantum computation, and then provides
each player with their part of the output. The works of
\textcite{CGS02,BCGHS06,LRW20} consider an ideal resource which is
guaranteed to provide the output to the honest
players. \textcite{CGS02} first show that this can be achieved if the
fraction of dishonest parties is $t < n/6$, where $n$ is the total
number of players. In \textcite{BCGHS06} this is improved to $t < n/2$
cheating parties. \textcite{LRW20} decrease the number of qubits and
communication complexity needed to get the same.

In \textcite{DNS12,DGJMS20,ACCHLS21} the ideal resource is defined
such that it first provides the dishonest parties with their share of
the output. They then provide a bit to the ideal resource, which
indicates whether the honest parties should also receive their output
or an abort symbol instead. This is called
\emph{unfairness}. Weakening the ideal resource in this way allows the
number of dishonest parties to be any $t < n$. \textcite{DNS12} first
show how to do this in the two-party case. \textcite{DGJMS20} extend
this to the multi\-/party setting.  \textcite{ACCHLS21} improve the
protocol to identify parties that abort, so that if an abort occurs,
the faulty party can be excluded and the others start again without
them.

We note that all these protocols assume that classical MPC is
available as a resource \--- usually, for the same
number of dishonest players and the same abort conditions as the
constructed MPQC. So all these results require the same setup
assumptions as the corresponding classical MPC. For example, for
$t < n/3$ and guaranteed output one can do classical MPC assuming only
pairwise secure channels between the players~\cite{CDN15}. For
$t < n/2$ and guaranteed output one additionally needs broadcast for
information\-/theoretic security, but pairwise authentic channels are
sufficient for computational security~\cite{CDN15}. If we drop
fairness, then in the case of computational security one gets unfair
security for any $t < n$ if one assumes oblivious transfer \cite{GMW87} or
a common reference string \cite{CLOS02}, and information\-/theoretic
security if one assumes common shared randomness \cite{IOZ14}.

\subsubsection{One-time programs}
\label{sec:mpc.OTP}

A special class of multi-party functionalities that have been studied
in more detail are non\-/reactive, sender\-/oblivious functions, i.e.,
one player is labeled ``sender'', another ``receiver'', and only the
receiver obtains the output of the function. This special structure
allows for non\-/interactive protocols to construct a resource that
computes such a function: communication goes only from the sender to
the receiver. The receiver can use the information obtained to
evaluate the function on one input. But by definition of the ideal
resource, he may not repeat this on a second input. The corresponding
resources are sometimes called \emph{one-time programs}.
\textcite{GISVW10} gave a construction for one-time programs, which
starts however from a resource that is similar to oblivious
transfer,\footnote{Note that oblivious transfer allows the player
  preparing the two strings $s_0,s_1$ to learn whether the other
  player has queried $s_c$. With only one-way communication of
  one-time programs, the resource used cannot allow this, but
  otherwise, it is identical to an oblivious transfer resource.} which
has been called \emph{one-time memory} or \emph{hardware token}
\cite{GISVW10} \--- since it could be implemented given hardware
assumptions, e.g., a one-time memory that contains the two strings
$s_0,s_1$, but self-destructs after producing an output.

These results have been generalized to the quantum setting by
\textcite{BGS13}, who show that one can construct quantum one-time
programs given access to the same one-time memory resources as for
classical one-time programs. More precisely, \textcite{BGS13} show
that any completely positive, trace\-/preserving map
$\Phi : \hilbert_A \otimes \hilbert_B \to \hilbert_C$ can be evaluated
with a non\-/interactive protocol by two distrustful parties holding
the inputs of registers $A$ and $B$, respectively, provided that only
one player is expected to receive the output in register~$C$.

\subsection{Relativistic cryptography}
\label{sec:relativistic}

So far we have predominantly discussed protocols whose security is based on the laws of quantum theory. One may however exploit further physical laws, such as those of special relativity. The latter imply an upper bound on the velocity by which information can spread --- the velocity of light. The combination of quantum information theory and relativity, apart from its relevance for fundamental questions \cite{Peres_Terno_RMP}, opens the possibility to achieve certain cryptographic tasks that are provably impossible based on quantum theory alone.

An example for this is relativistic bit commitment \cite{Ken99}, which
we mentioned already briefly in \secref{sec:mpc.BC}. Another one is
coin flipping. Here the two players Alice and Bob each have a trusted
agent at two locations $L_1$ and $L_2$. At $L_1$, agent $A_1$ is
instructed to provide agent $B_1$ with a random bit $a$, and at $L_2$,
agent $B_2$ provides $A_2$ with a random bit $b$. The agents then
inform Alice and Bob about these values, who can then output
$a \xor b$ as the result of the coin flip. The distance between the
locations $L_1$ and $L_2$ must be chosen large enough to ensure that,
if Bob is cheating, he cannot wait until he learns $a$ and then choose
$b$ depending on that value. Likewise, Alice cannot cheat for the same
reason. The output $a \xor b$ is hence uniformly random, provided that
at least one of the players chooses their bit uniformly at random.

This protocol does however not construct a coin flip resource, because
the parallel execution of two instances of the protocol does not
behave identically to two coin flip resources in parallel. Suppose
that Alice and Bob are running the protocol as well as Bob and
Charlie, who send their agents to the same locations $L_1$ and
$L_2$. Then at $L_1$, $A_1$ gives her random bit $a$ to $B_1$, who
gives a copy to $C_1$. And at $L_2$, $C_2$ gives his bit $c$ to $B_2$,
who gives a copy to $A_2$. Alice and Charlie then end up with exactly
the same coin flip $a \xor c$. But if we were to run two coin flip
resources in parallel, we would obtain two independent bit flips.

Running the same kind of attack on the relativistic bit commitment
protocols \cite{Ken99,Ken12,KTHW13}, Bob may forward Alice's
commitment to Charlie, and convince Charlie that he is committed to a
known bit, whereas in reality he does not know the commitment, and
thus does not satisfy the requirement of the bit commitment
resource. The same principle was used by \textcite{BCMS98} to prove
that some non\-/composable definitions of bit commitment appearing the
classical literature cannot be used to force a quantum player to
commit to a measurement outcome. This technique was then used by
\textcite[Appendix A]{Kan15} to prove that composing the relativistic
bit commitment protocols mentioned perivously with the oblivious
transfer protocol from \textcite{Unr10} is insecure. More precisely,
in the attack from \textcite{BCMS98,Kan15}, the committer does not
measure her state as required by the protocol, but runs the protocol
in superposition and only measures the strings that she needs to send
to the receiver as part of the commitment protocol. If she is asked to
open the commitment, she can still perform the required measurement
and open correctly. But if she is not asked to open, she can ``undo''
this measurement and recover the original state.

Relativistic bit commitment (and coin flipping) was analyzed more
systematically in \textcite{VPdR19,Pro20} using the Abstract
Cryptography framework \cite{MR11}. More precisely, the authors
instantiated the systems model from \textcite{MR11} with the Causal
Boxes framework \cite{PMMRT17}, which can model information with
positions in space\-/time. The resulting framework was used to prove
both impossibility and possibility results for relativistic
cryptography.

\textcite{VPdR19} show that it is impossible to construct a biased
coin flip resource between two distrustful players without assuming
any resources to help them, even in a relativistic setting.  Since
such a biased coin flip can be constructed from bit commitment
\cite{Blu83,DM13}, this immediately implies that it is also impossible
to construct a bit commitment resource in a relativistic setting. The
impossibility results also hold against adversaries that are
computationally bounded or have bounded storage
\cite{VPdR19}. \textcite{Pro20} analyzed what extra resources one can
assume to have in the real world to get positive results. The
techniques from \textcite{VPdR19} were extended in \textcite{LdR21} to
prove that oblivious transfer is also impossible without other setup
assumptions than communication channels, even if the adversary is
computationally bounded and has bounded or noisy quantum storage.

Another task for which special relativity is taken into account is
\emph{position verification}: a prover wishes to convince some
verifier that she is in a specific location \cite{CGMO09}. Protocols
based on relativity have been designed for other tasks as well, e.g.,
position verification and authentication \cite{BCFGGOS14,Unr14}. Here,
\textcite{BCFGGOS14} showed that such a task is impossible in the
presence of multiple colluding quantum adversaries that share
entanglement. They also consider a model in which holding shared
entanglement is not allowed, and propose a protocol for position
verification in this model. Similarly, \cite{Unr14} proposes a
protocol for position verification in the random oracle model, but
with no restriction on entanglement or memory. But neither of these
results provides a composable security proof, so it remains open to
prove exactly what these protocols achieve (see also
\secref{sec:open.other}).

\subsection{Secure quantum message transmission with computational security}
\label{sec:computational}

Most of the quantum cryptography literature is dedicated to
\emph{information\-/theoretic} security because the main motivation
of this field of research is to reduce cryptographic security to
physical principles. This means that regardless of the computational
ability of the adversary, the scheme may not be broken, as it does not
leak any information about the key or message. It is nevertheless
sensible to consider \emph{computational} security, as certain
cryptographic tasks may only be possible under such restricted
security guarantees (see \textcite{ABFGSSJ16} and the references
therein). In such a paradigm, a scheme may not be broken by an
adversary that is computationally bounded, but with unlimited
computational power it may be possible to obtain secret keys or read
private messages.

Composable frameworks such as \textcite{PW00,PW01,Can01} and the
quantum version by \textcite{Unr10} all define both computational and
information\-/theoretic security. However, they only define security
\emph{asymptotically}, i.e., a protocol is parametrized by some
security parameter $n$ \--- typically, this might correspond to the
number of signals exchanged between the parties or length of a tag
appended to a message \--- and security is proven in the limit as
$n \to \infty$. Abstract Cryptography \cite{MR11} on the other hand
considers \emph{finite security}, i.e., a security statement is made
for every $n$ (the limit is ignored and may not even be
well\-/defined).

Since any implementation is necessarily finite \--- the players fix a
value for the security parameter $n$ which they consider to be
sufficient and implement the corresponding protocol \--- an asymptotic
statement is arguably of limited interest in practice. For this reason
a paradigm known as \emph{concrete security} has been proposed
\cite{BDJR97}, in which parameters and reductions are given
explicitly instead of being hidden in $O$-notation and poly-time
statements. This allows a user to recover exact bounds for every $n$,
instead of only being provided with the limit values.

Concrete security is however still intrinsically asymptotic, since
adversaries are required to be \emph{poly-time} in $n$, protocols,
reductions and simulators need to be \emph{efficient} in $n$, errors
have to be \emph{negligible} in $n$, and such concepts are all defined
asymptotically. In finite security, one analyzes the security of a
protocol for individual values $n = n_0$. Hence concepts such as
poly-time, efficiency or negligibility are not necessarily
well\-/defined in a finite analysis, and cannot be part of a security
definition.

In \secref{sec:computational.def} we explain how to define finite
computational security. This follows the paradigm of
AC~\cite{MR11,Mau12,MR16}, and has been used in e.g.,
\textcite{MRT12,CMT13,BMPZ19}. In \secref{sec:computational.qmt} we
review the results from \textcite{BMPZ19} on computational security of
quantum message transmission (QMT). And in
\secref{sec:computational.literature} we discuss some asymptotic
game\-/based security definitions for QMT that have been proposed in the
literature~\cite{AGM18}.

\subsubsection{Defining composable and finite computational security}
\label{sec:computational.def}

To adapt the framework described in \secref{sec:ac} to capture
computational security, one needs to change the pseudo\-/metric used
to distinguish systems. We first recall \eqnsref{eq:adv2} and
\eqref{eq:adv4} from \secref{sec:ac.interpretation}, namely that the
distinguishing advantage for a distinguisher $\fD$ is defined as
  \begin{equation}
  \label{eq:adv2.bis} 
    d^{\fD}\left(\aR,\aS\right) \coloneqq \left| \Pr[\fD(\aR) = 0] - \Pr[\fD(\aS) = 0] \right|,
  \end{equation}
and the distinguishing advantage for a class of distinguishers $\bD$ is
given by
  \begin{equation}
  \label{eq:adv4.bis} 
  d^{\bD}\left(\aR,\aS\right)  \coloneqq \sup_{\fD \in \bD} d^{\fD}\left(\aR,\aS\right).
\end{equation}

So far in this work we have taken $\bD$ to be the set of all
distinguishers. If a protocol is only computationally secure,
\eqnref{eq:adv4.bis} could be large, since some (``inefficient'')
distinguisher might be able to distinguish the real and ideal
system. So instead of bounding the supremum over all distinguishers as
in \eqnref{eq:adv4.bis}, we bound \eqnref{eq:adv2.bis} for all $\fD$,
i.e., one needs to find some function $f : \bD \to \reals$ such that
for all $\fD$,\footnote{In asymptotic security one may still use
  \eqnref{eq:adv4.bis} instead of \eqnref{eq:adv5}, but replace $\bD$
  by the set of all efficient distinguishers. This is not
  well\-/defined in the finite setting, since \emph{efficiency} is
  only defined asymptotically.}
\begin{equation}
  \label{eq:adv5} 
  d^{\fD}\left(\aR,\aS\right) \leq f(\fD).
\end{equation}

Typically, such a bound is given by a \emph{reduction}, i.e., one
proves that if a distinguisher $\fD$ can distinguish the real from the
ideal system, then $\fD$ may be used to solve some problem believed to
be hard. In \eqnref{eq:adv5}, $f(\fD)$ may then correspond to the
probability that this problem may be solved using $\fD$.\footnote{A
  detailed explanation of this paradigm for modeling computational
  security by a reduction to computationally hard problems is provided
  in \textcite{Rog06} within a classical asymptotic model.}

Note that information\-/theoretic security corresponds to the special
case where one can show that $f(\fD)$ is small for all $\fD$, i.e., a
security proof with error $\eps$ means that \eqnref{eq:adv5} is
bounded by $f(\fD) = \eps$ for all $\fD$.



For a longer exposition on finite computational security we refer the
reader to \textcite{BMPZ19}.

\subsubsection{Secure quantum message transmission}
\label{sec:computational.qmt}

To construct a secure quantum channel with information\-/theoretic
security, the secret key shared by the honest players needs to be
longer than twice the length of quantum message sent (see
\secref{sec:qmt}). Although QKD or key recycling (see
\secref{sec:recycle}) may be used to obtain more key, this only works
when the noise on the channel is sufficiently low or the adversary
decides not to tamper with the messages. Should the noise be too high,
the used key is irremediably lost and the honest players may run out
of key and not be able to communicate securely anymore.

With computational security, it is believed\footnote{Computational
  security always relies on the belief that some problem is hard to
  solve. A cryptographic security proof then consists in showing that
  if an adversary can break the scheme, then this adversary can also
  solve the hard problem.} that the key can be much shorter than the
message, essentially allowing the same key to be used over and over,
even when conditions do not allow for recycling. For example, it is
believed that one can construct (quantum resistant) pseudo\-/random
function (PRF) families \cite{Zha12}, i.e., there exist families of
functions $\{f_k : \{0,1\}^m \to \{0,1\}^n\}_{k \in \cK}$ such that if
$k$ is chosen uniformly at random then the output of $f_k$ can not be
distinguished by a computationally bounded player from a random oracle
(RO) that outputs a uniform random string for every new input. If
Alice and Bob share a secret key $k$ which they use to pick a function
$f_k$, then they may encrypt their first message using $f_k(1)$ as
key, encrypt their second message using $f_k(2)$ as key, etc. To a
computationally bounded adversary not knowing $k$, the encryption keys
$f_k(i)$ would look random, and hence, by composition, a scheme
requiring a uniformly random secret key would be secure when used with
these pseudo\-/random keys.

Exactly this was done in \textcite{BMPZ19}, where the authors compose
a PRF family with the information\-/theoretic quantum message
transmission (QMT) protocol from \secref{sec:qmt}.\footnote{A variant
  of this protocol that allows the adversary to jumble the order of
  the messages was first proposed in \textcite{AGM18}, but security
  was only proven using asymptotic game\-/based definitions \--- see
  \secref{sec:computational.literature}.} They prove that if the error
of the PRF is bounded by
\[
 d^{\fD}(\prf,\ro) \leq \eps^{\prf}(\fD),
\]
and if the QMT protocol has error $\eps^{\qmt}$ and is used to send at
most $\ell$ messages, then the composed protocol essentially
constructs $\ell$ copies of the secure channel from
\figref{fig:secure.quantum.resource}, and the error of this
construction is bounded by
\[\eps(\fD) = \ell \eps^{\qmt} + \eps^\prf(\fD'),\]
where $\fD'$ is the same distinguisher as $\fD$ with the addition that
it can perform $\ell$ extra encryption and decryption operations.

\subsubsection{Relation to other security definitions}
\label{sec:computational.literature}

Computational security for secure message transmission is often
defined with asymptotic game\-/based definitions, e.g., an adversary
chooses two plaintexts, receives a ciphertext for one of the two and
has to guess to which of the two plaintexts it corresponds. In order
to model situations where the same keys can be used to encrypt and
decrypt other messages that may be accessible to the adversary, she is
also given oracle access to either encryption or decryption functions
at various points of the game \cite{BDPR98,KY06}. These definitions
have been adapted to the quantum case in
\textcite{BJ15,ABFGSSJ16,AGM18}.

Before such definitions may be safely used in practice, it is
essential to understand what security guarantees they provide, i.e, what
resources they assume and what resources they construct. The
accessible information security definition for QKD that was discussed
in \secref{sec:qkd.other.ai} (see also
\secref{sec:alternative.memoryless}) turned out to implicitly assume
that the adversary has no quantum memory. In the case of these
game\-/based definitions, a series of results show that some of them
have the opposite flaw: they construct a resource that is
unnecessarily strong and exclude certain protocols that should be
considered secure \cite{CKN03,CMT13,BMPZ19}.

The strongest of these definitions, called
\emph{quantum authenticated encryption} (QAE) in \cite{AGM18}, is the
most similar to the construction of secure channels used in
\secref{sec:computational.qmt}. \textcite{BMPZ19} show that QAE
essentially corresponds to constructing a secure channel, but with a
fixed simulator, whereas a security definition within the Abstract Cryptography
framework only requires the
existence of a simulator. A protocol for which the simulator
hard\-/coded in QAE is a good simulator will be deemed
secure. However, for a protocol that requires a different simulator to
prove its security, the QAE definition will just declare it insecure,
even though it constructs a secure channel.


\section{Open problems}
\label{sec:open}

Only a relatively small part of the protocols appearing in the quantum
cryptography literature have been analysed and proved secure within a
composable framework. To understand the security guarantees they
actually provide, and in what contexts they can be used safely, such
an analysis would however be crucial, thus representing a major task
for quantum cryptographers to be completed in the future. Here we
illustrate this task, focusing on a few areas that we consider
interesting. The first is the problem of reusing devices in device
independent cryptography (\secref{sec:open.di}). The second is
modeling quantum cryptography with non\-/asymptotic computational
assumptions (\secref{sec:open.computational}). And the third consists
in studying setup assumptions that are needed to achieve a broader
range of constructions (\secref{sec:open.other}).

\subsection{Reusing devices in device-independent cryptography}
\label{sec:open.di}

In \secref{sec:alternative.di} we modeled device\-/independent
(DI) QKD. There, the (untrusted) devices correspond to resources that
are available to the honest players. If Alice and Bob want to run
another DI-QKD protocol to generate more key once the first run is
over, they will again need all the same resources, i.e., they will
need such devices once more. Obviously, if Alice and Bob have acces to
new, fresh devices, they can run the protocol a second time with
these. However, it does not follow from that analysis that the
\emph{same} devices can be used again. In fact it has been shown by
\textcite{BCK13} that in general these devices cannot be reused a
second time: the internal memory of a device used for key (or randomness) generation may contain
information about the secret key (or random number) generated in the
first round, and the device may thus leak this information when being
reused in a second round. A secret bit may be leaked in a subtle manner.  For example, if the bit equals $0$ the device may perform the expected operations during the second round, and if the bit equals $1$ force an abort. 

Reusing devices in DI cryptography is very similar to reusing keys. In
general it cannot be done. However, in the case of keys, if one can
prove that the key is (close to) uniform and independent from the
adversary's information, then it can be recycled \--- this was covered
in \secref{sec:recycle}. The same approach could be used to recycle
devices: instead of the ideal world just consisting of a key resource,
it should as well provide access to devices that are independent of
this key as depicted in \figref{fig:open.di}.

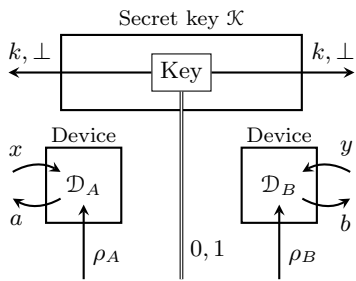
\begin{figure}[tb]

\begin{tikzpicture}[
sArrow/.style={->,>=stealth,thick},
thinResource/.style={draw,thick,minimum width=3.2cm,minimum height=1cm},
sqResource/.style={draw,thick,minimum width=1cm,minimum height=1cm},
pnode/.style={minimum width=.6cm,minimum height=.5cm}]

\small

\def\t{2.3} 
\def\u{.25}
\def\v{.75}
\def\w{1.3} 
\def\z{2}

\node[thinResource] (keyBox) at (0,\v) {};
\node[yshift=-1.5,above] at (keyBox.north) {\footnotesize
  Secret key $\aK$};

\node[draw] (key) at (0,\v) {Key};
\node (a1) at (-\t,\v) {};
\node (b1) at (\t,\v) {};
\node (eve) at (0,-\z) {};

\draw[sArrow] (key) to node[pos=.85,auto] {$k,\bot$} (b1.center);
\draw[sArrow] (key) to node[pos=.85,auto,swap] {$k,\bot$} (a1.center);
\draw[double] (eve.center) to node[pos=.15,auto,swap] {$0,1$} (key);

\node[pnode] (a2) at (-\t-\u,-\v) {};
\node[pnode] (b2) at (\t+\u,-\v) {};

\node[sqResource] (da) at (-\w,-\v) {$\aD_A$};
\node[yshift=-1.5,above] at (da.north) {\footnotesize
  Device};
\node[pnode] (dan) at (-\w,-\v) {};
\node[sqResource] (db) at (\w,-\v) {$\aD_B$};
\node[yshift=-1.5,above] at (db.north) {\footnotesize
  Device};
\node[pnode] (dbn) at (\w,-\v) {};

\node (eveq1) at (-\w,-\z) {};
\node (eveq2) at (\w,-\z) {};

\draw[sArrow] (eveq1.center) to node[pos=.3,auto,swap] {$\rho_A$} (dan);
\draw[sArrow] (eveq2.center) to node[pos=.3,auto,swap] {$\rho_B$} (dbn);

\draw[sArrow,bend left] (a2) to node[pos=.4,auto] {$x$} (dan);
\draw[sArrow,bend left] (dan) to node[pos=.6,auto] {$a$} (a2);
\draw[sArrow,bend right] (b2) to node[pos=.4,auto,swap] {$y$} (dbn);
\draw[sArrow,bend right] (dbn) to node[pos=.6,auto,swap] {$b$} (b2);

\end{tikzpicture}

\caption[Reusing devices]{\label{fig:open.di}An ideal world in which a
  secret key is produced by $\aK$ and new devices $\aD_A$ and $\aD_B$
  independent from $\aK$ are accessible to the players..}
\end{figure}

Unfortunately, no DI-QKD protocol has ever been shown to construct the
ideal system from \figref{fig:open.di} and it might well be impossible
to do so. But even if this is the case, it does not exclude that one
can construct an ideal system that is stronger than just the shared
secret key considered in \secref{sec:alternative.di}, e.g., one in
which the devices have some partial independence from the key or are
fully independent in certain contexts.\footnote{Context restricted
  composability is a promising research path for protocols that do not
  construct the desired ideal resource. Its investigation has been
  initiated in \textcite{JM18}, and is beyond the scope of this
  review.}

We note that weaker models such as measurement\-/device\-/independent
(MDI) QKD \--- see \secref{sec:alternative.semi} \--- do not
suffer from the same problem of device reuse as DI-QKD. The reason is
that in MDI-QKD one does not need to make any assumption about the
measurement devices at all (the adversary does the measurements for
the honest players), whereas in DI-QKD one has to assume that no
unauthorized information leaves the devices.



\subsection{Computational security}
\label{sec:open.computational}

Computational security is a fairly unexplored area of quantum
cryptography. The main motivation for studying this is to achieve
results that are not possible with information\-/theoretic
security. For example, in \secref{sec:dqc} we mentioned a
computationally secure protocol for delegated quantum computation with
a classical client \cite{GV19}, which is believed not to be possible
with information\-/theoretic security~\cite{ACGK19}. The
computationally secure message transmission from
\secref{sec:computational} allows keys to be reused without the extra
communication required by QKD (\secref{sec:qkd}) or key recycling
(\secref{sec:recycle}) \--- and thus, without the possibility of an
adversary interrupting this communication and preventing the key from
being reused. And the work from \textcite{Unr13} discussed in
\secref{sec:mpc.ever} removes the need for a shared secret key in QKD
by using signature cards instead.

One may essentially analyze any area of cryptography with
computational security to study how assumptions needed for
information\-/theoretic security may be weakened in the computational
setting. There is however no single way to model computational
assumptions, and important open questions in the field are to identify
the (best) ways of doing this. Most frequently, one proves a
reduction, i.e., if some distinguisher can guess whether it is
interacting with the real or ideal system, then this distinguisher can
be used to solve some problem which is believed to be hard. In
\secref{sec:computational} we reviewed the finite reductions from
\textcite{BMPZ19}, in which the probability of a distinguisher $\fD$
distinguishing the real and ideal worlds is bounded by the probability
of this distinguisher being successfully used \--- as part of a new
distinguisher $\fD'$ \--- to distinguish a pseudo\-/random function
from a uniform random function; see also \textcite{Rog06} for a
discussion of reductions.

Another way to define computational security would be to define an
ideal resource that falls under the control of the adversary if she
can solve some problem believed to be hard (e.g., find a collision for
a hash function). This is essentially the ``identical-until-bad''
concept of \textcite{BR06}, but adapted to composable security instead of
game-based security. To the best of our knowledge, this paradigm remains
completely unexplored in quantum cryptography.

Other works such as \textcite{CCLVW17} bound adversaries by
circuit sizes. It is not clear how to model that in a finite,
composable framework, and is important open work.

\subsection{Other setup assumptions}
\label{sec:open.other}

When a security definition is considered ``not composable'', it often
has some (setup) assumption hard-coded in it, which is not present in
the obvious composable definition, and is therefore strictly
weaker. By modeling this assumption in a composable framework, one can
get another, equivalent composable definition. We illustrated this
in \secref{sec:alternative.memoryless} by explaining how a definition for QKD based on the accessible information, which is normally not composable, can be turned into a composable
one within a model where an adversary has no (long term) quantum
memory.

Similar techniques have been used by \textcite{Unr11} to obtain
commitments in the bounded storage model. While it follows from
\textcite{VPdR19} that coin flipping and bit commitment are impossible
in a bounded storage model without further assumptions,
\textcite{Unr11} avoids these by putting a bound on the number of
times a protocol can be run in parallel, and designing protocols that
are secure for this limited number of compositions.\footnote{This
  effectively restricts what the distinguisher/environment may do to
  distinguish the real and ideal systems, since the bound on the
  number of executions of a protocol applies to the distinguisher as
  well.} Likewise, \textcite{Pro20} has made extra setup assumptions
in the relativistic model to avoid the impossibility results of
\textcite{VPdR19}.
  
  There are numerous works where security is proved based on the assumption that adversaries are restricted. For example, adversaries cannot share entanglement in
\textcite{BCFGGOS14}, the adversaries' memory size is bounded in
\textcite{DFSS07,DFSS08}, the adversaries' memory is noisy in
\textcite{WST08,STW09,KWW12}, or adversaries can only perform local
operations on single qubits and communicate classically in
\textcite{Liu14,Liu15}. It remains open how to model these assumptions
to get composable security statements and prove in what setting such
protocols are secure. Similarly, to capture position\-/based
cryptography \textcite{Unr14} uses a model of circuits with positions in
space\-/time. Here too, it is not clear how to fit these results in a
composable framework and identify the resource that is constructed by these
protocols.


\begin{acknowledgments}
  We would like to thank Claude Cr\'epeau, Marco Lucamarini, Mark
  Wilde, Ramona Wolf and the anonymous referees for their helpful
  comments and suggestions.  This work has been funded by the European
  Research Council (ERC) via grant No.~258932, the Swiss National
  Science Foundation via the National Centre of Competence in Research
  ``QSIT'', the QuantERA project ``eDICT'', the US Air Force Office of
  Scientific Research (AFOSR) via grants~FA9550-16-1-0245 and
  FA9550-19-1-0202, the Quantum Center of ETH Zurich, and the Zurich
  Information Security and Privacy Center.
\end{acknowledgments}


\appendix

\section{Trace distance}
\label{app:op}

Many of the statements in this work use the well\-/known fact that the
distinguishing advantage between two systems that output states $\rho$
and $\sigma$ is equivalent to the trace distance between these
states. In this appendix we gather lemmas and theorems which prove
this fact and help interpret the meaning of the trace distance.

In \appendixref{app:op.definitions} we first define the trace distance
\--- as well as its classical counterpart, the total variation
distance\--- and provide some basic lemmas that can also be found in
textbooks such as \textcite{nielsen2010quantum}. In
\appendixref{app:op.distadv} we then show the connection between trace
distance and distinguishing advantage, which was originally proven by
\textcite{Hel76}. In \appendixref{app:op.failure} we prove that
we can alternatively think of the trace distance between a real and
an ideal system as a bound on the probability that a failure occurs in
the real system as suggested in \textcite{Ren05}. Finally, in
\appendixref{app:op.local} we bound two typical information theory
notions of secrecy \--- the conditional entropy of a key given the
eavesdropper's information and her probability of correctly guessing
the key \--- in terms of the trace distance. Although such measures of
information are generally ill-suited for defining cryptographic
security, they can help interpret the notion of a key being
$\eps$\-/close to uniform.

\subsection{Metric definitions}
\label{app:op.definitions}

In the case of a classical system, statistical security is defined by
the total variation (or statistical) distance between the probability
distributions describing the real and ideal settings, which is defined
as follows.\footnote{We employ the same notation $D(\cdot,\cdot)$ for
  both the total variation and trace distance, since the former is a
  special case of the latter.}

\begin{deff}[Total variation distance]
  \label{def:vdist}
  The total variation distance between two probability distributions
  $P_Z$ and $P_{\tilde{Z}}$ over an alphabet $\cZ$ is defined as
  \[D(P_Z,P_{\tilde{Z}}) \coloneqq \frac{1}{2} \sum_{z \in \cZ} \left| P_Z(z) -
    P_{\tilde{Z}}(z) \right|.\]
\end{deff}

Using the fact that $|a-b| = a+b-2 \min(a,b)$, the total variation
distance can also be written as
\begin{equation} \label{eq:vdist.alt} D(P_Z, P_{\tilde{Z}}) = 1-
  \sum_{z \in \cZ} \min[P_Z(z), P_{\tilde{Z}}(z)].\end{equation}

In the case of quantum states instead of classical random variables,
the total variation distance generalizes to the trace distance. More
precisely, the trace distance between two density operators that are
diagonal in the same orthonormal basis is equal to the total variation
distance between the probability distributions defined by their
respective eigenvalues.

\begin{deff}[trace distance]
  \label{def:tdist}
  The trace distance between two quantum states $\rho$ and $\sigma$ is
  defined as
  \[D(\rho,\sigma) \coloneqq \frac{1}{2} \tr |\rho-\sigma|.\]
\end{deff}

We now introduce some technical lemmas involving the trace distance,
which help us derive the theorems in the next sections. Proofs of
these lemmas may be found in \textcite{nielsen2010quantum}.

\begin{lem}
  \label{lem:op.definitions.pos}
  For any two states $\rho$ and $\sigma$ and any operator $0 \leq
  M \leq I$,
  \begin{equation} \label{eq:op.definitions.pos}
    D(\rho,\sigma) \geq \trace{M(\rho-\sigma)}.
  \end{equation}
  Furthermore, this inequality is tight for some value of $M$.
\end{lem}

The trace distance can thus alternatively be written
as \begin{equation} \label{eq:op.definitions.pos.d}
D(\rho,\sigma) = \max_M \trace{M(\rho-\sigma)}.\end{equation}

Let $\{\Gamma_x\}_{x}$ be a positive operator\-/valued measure (POVM)
\--- a set of operators $0 \leq \Gamma_x \leq I$ such that
$\sum_x \Gamma_x = I$ \--- and let $P_X$ denote the outcome of
measuring a quantum state $\rho$ with $\{\Gamma_x\}_{x}$, i.e.,
$P_X(x) = \trace{\Gamma_x \rho}$. Our next lemma says that the trace
distance between two states $\rho$ and $\sigma$ is equal to the total
variation between the outcomes \--- $P_X$ and $Q_X$ \--- of an optimal
measurement on the two states.

\begin{lem}
  \label{lem:op.definitions.opt}
  For any two states $\rho$ and $\sigma$,
  \begin{equation} \label{eq:op.definitions.opt} D(\rho,\sigma) =
    \max_{\left\{\Gamma_x\right\}_x} D(P_X,Q_X),\end{equation}
  where $P_X$ and $Q_X$ are the probability distributions resulting
  from measuring $\rho$ and $\sigma$ with a POVM $\{\Gamma_x\}_x$,
  respectively, and the maximization is over all POVMs. Furthermore,
  if the two states $\rho_{ZB}$ and $\sigma_{ZB}$ have a classical
  subsystem $Z$, then the measurement satisfying
  \eqnref{eq:op.definitions.opt} leaves the classical subsystem
  unchanged, i.e., the maximum is reached for a POVM with
  elements \begin{equation} \label{eq:op.definitions.opt.classical}
    \Gamma_x = \sum_z \proj{z} \otimes M^z_x ,\end{equation} where
  $\{\ket{z}\}_z$ is the classical orthonormal basis of $Z$.
\end{lem}


\subsection{Distinguishing advantage}
\label{app:op.distadv}

\textcite{Hel76} proved that the advantage a distinguisher has in
guessing whether it was provided with one of two states with equal
priors, $\rho$ or $\sigma$, is given by the trace distance between the
two, $D(\rho,\sigma)$.\footnote{Actually, \textcite{Hel76} solved
  a more general problem, in which the states $\rho$ and $\sigma$ are
  picked with apriori probabilities $p$ and $1-p$, respectively,
  instead of $1/2$ as in the definition of the distinguishing
  advantage.} We first sketch the classical case, then prove the
quantum version.

Let a distinguisher be given a value sampled according to probability
distributions $P_Z$ or $P_{\tilde{Z}}$, where $P_Z$ and
$P_{\tilde{Z}}$ are each chosen with probability $1/2$. Suppose the
value received by the distinguisher is $z \in \cZ$. If $P_Z(z) >
P_{\tilde{Z}}(z)$, its best guess is that the value was sampled
according to $P_Z$. Otherwise, it should guess that it was
$P_{\tilde{Z}}$. Let $\cZ'\coloneqq \{z \in \cZ : P_Z(z) >
P_{\tilde{Z}}(z)\}$ and $\cZ'' \coloneqq \{z \in \cZ : P_Z(z) \leq
P_{\tilde{Z}}(z)\}$. There are a total of $2|\cZ|$ possible events:
the sample is chosen according to $P_Z$ or $P_{\tilde{Z}}$ and takes
the value $z \in \cZ$. These events have probabilities
$\frac{P_Z(z)}{2}$ and $\frac{P_{\tilde{Z}}(z)}{2}$. Conditioned on
$P_Z$ being chosen and $z$ being the sampled value, the distinguisher
has probability $1$ of guessing correctly with the strategy outlined
above if $z \in \cZ'$, and $0$ otherwise. Likewise, if $P_{\tilde{Z}}$
was selected, it has probability $1$ of guessing correctly if $z \in
\cZ''$ and $0$ otherwise. The probability of correctly guessing
whether it was given a value sampled according to $P_Z$ or
$P_{\tilde{Z}}$, which we denote $\distinguish{P_Z,P_{\tilde{Z}}}$, is
obtained by summing over all possible events weighted by their
probabilities. Hence
\begin{align*}
 &  \distinguish{P_Z,P_{\tilde{Z}}} \\
& \qquad = \sum_{z \in \cZ'} \frac{P_Z(z)}{2} + \sum_{z \in
    \cZ''} \frac{P_{\tilde{Z}}(z)}{2} \\
  & \qquad =  \frac{1}{2} \left(1- \sum_{z \in
    \cZ''} P_{Z}(z) \right) + \frac{1}{2} \left(1 - \sum_{z \in
    \cZ'} P_{\tilde{Z}}(z)\right) \\
  & \qquad = 1 - \frac{1}{2} \sum_{z \in \cZ} \min[P_Z(z), P_{\tilde{Z}}(z)] \\
  & \qquad = \frac{1}{2} + \frac{1}{2} D(P_Z,P_{\tilde{Z}}) , 
\end{align*}
where in the last equality we used the alternative formulation of the
total variation distance from \eqnref{eq:vdist.alt}.

We now generalize the argument above to quantum states with equal
priors, which is a special case of \textcite{Hel76}.

\begin{thm}
  \label{thm:op.distinguishing}
  For any states $\rho$ and $\sigma$, we have \[\distinguish{\rho,\sigma} =
  \frac{1}{2}+\frac{1}{2}D(\rho,\sigma) .\]
\end{thm}

\begin{proof}
  If a distinguisher is given one of two states $\rho$ or $\sigma$,
  each with probability $1/2$, its probability of guessing which one
  it holds is given by a maximization of all possible measurements it
  may do: it chooses some POVM $\{\Gamma_0,\Gamma_1\}$, where
  $\Gamma_0$ and $\Gamma_1$ are positive operators with
  $\Gamma_0+\Gamma_1 = I$, and measures the state it holds. If it gets
  the outcome $0$, it guesses that it holds $\rho$ and if it gets the
  outcome $1$, it guesses that it holds $\sigma$. The probability of
  guessing correctly is given by
\begin{align}
  \distinguish{\rho,\sigma} & = \max_{\Gamma_0,\Gamma_1} \left[
    \frac{1}{2}\trace{\Gamma_0\rho} +
    \frac{1}{2}\trace{\Gamma_1\sigma} \right] \notag \\
  & = \frac{1}{2}\max_{\Gamma_0} \left[ \trace{\Gamma_0\rho} +
    \trace{(I-\Gamma_0)\sigma} \right] \notag \\
  & = \frac{1}{2} + \frac{1}{2} \max_{\Gamma_0}
  \trace{\Gamma_0(\rho-\sigma)} . \label{eq:op.distinguishing.thm}
\end{align}
The proof concludes by plugging \eqnref{eq:op.definitions.pos.d} in
\eqnref{eq:op.distinguishing.thm}.
\end{proof}

\subsection{Probability of a failure}
\label{app:op.failure}

The trace distance is used as the security definition of QKD, because
the relevant measure for cryptographic security is the distinguishing
advantage (as discussed in \secref{sec:ac}), and as proven in
\thmref{thm:op.distinguishing}, the distinguishing advantage between
two quantum states corresponds to their trace distance. This
operational interpretation of the trace distance involves two worlds,
an ideal one and a real one, and the distance measure is the
(renormalized) difference between the probabilities of the
distinguisher correctly guessing to which world it is connected.

In this section we describe a different interpretation of the total
variation and trace distances. Instead of having two different worlds,
we consider one world in which the outcomes of interacting with the
real and ideal systems co-exist. And instead of these distance
measures being a difference between probability distributions, they
become the probability that any (classical) value occurring in one of
the systems does not \emph{simultaneously} occur in the other. We call
such an event a \emph{failure} \--- since one system is ideal, if the
other behaves differently, it must have failed \--- and the trace
distance becomes the probability of a failure occurring.

Given two random variables $Z$ and $\tilde{Z}$ with probability
distributions $P_Z$ and $P_{\tilde{Z}}$, any distribution
$P_{Z\tilde{Z}}$ with marginals given by $P_Z$ and $P_{\tilde{Z}}$ is
called a coupling of $P_Z$ and $P_{\tilde{Z}}$. The interpretation of
the trace distance treated in this section uses one specific coupling,
known as a \emph{maximal coupling} in probability theory~\cite{Tho00}.

\begin{thm}[Maximal coupling] \label{thm:difference}
  Let $P_Z$ and $P_{\tilde{Z}}$ be two probability distributions over
  the same alphabet $\cZ$. Then there exists a probability
  distribution $P_{Z \tilde{Z}}$ on $\cZ \times \cZ$ such that
  \begin{equation} \label{eq:equalityprob}
    \Pr[Z  = \tilde{Z}] := \sum_{z} P_{Z \tilde{Z}}(z, z) \geq 1-  D(P_Z, P_{\tilde{Z}}) 
  \end{equation}
  and such that $P_Z$ and $P_{\tilde{Z}}$ are the marginals of $P_{Z
    \tilde{Z}}$, i.e., 
  \begin{align} \label{eq:marginal1}
    P_Z(z) & = \sum_{\tilde{z}} P_{Z \tilde{Z}}(z,\tilde{z})  \quad
    (\forall z \in \cZ)
    \\ \label{eq:marginal2}
    P_{\tilde{Z}}(\tilde{z}) & = \sum_z P_{Z \tilde{Z}}(z,\tilde{z})
    \quad (\forall \tilde{z} \in \cZ) .
  \end{align}
\end{thm}

It turns out that the inequality in \eqnref{eq:equalityprob} is tight,
i.e., one can also show that for any distribution $P_{Z\tilde{Z}}$,
$\Pr[Z = \tilde{Z}] \leq 1 - D(P_Z, P_{\tilde{Z}})$. We will however
not use this fact here.

Consider now a real system that outputs values given by $Z$ and an
ideal system that outputs values according to
$\tilde{Z}$. \thmref{thm:difference} tells us that there exists a
coupling of these distributions such that the probability of the real
system producing a different value from the ideal system is bounded by
the total variation distance between $P_Z$ and $P_{\tilde{Z}}$. Thus,
the real system behaves ideally except with probability $D(P_Z,
P_{\tilde{Z}})$.

We first prove this theorem, then in \corref{cor:difference} here
below we apply it to quantum systems.

\begin{proof}[Proof of \thmref{thm:difference}]
  Let $Q_{Z \tilde{Z}}$ be the real function on $\cZ \times \cZ$
  defined by
  \begin{align*}
    Q_{Z \tilde{Z}}(z,\tilde{z}) = \begin{cases} \min[P_Z(z),
      P_{\tilde{Z}}(\tilde{z})] & \text{if $z=\tilde{z}$} \\ 0 & \text{otherwise}\end{cases} 
  \end{align*}
  (for all $z, \tilde{z} \in \cZ$).  Furthermore, let $R_Z$ and
  $R_{\tilde{Z}}$ be the real functions on $\cZ$ defined by
  \begin{align*}
    R_Z(z) & = P_Z(z) - Q_{Z \tilde{Z}}(z, z) , \\
    R_{\tilde{Z}}(\tilde{z}) & = P_{\tilde{Z}}(\tilde{z}) - Q_{Z
        \tilde{Z}}(\tilde{z}, \tilde{z})  .
  \end{align*}
  We then define $P_{Z \tilde{Z}}$ by
  \begin{align*}
    P_{Z \tilde{Z}}(z, \tilde{z}) = Q_{Z \tilde{Z}}(z, \tilde{z}) +
    \frac{1}{D(P_Z, P_{\tilde{Z}})} R_Z(z) R_{\tilde{Z}}(\tilde{z}) .
  \end{align*}
   
  We now show that $P_{Z \tilde{Z}}$ satisfies the conditions of the
  theorem. For this, we note that for any $z \in \cZ$
  \begin{align*}
    R_Z(z) = P_Z(z) - \min[P_Z(z), P_{\tilde{Z}}(z)] \geq 0, 
  \end{align*}
  i.e., $R_Z$, and, likewise, $R_{\tilde{Z}}$, are nonnegative. Since
  $Q_{Z \tilde{Z}}$ is by definition also nonnegative, we have that
  $P_{Z \tilde{Z}}$ is nonnegative, too. From \eqnref{eq:marginal1}
  or~\eqref{eq:marginal2}, which we will prove below, it follows that
  $P_{Z \tilde{Z}}$ is also normalized. Hence, $P_{Z \tilde{Z}}$ is a
  valid probability distribution.

  To show \eqnref{eq:equalityprob} we use again the non\-/negativity
  of $R_{Z}$ and $R_{\tilde{Z}}$, which implies
  \begin{align*}
    \sum_z P_{Z \tilde{Z}}(z,z) & \geq \sum_z Q_{Z \tilde{Z}}(z,z) \\
    & = \sum_z \min[P_Z(z), P_{\tilde{Z}}(z)] \\
    & = 1- D(P_Z, P_{\tilde{Z}}) ,
  \end{align*}
  where in the last equality we used the alternative formulation of the
  total variation distance from \eqnref{eq:vdist.alt}.

  To prove \eqnref{eq:marginal1}, we first note that
  \begin{align*}
    \sum_{\tilde{z}} R_{\tilde{Z}}(\tilde{z}) & = \sum_{\tilde{z}}
    P_{\tilde{Z}}(\tilde{z}) - \sum_{\tilde{z}}  Q_{Z
      \tilde{Z}}(\tilde{z},\tilde{z})  \\ & =
    1 - \sum_{\tilde{z}} \min[P_Z(\tilde{z}),
    P_{\tilde{Z}}(\tilde{z})] \\ & = D(P_Z,  P_{\tilde{Z}}) .
  \end{align*}
  Using this we find that for any $z \in \cZ$
  \begin{align*}
    & \sum_{\tilde{z}} P_{Z \tilde{Z}}(z,\tilde{z}) \\
  & \qquad = \sum_{\tilde{z}} Q_{Z \tilde{Z}}(z,\tilde{z}) + R_Z(z)
    \frac{1}{D(P_Z, P_{\tilde{Z}})} \sum_{\tilde{z}}  R_{\tilde{Z}}(\tilde{z})  \\
  & \qquad = Q_{Z \tilde{Z}}(z,z) +R_Z(z) = P_Z(z) .
  \end{align*}  
  By symmetry, this also proves \eqnref{eq:marginal2}.
\end{proof}

In the case of quantum states, \thmref{thm:difference} can be used to
couple the outcomes of any observable applied to the quantum systems.

\begin{cor}
\label{cor:difference}
For any states $\rho$ and $\sigma$ with trace distance $D(\rho,\sigma)
\leq \eps$, and any measurement given by its POVM operators
$\{\Gamma_w\}_w$ with outcome probabilities $P_W(w) = \trace{\Gamma_w
  \rho}$ and $P_{\tilde{W}}(w) = \trace{\Gamma_w \sigma}$, there
exists a coupling of $P_W$ and $P_{\tilde{W}}$ such that \[ \Pr[W \neq
\tilde{W}] \leq D(\rho,\sigma) .\]
\end{cor}

\begin{proof}
  Immediate by combining \lemref{lem:op.definitions.opt} and
  \thmref{thm:difference}.
\end{proof}

\corref{cor:difference} tells us that if two systems produce states
$\rho$ and $\sigma$, then for any observations made on those systems
there exists a coupling for which the values of each measurement will
differ with probability at most $D(\rho,\sigma)$. It is instructive to
remember that this operational meaning is not essential to the
security notion or part of the framework in any way. It is an
intuitive way of understanding the trace distance, so as to better
choose a suitable value. It allows this distance to be thought of as a
maximum failure probability, and the value for $\eps$ to be chosen
accordingly.

\subsection{Measures of uncertainty}
\label{app:op.local}

Non\-/composable security models often use measures of uncertainty to
quantify how much information an adversary might have about a secret,
e.g., entropy as used by Shannon to prove the security of the one-time
pad~\cite{Shannon49}. These measures are often weaker than what one
obtains using a global distinguisher, and in general do not provide
good security definitions. They are however quite intuitive and in
order to further illustrate the quantitative value of the
distinguishing advantage, we derive bounds on two of these measures of
uncertainty in terms of the trace distance, namely on the probability
of guessing the secret key in
\appendixref{app:op.local.guessing} and on the von Neumann entropy of the
secret key in \appendixref{app:op.local.entropy}.

\subsubsection{Probability of guessing}
\label{app:op.local.guessing}

Let $\rho_{KE} = \sum_{k \in \cK} p_k \proj{k}_K \otimes \rho^k_E$ be
the joint state of a secret key in the $K$ subsystem and Eve's
information in the $E$ subsystem. To guess the value of the key, Eve
can pick a POVM $\{\Gamma_k\}_{k \in \cK}$, measure her system, and
output the result of the measurement. Given that the key is $k$, her
probability of having guessed correctly is $\trace{\Gamma_k
  \rho^k_E}$. The average probability of guessing correctly for this
measurement is then given by the sum over all $k$, weighted by their
respective probabilities $p_k$. And Eve's probability of correctly
guessing the key is defined by taking the maximum over all
measurements,
\begin{equation} \label{eq:pguess} \pguess[\rho]{K|E} \coloneqq
  \max_{\{\Gamma_k\}} \sum_{k \in \cK} p_k \trace{\Gamma_k \rho^k_E}.
\end{equation}

\begin{lem}
  \label{lem:pguess}
For any bipartite state $\rho_{KE}$ with classical $K$, 
\[\pguess[\rho]{K|E} \leq
\frac{1}{|\cK|} + D\left( \rho_{KE},\tau_K \otimes \rho_E\right),
\] where $\tau_K$ is the fully mixed state.
\end{lem}

\begin{proof}
  Note that for $M \coloneqq \sum_k \proj{k} \otimes \Gamma_k$, where
  $\{\Gamma_k\}$ maximizes \eqnref{eq:pguess}, the guessing
  probability can equivalently be written \[\pguess[\rho]{K|E} =
  \trace{M \rho_{KE}}. \] Furthermore,
\[\ktrace{M \left(\tau_K \otimes \rho_{E}\right)} = \frac{1}{|\cK|}. \]
In \lemref{lem:op.definitions.pos} we proved that for any
operator $0 \leq M \leq I$,
\[\trace{M(\rho - \sigma)} \leq D(\rho,\sigma). \]
Setting $\rho = \rho_{KE}$ and $\sigma = \tau_K \otimes \rho_{E}$ in
the above inequality, we finish the proof:
\begin{align*}
& \trace{M\rho_{KE}} \leq \trace{M\left(\tau_K \otimes \rho_{E}\right)}
+ D\left( \rho_{KE} , \tau_K \otimes \rho_{E} \right) , \\
 & \implies \quad \pguess[\rho]{K|E} \leq \frac{1}{|\cK|} + D\left( \rho_{KE},\tau_K
   \otimes \rho_E\right) . \qedhere
\end{align*}
\end{proof}

\subsubsection{Entropy}
\label{app:op.local.entropy}

Let $\rho_{KE} = \sum_{k \in \cK} p_k \proj{k}_K \otimes \rho^k_E$ be
the joint state of a secret key in the $K$ subsystem and Eve's
information in the $E$ subsystem. We wish to bound the von Neumann
entropy of $K$ given $E$ \--- $\Ss(K|E)_{\rho} = \Ss(\rho_{KE}) -
\Ss(\rho_E)$, where $S(\rho) \coloneqq -\trace{\rho \log \rho}$ \---
in terms of the trace distance $D\left( \rho_{KE},\tau_K \otimes
  \rho_E\right)$. We first derive a lower bound on the von Neumann
entropy, using the following theorem from \textcite{AF04}.

\begin{thm}[From~\textcite{AF04}]
  \label{thm:AF04}
  For any bipartite states $\rho_{AB}$ and $\sigma_{AB}$ with trace
  distance $D(\rho,\sigma) = \eps \leq 1/4$ and $\dim \hilbert_A =
  d_A$, we have \[ \left| \Ss(A|B)_\rho - \Ss(A|B)_\sigma \right| \leq
  8 \eps \log d_A + 2h(2\eps) , \] where $h(p) = - p \log p - (1-p)
  \log (1-p)$ is the binary entropy.
\end{thm}

\begin{cor}
  \label{cor:AF04}
 For any state $\rho_{KE}$ with $D(\rho_{KE},\tau_K \otimes
  \rho_E) = \eps \leq 1/4$, where $\tau_K$ is the fully mixed state,
  we have
  \[ \Ss(K|E)_{\rho} \geq (1- 8\eps) \log |\cK| - 2h(2\eps) . \]
\end{cor}

\begin{proof}
  Immediate by plugging $\rho_{KE}$ and $\tau_K \otimes \rho_E$ in
  \thmref{thm:AF04}.
\end{proof}

Given the von Neumann entropy of $K$ conditioned on $E$,
$\Ss(K|E)_{\rho}$, one can also upper bound the trace distance of
$\rho_{KE}$ from $\tau_K \otimes \rho_E$ by relating $\Ss(K|E)_{\rho}$
to the relative entropy of $\rho_{KE}$ to $\tau_K \otimes \rho_E$ \---
the relative entropy of $\rho$ to $\sigma$ is defined as $\Ss(\rho \|
\sigma) \coloneqq \trace{\rho \log \rho} - \trace{\rho \log \sigma}$.

\begin{lem}
  \label{lem:entropybound}
  For any quantum state $\rho_{KE}$,
  \[ D(\rho_{KE},\tau_K \otimes \rho_E) \leq
  \sqrt{\frac{1}{2}\left(\log |\cK| - \Ss(K|E)_{\rho}\right)} . \]
\end{lem}

\begin{proof}
  From the definitions of the relative and von Neumann entropies we
  have
  \begin{align*}
    \Ss \left(\rho_{KE} \middle\| \tau_K \otimes \rho_E \right) & =
  \log |\cK| + \Ss \left(\rho_{KE} \middle\| \id_K \otimes \rho_E
  \right) \\ & = \log |\cK| - \Ss(K|E)_{\rho},\end{align*} where $\id_K$ is the
  identity matrix. We then use the following bound on the relative
  entropy \cite[Theorem~1.15]{OP93} to conclude the proof:
  \[ \Ss\left(\rho \middle\| \sigma \right) \geq 2
  \left(D(\rho,\sigma)\right)^2. \qedhere\]
\end{proof}

\corref{cor:AF04} and \lemref{lem:entropybound} can be written
together in one equation, upper and lower bounding the conditional von
Neumann entropy:
\[ (1- 8\eps) \log |\cK| - 2h(2\eps) \leq \Ss(K|E)_{\rho} \leq \log
|\cK| - 2\eps^2, \] where $\eps = D(\rho_{KE},\tau_K \otimes
\rho_E)$.


\section{Proofs from \secref{sec:qkd}}
\label{app:proofs}

In \secref{sec:qkd} we show how to define the security of QKD in a
composable framework and relate this to the trace distance security
criterion introduced in \textcite{Ren05}. This composable treatment of
the security of QKD follows the literature \cite{BHLMO05,MR09}, and
the results presented in \secref{sec:qkd} may be found in
\textcite{BHLMO05,MR09} as well. The formulation of the statements
differs however from those works, since we use here the Abstract
Cryptography framework of \textcite{MR11}. So for completeness, we
provide here proofs of the main results from \secref{sec:qkd}.

\begin{proof}[Proof of \thmref{thm:qkd}]
  Recall that in \secref{sec:security.simulator} we fixed the
  simulator and show that to satisfy \eqnref{eq:qkd.security} it is
  sufficient for \eqnref{eq:qkd.security.2} to hold. Here, we will
  break \eqnref{eq:qkd.security.2} into security [\eqnref{eq:qkd.cor}]
  and correctness [\eqnref{eq:qkd.sec}], thus proving the theorem.

  Let us define $\gamma_{ABE}$ to be a state obtained from
  $\rho^{\top}_{ABE}$ [\eqnref{eq:qkd.security.tmp}] by throwing away
  the $B$ system and replacing it with a copy of $A$, i.e., \[
  \gamma_{ABE} = \frac{1}{1-p^\bot} \sum_{k_A,k_B \in \cK} p_{k_A,k_B}
  \proj{k_A,k_A} \otimes \rho^{k_A,k_B}_E.\] From the triangle
  inequality we get \begin{multline*} D(\rho^\top_{ABE},\tau_{AB} \otimes
  \rho^\top_{E}) \leq \\ D(\rho^\top_{ABE},\gamma_{ABE}) +
  D(\gamma_{ABE},\tau_{AB} \otimes \rho^\top_{E}) .\end{multline*}

Since in the states $\gamma_{ABE}$ and
$\tau_{AB} \otimes \rho^\top_{E}$ the $B$ system is a copy of the $A$
system, it does not modify the distance. Furthermore,
$\trace[B]{\gamma_{ABE}} =
\trace[B]{\rho^{\top}_{ABE}}$. Hence
\[D(\gamma_{ABE},\tau_{AB} \otimes \rho^\top_{E}) =
  D(\gamma_{AE},\tau_{A} \otimes \rho^\top_{E}) =
  D(\rho^\top_{AE},\tau_{A} \otimes \rho^\top_{E}).\]

For the other term note that
\begin{align*}
  & D(\rho^\top_{ABE},\gamma_{ABE}) \\
  & \qquad \leq \sum_{k_A,k_B} \frac{p_{k_A,k_B}}{1-p^{\bot}}
    D\left(\proj{k_A,k_B} \otimes \rho^{k_A,k_B}_E,\right. \\
  & \qquad \qquad \qquad \qquad \qquad \qquad \left.\proj{k_A,k_A} \otimes \rho^{k_A,k_B}_E \right)\\
  & \qquad = \sum_{k_A \neq k_B} \frac{p_{k_A,k_B}}{1-p^{\bot}} = \frac{1}{1-p^{\bot}}\Pr
  \left[ K_A \neq K_B \right].
\end{align*}
Putting the above together with \eqnref{eq:qkd.security.2}, we get
\begin{align*} & D(\rho_{ABE},\tilde{\rho}_{ABE}) \\
  & \qquad = (1-p^\bot)
  D(\rho^\top_{ABE},\tau_{AB} \otimes \rho^\bot_{E}) \\ & \qquad \leq \Pr
  \left[ K_A \neq K_B \right] + (1-p^\bot) D(\rho^\top_{AE},\tau_{A}
  \otimes \rho^\top_{E}). \qedhere \end{align*}
\end{proof}

\begin{proof}[Proof of \lemref{lem:robustness}]
  By construction, $\aK_\delta$ aborts with exactly the same
  probability as the real system. And because $\sigma^{\qkd}_E$
  simulates the real protocols, if we plug a converter $\pi_E$ in
  $\aK\sigma^{\qkd}_E$ which emulates the noisy channel $\aQ_q$ and
  blogs the output of the simulated authentic channel, then
  $\aK_\delta = \aK\sigma^{\qkd}_E\pi_E$. Also note that by
  construction we have
  $\aQ_q \| \aA' = \left(\aQ \| \aA\right) \pi_E$. Thus
  \begin{multline*} d\left( \pi_A^{\qkd}\pi_B^{\qkd}(\aQ_q \| \aA')
      ,\aK_\delta\right) \\ = d\left( \pi_A^{\qkd}\pi_B^{\qkd}\left(\aQ
        \| \aA\right) \pi_E , \aK\sigma^{\qkd}_E\pi_E\right). \end{multline*}

  Finally, because the converter $\pi_E$ on both the real and ideal
  systems can only decrease their distance (see
  \secref{sec:ac.systems}), the result follows.
\end{proof}



%

\end{document}